\documentclass[a4paper,11pt]{article}
\usepackage[pdftex]{graphicx,xcolor}
\usepackage{float}
\usepackage{amsmath,amssymb,bm,amsthm}
\usepackage{ascmac}
\usepackage{array}
\usepackage{hyperref}
\usepackage[hang,small,bf]{caption}
\usepackage[subrefformat=parens]{subcaption}
\usepackage{natbib}
\usepackage{enumitem}

\bibpunct[:]{(}{)}{,}{a}{}{,}
\hypersetup{
	colorlinks=true,
	citecolor=blue,
	linkcolor=red,
	urlcolor=orange,
}

\title{A New Proof for the Linear Filtering and Smoothing Equations, and its Application to Nonlinear Filtering}
\author{Masahiro Kurisaki\thanks{AIP Center, RIKEN\protect\\ email: \texttt{masahiro.kurisaki@riken.jp}}~\thanks{Japan Science and Technology Agency CREST \protect\\ The author was supported by Japan Science and Technology Agency CREST JP-MJCR2115, JSPS KAKENHI Grant Number JP24KJ0667, and RIKEN Special Postdoctoral Researcher Program.}}

\numberwithin{equation}{section}
\allowdisplaybreaks[3]
\newtheorem{theorem}{Theorem}[section]
\newtheorem{proposition}[theorem]{Proposition}
\newtheorem{lemma}[theorem]{Lemma}
\newtheorem{corollary}[theorem]{Corollary}

\theoremstyle{definition}

\theoremstyle{remark}

\begin{document}
\maketitle
\begin{abstract}
  In this paper, we propose a new asymptotic expansion approach for nonlinear filtering based on a small parameter in the system noise. This method expresses the filtering distribution as a power series in the noise level, where the coefficients can be computed by solving a system of ordinary differential equations. As a result, it addresses the trade-off between computational efficiency and accuracy inherent in existing methods such as Gaussian approximations and particle filters. In the course of our derivation, we also show that classical linear filtering and smoothing equations, namely Kalman-Bucy filter and Rauch-Tung-Striebel smoother, can be obtained in a unified and transparent manner from an explicit formula for the conditional distribution of the hidden path.
\end{abstract}


\section{Introduction}
Estimating hidden states from noisy observations is a fundamental problem in various fields, including signal processing, control theory, and finance. One classical framework for addressing this problem is filtering, which aims to estimate the hidden states of an unobserved stochastic process \(\{X_t\}_{t \geq 0}\) based on the observed data from a related stochastic process \(\{Y_t\}_{t \geq 0}\). Here, \(\{Y_t\}_{t \geq 0}\) provides partial information about \(\{X_t\}_{t \geq 0}\) through an observation model. 

The objective of the filtering problem is to compute the conditional expectation \(E[X_t|\mathcal{Y}_t]\) for every \(t \geq 0\), where \(\mathcal{Y}_t\) denotes the \(\sigma\)-field generated by \(\{Y_s\}_{0 \leq s \leq t}\). This conditional expectation serves as the \(L^2\)-optimal estimator of \(X_t\) as a function of the observation process \(\{Y_s\}_{0 \leq s \leq t}\). 

In addition to filtering, the smoothing problem considers the estimation of past states of \(\{X_t\}\) given future observations of \(\{Y_t\}\). Specifically, the goal of smoothing is to compute \(E[X_s|\mathcal{Y}_t]\) for \(0 \leq s \leq t\). Together, filtering and smoothing play a central role in a wide range of applications.

In this paper, we assume that the signal process \( X \) satisfies the stochastic differential equation
\begin{align*}
  dX_t = A(t, X_t) \, dt + B(t, X_t) \, dV_t,
\end{align*}
and that the observation process \( Y \) evolves according to
\begin{align*}
  dY_t = h(t, X_t) \, dt + \sigma(t) \, dW_t,
\end{align*}
where \( V \) and \( W \) are independent standard Wiener processes, and \( A \), \( B \), \( h \), and \( \sigma \) are appropriately chosen functions. In the linear case, the filtering and smoothing problems reduce to finite-dimensional equations, known respectively as the Kalman filter and the Rauch-Tung-Striebel smoother \citep{kalman1960, kalman1961, Rauch1965, Liptser2001,Liptser2001-2}.

On the other hand, in the general nonlinear case, the stochastic filtering problem is governed by the Kushner-Stratonovich equation \citep{stratonovich1960,kushner1964} or the equivalent Zakai equation \citep{Mortensen1966,Zakai1969,Duncan1970a,Duncan1970b}, which are generally known to be infinite-dimensional \citep{HAZEWINKEL1983331}. 

Since these equations are rarely tractable in practice due to the infinite-dimensionality, a wide variety of approximation methods have been developed, which may be broadly classified into three categories.

{\bf (i) Linearization or assumed-density methods} such as Extended Kalman filter \citep{picard1986,picard1991}, EnKF \citep{2003OcDyn..53..343E}, UKF \citep{882463}, Gaussian filter \citep{847749}, and projection filter \citep{Brigo1995,armstrong2019optimal} reduce the problem to a finite-dimensional system by locally linearizing the dynamics or projecting the solution onto a statistical manifold. While computationally efficient, these methods rely on strong structural assumptions and therefore offer limited theoretical guarantees regarding their accuracy.

{\bf (ii) PDE-based approaches}, such as finite difference \citep{Gyongy2003}, finite element \citep{Germani01021988} or spectral \citep{lototsky2011chaos} methods, attempt to solve the Zakai equation numerically. These offer rigorous approximations but suffer severely from the curse of dimensionality, making them impractical in high-dimensional systems.

{\bf (iii) Simulation-based methods}, including particle filters \citep{gordon1993, 4378823,6530707}, approximate the posterior distribution using sample-based methods and resampling. These are highly flexible and widely used, but require a large number of particles to ensure accuracy and often suffer from sample degeneracy.

These methods, despite their variety, ultimately face a fundamental trade-off: they either rely on strong approximations (such as linearization) to reduce computational complexity, or retain theoretical accuracy at the cost of prohibitively heavy computation.

In this paper, we propose a novel asymptotic expansion method for nonlinear filtering, designed to overcome the inherent trade-off between computational complexity and approximation accuracy present in existing approaches. Specifically, we introduce a small parameter $0<\epsilon<1$ and consider the model
\begin{align*}
  &dX_t^\epsilon = A(t, X_t^\epsilon) \, dt + \epsilon B(t, X_t^\epsilon) \, dV_t,\\
  &dY_t^\epsilon = h(t, X_t^\epsilon) \, dt + \sigma(t) \, dW_t.
\end{align*}
Here, although we assume the system noise to be small, this assumption is well justified in a wide range of practical scenarios. 

Under this setting, we expand the conditional expectation of the signal as a power series in $\epsilon$
\begin{align*}
  E[X_t^\epsilon|\mathcal{Y}_t^\epsilon]=m_t^{[0]}+m_t^{[1]}\epsilon+m_t^{[2]}\epsilon^2+\cdots,
\end{align*}
We demonstrate that the computation of each coefficient $m_t^{[i]}$ reduces to solving a system of (stochastic) ordinary differential equations that incorporate the observation process $Y$. This representation allows for significantly more efficient computation than PDE-based or simulation-based methods, while also offering improved accuracy and generality over linear approximation methods, as the expansion can be advanced to higher orders as needed.

In the course of deriving the recursive formulae for the expansion coefficients, we will observe that several classical results for linear systems—such as the Kalman filter and the Rauch–Tung–Striebel smoother, which are typically treated as separate theorems in the literature \citep{Liptser2001,Liptser2001-2}—naturally emerge as corollaries of our main theorem for the linear case.

This main theorem not only reproduces the existing equations but also establishes a significantly stronger result. In fact, while classical formulations typically provide information only about the conditional distribution of the state \( X_s \) at each time point, our theorem characterizes the full distribution of the process \( \{X_s\}_{0 \leq s \leq t} \) given \( \mathcal{Y}_t \). This enables access to quantities such as the conditional covariance \( \mathrm{Cov}(X_s, X_u \mid \mathcal{Y}_t) \), which plays a crucial role in the asymptotic expansion for the nonlinear case.

This paper is organized as follows. In Section~\ref{section-filtering-theory}, we begin with a review of the fundamentals of nonlinear filtering and propose a relaxation of standard assumptions. The Kallianpur--Striebel formula serves as a central tool in the developments that follow.

Before turning to the nonlinear case, we present our new results for linear filtering and smoothing in Section~\ref{section-linear-filtering}. The main theorem is proved in Section~\ref{section-main-theorem}, and the classical filtering and smoothing equations---including the Kalman--Bucy filter, the Rauch--Tung--Striebel smoother, and the fixed-point smoother---are derived as corollaries in Section~\ref{section-linear-equations}. The heuristic idea behind the proof is outlined in Section~\ref{section-heuristic-argument}.

These results are applied to the asymptotic expansion of nonlinear filtering in Section~\ref{section-application}. This section provides an overview of the algorithm without rigorous proofs; each step, however, will be justified in a subsequent paper.

We begin by showing that the nonlinear model with small system noise can be transformed into a polynomially perturbed linear model. The results for the linear system developed in the previous section are then employed to derive recursive formulae for the expansion coefficients. The validity of the proposed method is demonstrated through numerical experiments using a toy model and a cubic sensor model. The simulations suggest that higher-order expansions lead to a reduction in the integrated mean squared error of the filter.

Finally, in Section~\ref{section-discussion}, we summarise the advantages of the proposed algorithm and outline directions for future research beyond the scope of this paper.

\section{Filtering theory}\label{section-filtering-theory}
Let $(\Omega, \mathcal{F}, \{\mathcal{F}_t\}_{t \geq 0}, P)$ be a filtered probability space, and $\{X_t\}_{t \geq 0}$ be a $d_1$-dimensional, $\{\mathcal{F}_t\}$-adapted c\`{a}dl\`{a}g process. Let $\{H_t\}_{t \geq 0}$ be a $d_2$-dimensional progressively measurable process, and let $\{W_t\}_{t \geq 0}$ be a $d_2$-dimensional $\{\mathcal{F}_t\}$-Brownian motion independent of both $\{X_t\}$ and $\{H_t\}$. We consider the process $\{Y_t\}_{t \geq 0}$ given by
\begin{align}
  Y_t = \int_0^t H_s \, ds + W_t.
\end{align}
Let $\{\mathcal{X}_t\}_{t \geq 0}$ and $\{\mathcal{Y}_t\}_{t \geq 0}$ be the usual augmentations (with null sets) of the filtrations generated by $\{X_t\}$ and $\{Y_t\}$, respectively. We define the conditional distribution $\pi_t$ of $X_t$ given the information $\{\mathcal{Y}_t\}$. More precisely, there exists a $\mathcal{P}(\mathbb{R}^{d_1})$-valued, $\{\mathcal{Y}_t\}$-adapted process $\{\pi_t\}_{t \geq 0}$ such that, for any bounded measurable function $f : \mathbb{R}^{d_1} \to \mathbb{R}$,
\begin{align}
  \label{eq-pi_t}
  \pi_t(f) = E[f(X_t) | \mathcal{Y}_t] \quad \text{a.s.},
\end{align}
where $\mathcal{P}(\mathbb{R}^{d_1})$ denotes the space of all probability measures on $\mathbb{R}^{d_1}$ \citep[Section 2]{bain2009fundamentals}. Equations (\ref{eq-pi_t}) is still valid for measurable functions $f$ such that $E[|f(X_t)|]<\infty$ (see Remark 2.27 in \citet{bain2009fundamentals}).

Assume that for all \( t \geq 0 \),
\begin{align}
  \label{eq-locally-boundedness-of-H}
  \int_0^t |H_s|^2 \, ds < \infty \quad \text{a.s.},
\end{align}
and define the process \( \{Z_t\}_{t \geq 0} \) by
\begin{align}
  \label{def-Zt-general}
  \begin{split}
    Z_t &= \exp\left( \int_0^t H_s^\top \, dW_s + \frac{1}{2} \int_0^t |H_s|^2 \, ds \right) \\
         &= \exp\left( \int_0^t H_s^\top \, dY_s - \frac{1}{2} \int_0^t |H_s|^2 \, ds \right).
  \end{split}
\end{align}
We now introduce a new probability measure \( Q \) on $\mathcal{F}_T$ for an arbitrary fixed \( T > 0 \), defined by
\begin{align}
  \label{eq-def-Q}
  Q(A) = E[1_A Z_T^{-1}] \quad (A \in \mathcal{F}_T).
\end{align}
To ensure that \( Q \) is well-defined, we need to verify that
\begin{align}
  \label{eq-EZ=1}
  E[Z_T^{-1}] = 1.
\end{align}
In \citet{bain2009fundamentals}, additional conditions are imposed to verify this condition:
\begin{align}
  \label{eq-stronger-condition-1}
  E\left[ \int_0^T |H_s|^2 \, ds \right] < \infty
\end{align}
and
\begin{align}
  \label{eq-stronger-condition-2}
  E\left[ \int_0^T Z_s |H_s|^2 \, ds \right] < \infty.
\end{align}

However, we show below that these conditions can be removed and that the local boundedness condition (\ref{eq-locally-boundedness-of-H}) alone is sufficient to ensure (\ref{eq-EZ=1}). To this end, we first establish the following lemma.

In what follows, for each \( t \geq 0 \), let \( S_t \) denote the space of measurable functions \( g: [0, t] \to \mathbb{R}^{d_2} \) such that \(\displaystyle \int_0^t |g(s)|^2 \, ds < \infty \). Let \( \mathcal{S}_t \) be the \(\sigma\)-algebra generated by the projections \( \varpi_s: S_t \to \mathbb{R}^{d_2} \) for \( 0 \leq s \leq t \), where \( \varpi_s(g) = g(s) \) for \( g \in S_t \). Then the following lemmas hold.

\begin{lemma}
  Let $\{\xi_{s}\}_{0\leq s\leq t}$ be an $\mathbb{R}^{d_2}$-valued progressively measurable process satisfying 
  \begin{align*}
    \int_0^t |\xi_s|^2 ds<\infty,  
  \end{align*}
  and $\xi.$ be its path. Then the $\sigma$-algebra $\mathcal{F}_t^\xi=\sigma(\xi_{s};0\leq s\leq t)$ can be represented as
  \begin{align*}
    \mathcal{F}_t^\xi=\{\xi_\cdot^{-1}(A);A \in \mathcal{S}_t\}.
  \end{align*}
  In particular, $\xi.$ is an $S_t$-valued random variable.
\end{lemma}
\begin{proof}
  Write $\mathcal{G}_t^\xi=\{\xi.^{-1}(A);A \in \mathcal{S}_t\}$. Then it is clear that $\mathcal{F}_t^\xi \subset \mathcal{G}_t$ since $\xi_s^{-1}(B)=\xi_\cdot^{-1}(\varpi_{s}^{-1}(B))\in \mathcal{G}_t$ for any $B \in \mathcal{B}(\mathbb{R}^{d_2})$. To prove the converse inclusion, let 
  \begin{align*}
    \mathcal{A}=\{A \subset S_t; \xi_{\cdot}^{-1}(A) \in \mathcal{F}_t^\xi\}.
  \end{align*}
  Then $\mathcal{A}$ is a $\sigma$-algebra, and $\pi_{s}^{-1}(B) \in \mathcal{A}$ for any $0\leq s\leq t$ and $B\in \mathcal{B}(\mathbb{R}^{d_2})$ since $\xi_{\cdot}^{-1}(\pi_{s}^{-1}(B))=\xi_s^{-1}(B) \in \mathcal{F}_t^\xi$. Thus it follows $\mathcal{S}_t\subset \mathcal{A}$ and the desired result.
\end{proof}

\begin{lemma}\label{lemma-stochastic-integral-measurable-function}
  Let $\{W_t\}_{t\geq 0}$ be a $d_2$-dimensional Brownian motion, and $\{H_t\}_{t\geq 0}$ be a $d_2$-dimensional progressively measurable process such that for all $t\geq 0$
  \begin{align*}
    \int_0^t |H_s|^2 ds <\infty~~\mathrm{a.s.},
  \end{align*}
  and independent of $\{W_t\}$.
  
  Then there exists a measurable function $\Phi_t:(S_t\times S_t,\mathcal{S}_t\otimes \mathcal{S}_t)\to \mathbb{R}$ such that
  \begin{align}
    \label{eq2-1}\Phi_t(H.,W.)=\int_0^t H_s dW_s~~\textrm{a.s.},
  \end{align}
  and for $\mu_H$-a.e. $h\in S_t$
  \begin{align}
    \label{eq2-2}\Phi_t(h,W.)=\int_0^t h(s) dW_s~~\textrm{a.s.}
  \end{align} 
  Here, $H.$ and $W.$ denote the paths of $\{H_s\}_{0\leq s\leq t}$ and $\{W_s\}_{0\leq s\leq t}$, and $\mu_H$ is the distribution of $H.$ on $S_t$.
\end{lemma}
\begin{proof}
  (i) First assume that $\{H_s\}_{0\leq s \leq t}$ is continuous. Then we have
  \begin{align*}
    \sum_{i=1}^n H_{t/{2^{n-1}}}^\top(W_{t/{2^n}}-W_{t/{2^{n-1}}})\xrightarrow{P} \int_0^t H_s dW_s.
  \end{align*}
  Hence we can take a subsequence $\{n(k)\}_{k \in \mathbb{N}}$ such that
  \begin{align}
    \label{eq2-3}\sum_{i=1}^{n(k)} H_{t/{2^{n(k)-1}}}^\top(W_{t/{2^{n(k)}}}-W_{t/{2^{{n(k)}-1}}})\xrightarrow{\textrm{a.s.}} \int_0^t H_s dW_s.
  \end{align}
  Now we introduce a sequence of measurable functions \( \Phi_t^k : S_t \times S_t \to \mathbb{R} \) for \( k \in \mathbb{N} \), defined by
\begin{align*}
    \Phi_t^k(h, w) = \sum_{i=1}^{n(k)} h\left(\frac{t}{2^{n(k)-1}}\right)^\top \left(w\left(\frac{t}{2^{n(k)}}\right) - w\left(\frac{t}{2^{n(k)-1}}\right)\right).
\end{align*}
We then define \( \Phi_t : S_t \times S_t \to \mathbb{R} \) by
\begin{align*}
    \Phi_t(h, w) = \begin{cases}
      \displaystyle \lim_{k \to \infty} \Phi_t^k(h, w) & \text{if } \Phi_t^k(h, w) \text{ converges,} \\
      0 & \text{otherwise}.
    \end{cases}
\end{align*}
Since each \( \Phi_t^k \) is measurable, \( \Phi_t \) is also measurable, and (\ref{eq2-1}) follows directly from the construction of \( \Phi_t \).

  Let us verify (\ref{eq2-2}). If we set
  \begin{align*}
    A=\left\{ (h,w)\in S_t\times S_t; \Phi_t^k(h,w)~\textrm{converges~to~}\Phi_t(h,w)  \right\},
  \end{align*}
  then (\ref{eq2-3}) means
  \begin{align*}
    \int_{S_t\times S_t}1_{A^c}(h,w)d\mu_{(H,W)}(h,w)=0,
  \end{align*}
  where $\mu_{(H,W)}$ is the distribution of $(H.,W.)$. Due to the independence of $H.$ and $W.$, this can be rewritten as
  \begin{align*}
    \int_{S_t}\int_{S_t}1_{A^c}(h,w)d\mu_{W}(w)d\mu_{H}(h)=0,
  \end{align*}
  and thus we have for $\mu_H$-a.e. $h\in S_t$,
  \begin{align*}
    \int_{S_t}1_{A^c}(h,w)d\mu_{W}(w)=0.
  \end{align*}
  From this, it follows that for $\mu_H$-a.s. $h\in S_t$, the sequence of random variables
  \begin{align}
    \label{eq2-4}\Phi_t^k(h,W.)=\sum_{i=1}^{n(k)} h({t/{2^{n(k)-1}}})^\top(W_{t/{2^{n(k)}}}-W_{t/{2^{{n(k)}-1}}})
  \end{align}
  almost surely converges to $\Phi_t(h,W)$. On the other hand, since \( \{H_s\}_{0 \leq s \leq t} \) is continuous, we may assume that \( h \) is continuous. For such \( h \), the right-hand side of (\ref{eq2-4}) converges in probability to \( \int_0^t h(s) \, dW_s \), thereby yielding (\ref{eq2-2}).
  \\
  (ii) Next, let us assume $\{H_s\}_{0\leq s\leq t}$ is bounded but not necessarily continuous. For $m \in \mathbb{N}$ and $h \in S_t$, define $h^m$ by
  \begin{align*}
    h^m(s)=m\int_{\left( s-\frac{1}{m} \right)\vee 0}^s h(u) du.
  \end{align*}
  Then according to the Lebesgue differentiation theorem, it holds for all $\omega \in \Omega$
  \begin{align*}
    \lim_{m \to \infty}H_s^m(\omega)=H_s(\omega)~~\textrm{a.e.}~s\in[0,t].
  \end{align*}
  Hence, we have $H_s^m(\omega)\to H_s(\omega)$ for almost every $(\omega,s)$ by Fubini's theorem, and it follows from the bounded convergence theorem that
  \begin{align*}
    E\left[ \left(\int_0^t H_s^m dW_s - \int_0^t H_s dW_s \right)^2\right]\leq E\left[ \int_0^t (H_s^m-H_s)^2ds \right] \to 0~(m \to \infty).
  \end{align*}
  Thus, we can take a subsequence $\{m(k)\}_{k \in \mathbb{N}}$ such that
  \begin{align}
    \label{eq2-5}\int_0^t H_s^{m(k)} dW_s\xrightarrow{\textrm{a.s.}} \int_0^t H_s dW_s~~(m\to \infty).
  \end{align}
  On the other hand, since $\{H_s^m\}$ is a continuous progressively measurable process, we can take measurable functions $\Phi_t^m:S_t\times S_t \to \mathbb{R}~(m\in \mathbb{N})$ such that
  \begin{align}
    \label{eq2-6}\Phi_t^m(H.^m,W.)=\int_0^t H_s^m dW_s~~\textrm{a.s.},
  \end{align}
  and for $\mu_{H}$-a.s. $h \in S_t$
  \begin{align}
    \label{eq2-7} \Phi_t^m(h^m,W.)=\int_0^t h^m(s)dW_s.
  \end{align}
  Now we define $\Phi_t:S_t\times S_t \to \mathbb{R}$ by
  \begin{align*}
    \Phi_t(h,w)=\begin{cases}
      \displaystyle \lim_{k\to \infty}\Phi_t^{m(k)}(h^{m(k)},w)&\textrm{when }\Phi_t^{m(k)}(h^{m(k)},w)\textrm{ converges,}\\
      0&\textrm{otherwise}.
    \end{cases}
  \end{align*}
  Then (\ref{eq2-1}) follows directly from (\ref{eq2-5}) and (\ref{eq2-6}). To prove (\ref{eq2-2}), we proceed similarly to part (i), showing that for \( \mu_H \)-a.s. \( h \in S_t \), the sequence \( \Phi_t^{m(k)}(h^{m(k)}, W.) \) converges almost surely to \( \Phi_t(h, W) \). By (\ref{eq2-7}), this limit must be \( \int_0^t h(s) \, dW_s \), thereby establishing (\ref{eq2-2}).\\
(iii) Finally, consider the general case where \( \{H_s\} \) is not necessarily continuous or bounded. In this case, we define \( H_s^m = H_s 1_{\{|H_s| \leq m\}} \) for \( m \in \mathbb{N} \). Then \( \{H_s^m\} \) is bounded, and since \( |H_s^m| \leq |H_s| \), we have
\begin{align*}
    \int_0^t H_s^m \, dW_s \xrightarrow{P} \int_0^t H_s \, dW_s,
\end{align*}
so we obtain the desired result by applying the same argument as in part (ii).
\end{proof}

We now apply this lemma to \( Z_t^{-1} \) as defined in (\ref{def-Zt-general}).
\begin{proposition}\label{prop-Z-martingale}
  Under (\ref{eq-locally-boundedness-of-H}), $\{Z_t^{-1}\}_{t\geq 0}$ is a martingale.
\end{proposition}
\begin{proof}
  First, we show that \( E[Z_t^{-1}] = 1 \) for all \( t \geq 0 \). Let \( (S_t, \mathcal{S}_t) \) and \( H. \) be as in Lemma \ref{lemma-stochastic-integral-measurable-function}, and let \( \Phi_t \) satisfy (\ref{eq2-1}) and (\ref{eq2-2}). Then \( Z_t^{-1} \) can be expressed as
\begin{align*}
    Z_t^{-1} = \exp\left( -\Phi_t(H., W.) - \frac{1}{2} \int_0^t |H_s|^2 \, ds \right) \quad \text{a.s.},
\end{align*}
and we have \( E[Z_t^{-1}] \leq E[Z_0^{-1}] = 1 \) since \( \{Z_t^{-1}\}_{t \geq 0} \) is a supermartingale.
Therefore, we have
  \begin{align*}
    E[Z_t^{-1}|\mathcal{H}_t]=\left.E\left[ \exp\left( -\Phi_t(h,W.)-\frac{1}{2}\int_0^t |h(s)|^2ds \right) \right]\right|_{h=H.},
  \end{align*}
  where $\mathcal{H}_t=\sigma(H_s;0\leq s\leq t)$. Due to (\ref{eq2-2}), the right-hand side almost surely equals
  \begin{align*}
    \left.E\left[ \exp\left( -\int_0^t h(s)dW_s-\frac{1}{2}\int_0^t |h(s)|^2ds \right) \right]\right|_{h=H.}=1,
  \end{align*}
  and thus we obtain
  \begin{align}
    \label{eq-Zt=1}E[Z_t^{-1}]=E[E[Z_t^{-1}|\mathcal{H}_t]]=1.
  \end{align}
  Now the martingale property of $\{Z_t^{-1}\}$ easily follows. In fact, noting that $\{Z_t^{-1}\}$ is a supermartingale, we have
  \begin{align*}
    E[Z_t^{-1}-Z_s^{-1}|\mathcal{F}_s]\leq 0.
  \end{align*}
  On the other hand, it follows from (\ref{eq-Zt=1}) that
  \begin{align*}
    E[E[Z_t^{-1}-Z_s^{-1}|\mathcal{F}_s]]=E[Z_t^{-1}]-E[Z_s^{-1}]=1-1=0.
  \end{align*} 
  Therefore we have $E[Z_t^{-1}-Z_s^{-1}|\mathcal{F}_s]=0$.
\end{proof}

By Proposition \ref{prop-Z-martingale}, we can define the probability measure \( Q \) using (\ref{eq-def-Q}) under the weaker condition (\ref{eq-locally-boundedness-of-H}), instead of the stronger conditions (\ref{eq-stronger-condition-1}) and (\ref{eq-stronger-condition-2}). Moreover, the same results in Section 3 of \citet{bain2009fundamentals} still hold.

In particular, the following two results are of central importance in this paper.
\begin{proposition}\label{prop-distribution-in-Q}
  Under the new measure $Q$, the process $\{Y_t\}_{0\leq t\leq T}$ is a Brownian motion independent of $\{X_t\}_{0\leq t\leq T}$, and the law of $\{X_t\}_{0\leq t\leq T}$ is the same as its law under $P$.

  More generally, if \( U \) is a random variable independent of \( \{W_t\}_{0 \leq t \leq T} \), then under \( Q \), \( U \) is independent of \( \{Y_t\}_{0 \leq t \leq T} \) and has the same law as under \( P \).
\end{proposition}
\begin{proof}
  Although a proof of this proposition is provided in Proposition 3.13 of \citet{bain2009fundamentals}, we present a brief proof here, utilizing Lemma \ref{lemma-stochastic-integral-measurable-function}.

Let $\mathcal{U}$ be a $\sigma$-field generated by $U$, and let $\mathcal{H}_T$ denote the $\sigma$-field generated by $\{H_s\}_{0 \leq s \leq T}$. Consider $\xi_1, \dots, \xi_k \in \mathbb{R}^{d_2}$ and $t_1, \dots, t_k \in [0, T]$. Using the same arguments as in Proposition \ref{prop-Z-martingale} and applying the Girsanov theorem, along with the independence of $W$, $X$, $U$, and $H$, we deduce:
\begin{align*}
    &E_Q\left[\exp\left( \sum_{i=1}^k \xi_i^\top Y_{t_i} \right) \middle| \mathcal{X}_T \vee \mathcal{U} \right] \\
    =& E\left[\exp\left( \sum_{i=1}^k \xi_i^\top Y_{t_i} \right) Z_T^{-1} \middle| \mathcal{X}_T \vee \mathcal{U} \right] \\
    =& E\left[\exp\left( \sum_{i=1}^k \xi_i^\top \left( \int_0^{t_i} H_s \, ds + W_{t_i} \right) \right) \right. \\
    &\quad \left. \times \exp\left( -\int_0^T H_s^\top \, dW_s - \frac{1}{2} \int_0^T |H_s|^2 \, ds \right) \middle| \mathcal{X}_T \vee \mathcal{U} \right] \\
    =& E\left[E\left[\exp\left( \sum_{i=1}^k \xi_i^\top \left( \int_0^{t_i} H_s \, ds + W_{t_i} \right) \right) \right. \right. \\
    &\quad \left. \left. \times \exp\left( -\int_0^T H_s^\top \, dW_s - \frac{1}{2} \int_0^T |H_s|^2 \, ds \right) \middle| \mathcal{X}_T \vee \mathcal{U} \vee \mathcal{H}_T \right] \middle| \mathcal{X}_T \vee \mathcal{U} \right] \\
    =& E\left[E\left[\exp\left( \sum_{i=1}^k \xi_i^\top \left( \int_0^{t_i} h(s) \, ds + W_{t_i} \right) \right) \right. \right. \\
    &\quad \left. \left. \times \exp\left( -\int_0^T h(s)^\top \, dW_s - \frac{1}{2} \int_0^T |h(s)|^2 \, ds \right) \right]\bigg|_{h=H.} \middle| \mathcal{X}_T \vee \mathcal{U} \right] \\
    =& E\left[\exp\left( \sum_{i=1}^k \xi_i^\top W_{t_i} \right) \right] = E\left[\exp\left( \sum_{i=1}^k \xi_i^\top W_{t_i} \right) \middle| \mathcal{X}_T \vee \mathcal{U} \right].
\end{align*}
Here, $h \in S_T$, and $H.$ denotes the path of $\{H_s\}_{0 \leq s \leq T}$, and we used the independence of $\mathcal{X}_T\vee \mathcal{U}$ and $W$ in the last equality.

Now, let us fix $A \in \mathcal{X}_T \vee \mathcal{U}$ arbitrarily, and define measures $\mu_1$ and $\mu_2$ on $\mathbb{R}^{d_2 \times n}$ by
\begin{align*}
    \mu_1(B) &= E_Q[1_{B}(Y_{t_1}, \dots, Y_{t_k}) 1_A],~~~\mu_2(B) = E[1_{B}(W_{t_1}, \dots, W_{t_k}) 1_A],
\end{align*}
for any $B \in \mathcal{B}(\mathbb{R}^{d_2 \times n})$. By the earlier argument, the characteristic functions of $\mu_1$ and $\mu_2$ are identical. Hence, for $t_1,\cdots.t_k,s_1,\cdots,s_l \in [0,T]$, the law of $(Y_{t_1}, \dots, Y_{t_k}, X_{s_1}, \dots, X_{s_l}, U)$ under $Q$ is the same as the law of $(Y_{t_1}, \dots, W_{t_k}, W_{s_1}, \dots, W_{s_l}, U)$ under $P$ . This establishes the desired conclusion.
\end{proof}

\begin{proposition}\label{prop-Kallianpur-Striebel}(Kallianpur-Striebel formula)\\
  For any $0\leq t\leq T$ and a random variable $U$ such that $E[|U|]<\infty$, it holds that
  \begin{align}
    \label{eq-Kallianpur-Striebel}E[U|\mathcal{Y}_t]=\frac{\displaystyle E_Q\left[ UZ_t\middle|\mathcal{Y}_t \right]}{\displaystyle E_Q\left[ Z_t\middle|\mathcal{Y}_t \right]}~~\textrm{a.s.},
  \end{align}
  where $E_Q$ denotes the expectation under the measure $Q$.
\end{proposition}
\begin{proof}
  This proposition is a generalized version of Proposition 3.16 in \citet{bain2009fundamentals}, but the proof is exactly the same as in the original. 
  
  First note that 
  \begin{align*}
    E_Q[|UZ_t|]=E[|U|Z_tZ_t^{-1}]=E[|U|]<\infty.
  \end{align*}
  Furthermore, by
  \begin{align*}
    P(Z_t=0)=E[1_{\{Z_t=0\}}]=E_Q[1_{\{Z_t=0\}}Z_t]=0,
  \end{align*}
  and the equivalence of $P$ and $Q$, $Z_t>0$ holds $P$ and $Q$ a.s. Hence, it is sufficient to show that
  \begin{align}
    \label{eq2-8}E[U|\mathcal{Y}_t]E_Q\left[ Z_t\middle|\mathcal{Y}_t \right]=E_Q\left[ UZ_t\middle|\mathcal{Y}_t \right].
  \end{align}
  To prove this, let $A \in \mathcal{Y}_t$. Then we have 
  \begin{align*}
    E_Q\bigl[ E[U|\mathcal{Y}_t]E_Q\left[ Z_t\middle|\mathcal{Y}_t \right]1_A \bigr]
    &=E_Q\bigl[ E_Q\left[E[U|\mathcal{Y}_t] Z_t 1_A\middle|\mathcal{Y}_t \right] \bigr]\\
    &=E_Q\bigl[ E[U|\mathcal{Y}_t] Z_t 1_A \bigr]\\
    &=E\bigl[ E[U1_A|\mathcal{Y}_t]  \bigr]\\
    &=E[U1_A]=E_Q[U1_AZ_t],
  \end{align*}
  which implies (\ref{eq2-8}).
\end{proof}

In the case where \( H \) depends on \( Y \) and is therefore not independent of \( W \), a convenient sufficient condition for \( Z_t^{-1} \) to be a martingale is provided in \citet{kunita2011nonlinear}.

\section{Linear Filtering and Smoothing}\label{section-linear-filtering}
In this section, we establish a stronger result for the linear case as a foundation for the subsequent analysis of the nonlinear case. Here, we assume that the processes \( \{X_t\} \) and \( \{Y_t\} \) are solutions of the following linear equations:
\begin{align}
  \label{eq-linear-system-X} &X_t = X_0+\int_0^t a(s) X_s \, ds + \int_0^t b(s) \, dV_s, \\
  \label{eq-linear-system-Y} &Y_t = \int_0^t c(s) X_s \, ds + \int_0^t \sigma(s) \, dW_s,
\end{align}
where \( a \), \( b \), \( c \), and \( \sigma \) are measurable functions taking values in \( M_{d_1}(\mathbb{R}) \), \( M_{d_1,m_1}(\mathbb{R}) \), \( M_{d_2,d_1}(\mathbb{R}) \), and \( M_{d_2,m_2}(\mathbb{R}) \), respectively. Here, \( \{V_t\} \) and \( \{W_t\} \) are independent \( m_1 \)- and \( m_2 \)-dimensional \( \{\mathcal{F}_t\} \)-Brownian motions. We also assume that \( X_0 \) is normally distributed and independent of \( \{V_t\} \) and \( \{W_t\} \), so that \( \{(X_t, Y_t)\} \) forms a \( d_1 + d_2 \)-dimensional Gaussian process.

Additionally, we assume that for every \( t \geq 0 \),
\begin{align}
  \label{eq-assumption-a} &\int_0^t |a(s)| \, ds < \infty, \\
  \label{eq-assumption-b} &\int_0^t |b(s)|^2 \, ds < \infty, \\
  \label{eq-assumption-c} &\int_0^t |c(s)|^2 \, ds < \infty, \\
  \label{eq-assumption-sigma-1} &\int_0^t |\sigma(s)|^2 \, ds < \infty,
\end{align}
and that there exists a constant \( C > 0 \) such that for every \( t \geq 0 \),
\begin{align}
  \label{eq-assumption-sigma-2} \lambda_{\min}(\sigma(t) \sigma(t)^\top) > C,
\end{align}
where \( \lambda_{\min}(\sigma(t) \sigma(t)^\top) \) denotes the smallest eigenvalue of \( \sigma(t) \sigma(t)^\top \).
Throughout this paper, $|A|$ denotes the Frobenius norm for a matrix $A$.

Let us define 
\begin{align*}
  \overline{W}_t=\int_0^t (\sigma(s)\sigma(s)^\top)^{-\frac{1}{2}}\sigma(s)dW_s.
\end{align*}
Then $\left\{ \overline{W}_t \right\}_{t\geq 0}$ is a Brownian motion since
\begin{align*}
  \langle \overline{W}\rangle_t&=\int_0^t (\sigma(s)\sigma(s)^\top)^{-\frac{1}{2}}\sigma(s)\sigma(s)^\top (\sigma(s)\sigma(s)^\top)^{-\frac{1}{2}}ds\\
  &=\int_0^t (\sigma(s)\sigma(s)^\top)^{-\frac{1}{2}}(\sigma(s)\sigma(s)^\top)^{\frac{1}{2}}(\sigma(s)\sigma(s)^\top)^{\frac{1}{2}} (\sigma(s)\sigma(s)^\top)^{-\frac{1}{2}}ds\\
  &=I_{d_2}t,
\end{align*}
where $I_{d_2}$ is the $d_2$-dimensional identity matrix. Furthermore, let us write
\begin{align*}
  \overline{c}(t)=(\sigma(t)\sigma(t)^\top)^{-\frac{1}{2}}c(t),~~\overline{Y}_t=\int_0^t (\sigma(s)\sigma(s)^\top)^{-\frac{1}{2}}dY_s,
\end{align*}
so that 
\begin{align*}
  \overline{Y}_t=\int_0^t\overline{c}(s)X_sds+\overline{W}_t.
\end{align*}
Then by (\ref{eq-assumption-c}) and (\ref{eq-assumption-sigma-2}), we have for $t\geq 0$
\begin{align*}
  \int_0^t |\overline{c}(s)X_s|^2ds
  \leq \int_0^t |\overline{c}(s)|^2|X_s|^2ds \leq \frac{1}{\sqrt{C}}\int_0^t |c(s)|^2|X_s|^2ds<\infty,
\end{align*}
where $C$ is given by (\ref{eq-assumption-sigma-2}). Furthermore, the filtration generated by $\{Y_t\}$ and $\{\overline{Y}_t\}$ are the same one because
\begin{align*}
  Y_t=\int_0^t (\sigma(s)\sigma(s)^\top)^\frac{1}{2}d\overline{Y}_s.
\end{align*}
Therefore, we can substitute \( H_t \) with \( \overline{c}(t) X_t \) and \( Y_t \) with \( \overline{Y}_t \) in Proposition \ref{prop-Kallianpur-Striebel} to obtain the following result.
\begin{proposition}\label{prop-Kallianpur-Striebel-linear}
  Let $T>0$ and define a probability measure $Q$ by
  \begin{align*}
    Q(A)=E\left[ 1_AZ_T^{-1} \right]&
  \end{align*}
  for $A \in \mathcal{F}$, where
  \begin{align}
    \label{eq-def-Z}\begin{split}
      Z_t=&\exp\left( \int_0^tX_s^\top c(s)^\top(\sigma(s)\sigma(s)^\top)^{-1}\sigma(s) dW_s\right.\\
      &\left.+\frac{1}{2}\int_0^tX_s^\top c(s)^\top(\sigma(s)\sigma(s)^\top)^{-1} c(s)X_s ds\right)\\
      =&\exp\left( \int_0^tX_s^\top c(s)^\top(\sigma(s)\sigma(s)^\top)^{-1} dY_s\right.\\
    &\left.-\frac{1}{2}\int_0^tX_s^\top c(s)^\top(\sigma(s)\sigma(s)^\top)^{-1} c(s)X_s ds\right).
    \end{split}    
  \end{align}
  Then, under the new measure $Q$, the process $\{\overline{Y}_t\}_{0\leq t\leq T}$ is a Brownian motion independent of $\{X_t\}_{0\leq t\leq T}$, and the law of $\{X_t\}_{0\leq t\leq T}$ is the same as its law under $P$.\\
  Furthermore, for $0\leq t \leq T$ and a random variable $U$ satisfying $E[|U|]<\infty$, it holds that $E[Z_t|\mathcal{Y}_t]>0~\mathrm{a.s.}$ and 
  \begin{align}
    \label{Kallianpur-Striebel-linear}&E[U|\mathcal{Y}_t]=\frac{\displaystyle E_Q\left[ UZ_t\middle|\mathcal{Y}_t \right]}{\displaystyle E_Q\left[ Z_t\middle|\mathcal{Y}_t \right]}.
  \end{align}  
\end{proposition}

To examine the properties of the conditional distribution, we fix \( t \in [0, T] \) and define the random measure \( \tilde{P}_t \) on \( (\Omega, \mathcal{X}_t) \) by
\begin{align}
  \label{def-tilde-P}\tilde{P}_t(A) = \tilde{P}_t(A | \mathcal{Y}_t) = \frac{\displaystyle E_Q\left[ 1_A Z_t \middle| \mathcal{Y}_t \right]}{\displaystyle E_Q\left[ Z_t \middle| \mathcal{Y}_t \right]},
\end{align}
and let \( \tilde{E}_t \) denote the expectation with respect to \( \tilde{P}_t \).

Note that this measure is well-defined. Specifically, if \( C \) denotes the space of \( \mathbb{R}^{d_1 + d_2} \)-valued continuous functions on \( [0, t] \), equipped with the Borel \(\sigma\)-algebra associated with the supremum norm, then any \( \mathcal{X}_t \vee \mathcal{Y}_t \)-measurable random variable \( U \) can be represented as
\begin{align*}
  U = \Phi(X., Y.) \quad \text{a.s.},
\end{align*}
where \( \Phi : C \to \mathbb{R} \) is a measurable function, and \( X. \) and \( Y. \) denote the paths of \( \{X_s\}_{0 \leq s \leq t} \) and \( \{Y_s\}_{0 \leq s \leq t} \), respectively. It is well known that \( C \) is a complete separable metric space, and therefore the regular conditional distribution of \( (X., Y.) \) given \( Y. \) is well-defined. 

Let us denote this distribution by $\mu_{(X,Y)}(\,\cdot\,|Y.)$. Then for $A=X.^{-1}(B)\in \mathcal{X}_t$ ($B \in \mathcal{B}(C)$), we can define $\tilde{P}_t$ by 
\begin{align*}
  \tilde{P}_t(X.^{-1}(B))=\mu_{(X,Y)}(B\times C|Y.).
\end{align*}
This $\tilde{P}$ is well-defined. In fact, if $A=X.^{-1}(B_1)=X.^{-1}(B_2)$, then we have $X.^{-1}(B_1\backslash B_2)=X.^{-1}(B_2\backslash B_1)=\emptyset$ and thus 
\begin{align*}
  \mu_{(X,Y)}(B_1\times C|Y.)=\mu_{(X,Y)}((B_1\cap B_2)\times C|Y.)=\mu_{(X,Y)}(B_2\times C|Y.).
\end{align*}
Furthermore, choose $A_1=X.^{-1}(B_1),A_2=X.^{-1}(B_2),\cdots \in \mathcal{X}_t$ where $B_1$, $B_2$,$\cdots \in \mathcal{B}(C)$, and assume $A_i\cap A_j=\emptyset~(i\neq j)$. Also, set
\begin{align*}
  B_1'=B_1,~~B_2'=B_2\backslash B_1,~~\cdots,B_i'=B_i\backslash \bigcup_{j=1}^{i-1} B_j.
\end{align*}
Then we have $A_i=X.^{-1}(B_i')$ for every $i \in \mathbb{N}$. In fact, if there exists $\omega \in X.^{-1}(B_i)\backslash X.^{-1}(B_i')=X.^{-1}(B_i \backslash B_i')$ for $i\geq 2$, then $\displaystyle X.(\omega) \in B_i\cap \bigcup_{j=1}^{i-1} B_j$ and thus $\displaystyle \omega \in A_i \cap \bigcup_{j=1}^{i-1}A_j$. This contradicts $A_i \cap A_j=\emptyset$.

Hence, we have 
\begin{align*}
  P\left( \bigcup_{i=1}^\infty A_i \right)=&P\left( \bigcup_{i=1}^\infty X.^{-1}(B_i') \right)
  =\mu_{(X,Y)}\left( \bigcup_{i=1}^\infty B_i'\times C \middle|Y.\right)\\
  =&\sum_{i=1}^\infty \mu_{(X,Y)}\left( B_i'\times C \right)=\sum_{i=1}^\infty P(A_i).
\end{align*}

\subsection{The main theorem}\label{section-main-theorem}
In this section, we prove our main theorem for the linear case, which determines the distribution of the process \( \{X_s\}_{0 \leq s \leq t} \) under the measure \( \tilde{P}_t \). To describe the theorem, for each fixed \( t \geq 0 \), we introduce an \( M_{d_1}(\mathbb{R}) \)-valued differentiable function \( \phi(s; t) \) as the negative semidefinite symmetric solution of the matrix Riccati equation
\begin{align}
  \label{eq-def-phi}
  \begin{split}
    \frac{d\phi}{ds}(s; t) &= -\phi(s; t) b(s) b(s)^\top \phi(s; t) - a(s)^\top \phi(s; t) - \phi(s; t) a(s) \\
    &\quad + c(s)^\top (\sigma(s) \sigma(s)^\top)^{-1} c(s)
  \end{split}
\end{align}
with the boundary condition \( \phi(t; t) = 0 \). This solution is well-defined. Specifically, the Riccati equation
\begin{align*}
  \frac{d\tilde{\phi}}{ds}(s; t) &= -\tilde{\phi}(s; t) b(t - s) b(t - s)^\top \tilde{\phi}(s; t) \\
  &\quad - a(t - s)^\top \tilde{\phi}(s; t) - \tilde{\phi}(s; t) a(t - s) \\
  &\quad + c(t - s)^\top (\sigma(t - s) \sigma(t - s)^\top)^{-1} c(t - s)
\end{align*}
with \( \tilde{\phi}(0; t) = 0 \) has a positive-semidefinite solution \( \tilde{\phi}(s; t) \) for \( 0 \leq s \leq t \), according to Theorems 2.1 and 2.2 in \citet{potter1965matrix}. It then follows immediately that \( \phi(s; t) = -\tilde{\phi}(t - s; t) \) satisfies (\ref{eq-def-phi}), and uniqueness can be established using Gronwall's lemma.

Also, let \( \{\tilde{V}_s\} \) be a \( d_1 \)-dimensional Brownian motion on \( (\Omega, \mathcal{F}, \{\mathcal{F}_t\}, P) \) that is independent of \( \{Y_s\} \) and \( X_0 \). Let \( \{\xi_{s; t}\}_{0 \leq s \leq t} \) be the solution to the stochastic differential equation
\begin{align}
  \label{eq-def-xi}
  d_s\xi_{s; t} = \left\{ a(s) + b(s) b(s)^\top \phi(s; t) \right\} \xi_{s; t} \, ds + b(s) \, d\tilde{V}_s,
\end{align}
where \( \xi_{0; t} \) is a normal random variable, independent of \( \{\tilde{V}_s\} \) and \( \{Y_s\} \), with mean \( 0 \) and covariance matrix
\begin{align}
  \label{eq-cov-xi0}
  V[\xi_{0; t}] = V[X_0]^{\frac{1}{2}} \left( I_{d_1} - V[X_0]^{\frac{1}{2}} \phi(0; t) V[X_0]^{\frac{1}{2}} \right)^{-1} V[X_0]^{\frac{1}{2}}.
\end{align}
Here, note that $I_{d_1} - V[X_0]^{\frac{1}{2}} \phi(0; t) V[X_0]^{\frac{1}{2}}$ is positive-definite since $\phi$ is defined as the negative-semidefinite solution. Furthermore, we define \( \zeta_{s; t} \) for \( 0 \leq s \leq t \) by
\begin{align}
  \label{eq-def-zeta}
  \begin{split}    
    \zeta_{s; t} = &E[X_s] + \xi_{s; t} \\
    &+ \int_0^t \mathrm{Cov}(\xi_{s; t}, \xi_{u; t}) c(u)^\top (\sigma(u) \sigma(u)^\top)^{-1} (dY_u - c(u) E[X_u] \, du).
  \end{split}
\end{align} 
Noting that it holds $E[\xi_{s;t}|\mathcal{Y}_t]=E[\xi_{s;t}]=0$ by $E[\xi_{0;t}]=0$ and the independence, $\{\zeta_{s;t}\}_{0\leq s\leq t}$ is a conditionally Gaussian process with mean
\begin{align}
  \label{eq-mean-zeta}\begin{split}
    E[\zeta_{s;t}|\mathcal{Y}_t]=&E[X_s]\\
  &+ \int_0^t \mathrm{Cov}(\xi_{s; t}, \xi_{u; t}) c(u)^\top (\sigma(u) \sigma(u)^\top)^{-1} (dY_u - c(u) E[X_u] \, du),
  \end{split}  
\end{align}
and covariance 
\begin{align}
  \label{eq-cov-zeta}\mathrm{Cov}(\zeta_{s;t},\zeta_{u;t}|\mathcal{Y}_t)=\mathrm{Cov}(\xi_{s;t},\xi_{u;t}).
\end{align}

Then the distribution of \( \{X_s\}_{0 \leq s \leq t} \) under \( \tilde{P}_t \) is characterized by the following result.

\begin{theorem}\label{main-theorem-linear}
  Let $C_t$ be the space of all continuous functions from $[0,t]$ to $\mathbb{R}^{d_1+d_2}$, and $\mathcal{C}_t$ be the Borel $\sigma$-filed on $C_t$.
  
  Then the law of the process $\{(\zeta_{s;t},Y_s)\}_{0\leq s\leq t}$ on a regular conditional probability measure $P(\,\cdot\,|\mathcal{Y}_t)$ is a version of the law of the process $\{(X_s,Y_s)\}_{0\leq s \leq t}$ on the probability measure $\tilde{P}_t(\,\cdot\,|\mathcal{Y}_t)$, in the sense that for every $A \in \mathcal{C}_t$
  \begin{align*}
    \tilde{P}_t((X.,Y.)\in A|\mathcal{Y}_t)=P((\zeta.,Y.)\in A|\mathcal{Y}_t),
  \end{align*}
  where $X.,Y.$ and $\xi.$ are paths of $\{X_s\}_{0\leq s\leq t},\{Y_s\}_{0\leq s\leq t}$ and $\{\zeta_{s;t}\}_{0\leq s\leq t}$, respectively.
\end{theorem}

\subsubsection{Heuristic Argument}\label{section-heuristic-argument}
Before proving Theorem \ref{main-theorem-linear}, we provide a heuristic argument to illustrate the main ideas behind the proof. The primary task in the proof is to compute the conditional distribution of to calculate conditional distribution of $\{X_s\}_{0\leq s\leq t}$ given $\mathcal{Y}_t$, by directly evaluating the quotient in (\ref{Kallianpur-Striebel-linear}). 

However, a significant challenge arises in applying stochastic calculus techniques, such as the Girsanov theorem, because the numerator and denominator are conditional expectations. In the following, we introduce an "evaluate at 0 methodology" to reduce the quotient of these conditional expectations to the computation of standard expectations.\vspace{10pt}\\
{\bf (I) One-dimensional case:}\\
Before considering the original problem, we explain our idea in a one-dimensional model given by
\begin{align*}
  Y = cX + \sigma Z,
\end{align*}
where \( X \sim N(\mu_x, \sigma_x^2) \) and \( Z \sim N(0, 1) \) are independent random variables, with \( \sigma_x, \sigma > 0 \). Furthermore, we assume $c\neq 0$. In this model, the conditional expectation of \( X \) given \( Y \) can be expressed as
\begin{align}
  E[X | Y = y] &= \frac{\displaystyle \int_{\mathbb{R}} x \exp\left( -\frac{1}{2\sigma^2}(cx - y)^2 \right) p(x) \, dx}{\displaystyle \int_{\mathbb{R}} \exp\left( -\frac{1}{2\sigma^2}(cx - y)^2 \right) p(x) \, dx} \nonumber \\
  &= \frac{\displaystyle \int_{\mathbb{R}} x \exp\left( -\frac{c^2x^2}{2\sigma^2} + \frac{cxy}{\sigma^2} \right) p(x) \, dx}{\displaystyle \int_{\mathbb{R}} \exp\left( -\frac{c^2x^2}{2\sigma^2} + \frac{c xy}{\sigma^2} \right) p(x) \, dx}\nonumber\\
  \label{eq3-14} &=\frac{\displaystyle E\left[ X\exp\left( -\frac{c^2X^2}{2\sigma^2}+\frac{cXy}{\sigma^2} \right) \right]}{\displaystyle E\left[ \exp\left( -\frac{c^2X^2}{2\sigma^2}+\frac{cXy}{\sigma^2} \right) \right]}
\end{align}
where \( p(x) \) is the probability density function of \( X \). This formula is the analogue of (\ref{Kallianpur-Striebel-linear}).

This final expression can be further rewritten as
\begin{align}
   \label{eq3-17}\frac{\sigma^2}{c}\frac{d}{dy} \log E\left[ \exp\left( -\frac{c^2X^2}{2\sigma^2}+\frac{cXy}{\sigma^2} \right) \right],
\end{align}
as known as the Tweedie's formula. Due to the normality of \( X \), it can be shown that \(\displaystyle \log E\left[ \exp\left( -\frac{X^2}{2\sigma^2}+\frac{Xy}{\sigma^2} \right) \right] \) is a quadratic function of \( y \). Therefore, we can write
\begin{align*}
  E[X | Y = y]=\frac{\displaystyle E\left[ X\exp\left( -\frac{c^2X^2}{2\sigma^2}+\frac{cXy}{\sigma^2} \right) \right]}{\displaystyle E\left[ \exp\left( -\frac{c^2X^2}{2\sigma^2}+\frac{cXy}{\sigma^2} \right) \right]}=Ay+B.
\end{align*}
By taking expectation of both sides (with respect to $y$), we obtain
\begin{align*}
  E[X]=AE[Y]+B=AcE[X]+B.
\end{align*}
Thus we have $B=E[X]-AcE[X]$ and hence
\begin{align}
  \label{eq3-25}E[X | Y = y]=\frac{\displaystyle E\left[ X\exp\left( -\frac{c^2X^2}{2\sigma^2}+\frac{cXy}{\sigma^2} \right) \right]}{\displaystyle E\left[ \exp\left( -\frac{c^2X^2}{2\sigma^2}+\frac{cXy}{\sigma^2} \right) \right]}=E[X]+A(y-cE[X]).
\end{align}
Furthermore, by differentiating (\ref{eq3-25}), $A$ can be determined as
\begin{align}
  A=&\frac{d}{dy}\frac{\displaystyle E\left[ X\exp\left( -\frac{c^2X^2}{2\sigma^2}+\frac{cXy}{\sigma^2} \right) \right]}{\displaystyle E\left[ \exp\left( -\frac{c^2X^2}{2\sigma^2}+\frac{cXy}{\sigma^2} \right) \right]}\nonumber\\
  \label{eq3-26}\begin{split}
    =&\frac{c}{\sigma^2}\frac{\displaystyle  E\left[ X^2\exp\left( -\frac{c^2X^2}{2\sigma^2}+\frac{cXy}{\sigma^2} \right) \right]}{\displaystyle E\left[ \exp\left( -\frac{c^2X^2}{2\sigma^2}+\frac{cXy}{\sigma^2} \right) \right]} \\
  &- \frac{c}{\sigma^2} \left( \frac{\displaystyle E\left[ X\exp\left( -\frac{c^2X^2}{2\sigma^2}+\frac{cXy}{\sigma^2} \right) \right]}{\displaystyle E\left[ \exp\left( -\frac{c^2X^2}{2\sigma^2}+\frac{cXy}{\sigma^2} \right) \right]} \right)^2
  \end{split}  
\end{align}
Since $A$ does not depend on $y$, we can set $y=0$ in the final expression, which yields
\begin{align}
  \label{eq3-27}A=\frac{c}{\sigma^2} \left\{ \frac{\displaystyle E\left[ X^2 \exp\left( -\frac{c^2X^2}{2\sigma^2} \right) \right]}{\displaystyle E\left[ \exp\left( -\frac{c^2X^2}{2\sigma^2} \right) \right]} - \left( \frac{\displaystyle E\left[ X \exp\left( -\frac{c^2X^2}{2\sigma^2} \right) \right]}{\displaystyle E\left[ \exp\left( -\frac{c^2X^2}{2\sigma^2} \right) \right]} \right)^2 \right\}.
\end{align}
Therefore, (\ref{eq3-25}) ends up in the conditional mean formula
\begin{align}
  \label{eq3-28}\begin{split}
    &E[X | Y = y] = E[X] \\
  &+ \frac{c}{\sigma^2} \left\{ \frac{\displaystyle E\left[ X^2 \exp\left( -\frac{c^2X^2}{2\sigma^2} \right) \right]}{\displaystyle E\left[ \exp\left( -\frac{c^2X^2}{2\sigma^2} \right) \right]} - \left( \frac{\displaystyle E\left[ X \exp\left( -\frac{c^2X^2}{2\sigma^2} \right) \right]}{\displaystyle E\left[ \exp\left( -\frac{c^2X^2}{2\sigma^2} \right) \right]} \right)^2 \right\} (y - cE[X]).
  \end{split}  
\end{align}

Moreover, recalling the Bayesian formula (\ref{eq3-14}), we notice that (\ref{eq3-26}) is equivalent to
\begin{align*}
  A=\frac{c}{\sigma^2}(E[X^2|Y=y]-E[X|Y=y]^2)=\frac{c}{\sigma^2}E[V|Y=y].
\end{align*}
Since $A$ is given by (\ref{eq3-27}), we obtain the conditional variance formula
\begin{align}
  \label{eq3-29}V[X|Y=y]=\frac{\displaystyle E\left[ X^2 \exp\left( -\frac{c^2X^2}{2\sigma^2} \right) \right]}{\displaystyle E\left[ \exp\left( -\frac{c^2X^2}{2\sigma^2} \right) \right]} - \left( \frac{\displaystyle E\left[ X \exp\left( -\frac{c^2X^2}{2\sigma^2} \right) \right]}{\displaystyle E\left[ \exp\left( -\frac{c^2X^2}{2\sigma^2} \right) \right]} \right)^2.
\end{align}
Further differentiating (\ref{eq3-26}) yields the third and higher cumulants of $X$ given $Y$, but they are zero since $A$ is a constant. In this way, we can conclude that the conditional distribution of $X$ given $Y=y$ is normal with the mean and variance given by (\ref{eq3-28}) and (\ref{eq3-29}), respectively.

The key contribution here is not merely the reminiscent of Tweedie’s formula. Rather, we proceed by differentiating the conditional expectation, evaluating it at zero, and deriving the non-trivial expressions in (\ref{eq3-28}) and (\ref{eq3-29}), which represent the conditional distribution as a quotient of ordinary expectations. In the infinite-dimensional setting, this formulation enables us to analyse the conditional distribution using classical tools from stochastic calculus, such as It\^o's formula and the Girsanov theorem.\vspace{0.3cm}\\
\noindent
{\bf (II) The original case:}\\
Now, let us apply the same methodology to the original model. Although not all steps are rigorously justified, the following informal computations serve to convey the intuition and underlying structure of the argument. 

For simplicity, assume $a$, $b$, $c$, $X$ and $Y$ are one-dimensional, and consider the case of $\sigma=1$. Also, we assume that $c(t)>0$ for all $t\geq 0$.

Due to Proposition \ref{prop-Kallianpur-Striebel-linear}, the conditional expectation of $X_s~(0\leq s\leq t)$ given $\mathcal{Y}_t$ can be written as
\begin{align}
  \label{eq3-15}E[X_s|\mathcal{Y}_t]=&\frac{\displaystyle E_Q\left[ X_s \exp\left( \int_0^tc(s)X_sdY_s-\frac{1}{2}\int_0^tc(s)^2X_s^2ds\right)\middle|\mathcal{Y}_t \right]}{\displaystyle E_Q\left[ \exp\left( \int_0^tc(s)X_sdY_s-\frac{1}{2}\int_0^tc(s)^2X_s^2ds\right)\middle|\mathcal{Y}_t\right]}
\end{align}
({\bf corresponding to (\ref{eq3-14})}). 

Recall that under $Q$, $Y$ is a Brownian motion independent of $X$. Hence, if $D$ denotes the Malliavin derivative with respect to $Y$, then we heuristically observe that
\begin{align}
  \label{eq-Malliavin-derivative}\begin{split}
    &D_sE_Q\left[ \exp\left( \int_0^tc(s)X_sdY_s-\frac{1}{2}\int_0^tc(x)^2X_s^2ds\right)\middle|\mathcal{Y}_t\right]\\
  &=D_s\int\exp\left( \int_0^tc(s)x(s)dY_s-\frac{1}{2}\int_0^tc(s)^2x(s)^2ds\right)\mu_X(dx)\\
  &=\int D_s\exp\left( \int_0^tc(s)x(s)dY_s-\frac{1}{2}\int_0^tc(s)^2x(s)^2ds\right)\mu_X(dx)\\
  &=\int c(s)\exp\left( \int_0^tc(s)x(s)dY_s-\frac{1}{2}\int_0^tc(s)^2x(s)^2ds\right)\mu_X(dx)\\
  &=c(s)E_Q\left[X_s \exp\left( \int_0^tc(s)X_sdY_s-\frac{1}{2}\int_0^tc(s)^2ds\right)\middle|\mathcal{Y}_t\right]
  \end{split}  
\end{align}
Therefore, by the chain rule, (\ref{eq3-15}) can be rewritten as
\begin{align}
  \label{eq3-19}\frac{1}{c(s)}D_s \log E_Q\left[X_s \exp\left( \int_0^tc(s)X_sdY_s-\frac{1}{2}\int_0^tc(s)^2X_s^2ds\right)\middle|\mathcal{Y}_t\right]
\end{align}
({\bf corresponding to (\ref{eq3-17})}).\\
Here, as in the one-dimensional case, the Gaussianity of $X$ yields the form (which is shown in Lemma \ref{lemma-3-1})
\begin{align}
  \label{eq-2nd-Wiener-Chaos}\begin{split}
    &E_Q\left[X_s \exp\left( \int_0^tc(s)X_sdY_s-\frac{1}{2}\int_0^tc(s)^2X_s^2ds\right)\middle|\mathcal{Y}_t\right]\\
  &=\exp\left( h_0(t)+\int_0^t h_1(s;t)dY_s+\int_0^t\int_0^s h_2(s,u;t)dY_u dY_s \right).
  \end{split}  
\end{align}
Thus (\ref{eq3-15}) (which is equals to (\ref{eq3-19})) can be written in the form
\begin{align*}
  E[X_s|\mathcal{Y}_t]=&\frac{\displaystyle E_Q\left[ X_s \exp\left( \int_0^tc(s)X_sdY_s-\frac{1}{2}\int_0^tc(s)^2X_s^2ds\right)\middle|\mathcal{Y}_t \right]}{\displaystyle E_Q\left[ \exp\left( \int_0^tc(s)X_sdY_s-\frac{1}{2}\int_0^tc(s)^2X_s^2ds\right)\middle|\mathcal{Y}_t\right]}\\
  =&g_0(t)+\int_0^t g_1(s,u;t)dY_u.
\end{align*}
By taking the expectations of both sides, we obtain
\begin{align*}
  E[X_s]=g_0(t)+\int_0^t g_1(s,u;t)c(u)E[X_u]du,
\end{align*}
and thus it holds
\begin{align}  
    E[X_s|\mathcal{Y}_t]=&\frac{\displaystyle E_Q\left[ X_s \exp\left( \int_0^tc(s)X_sdY_s-\frac{1}{2}\int_0^tc(s)^2X_s^2ds\right)\middle|\mathcal{Y}_t \right]}{\displaystyle E_Q\left[ \exp\left( \int_0^tc(s)X_sdY_s-\frac{1}{2}\int_0^tc(s)^2X_s^2ds\right)\middle|\mathcal{Y}_t\right]}\nonumber\\
    \label{eq3-18}=&E[X_t]+\int_0^t g_1(s,u;t)(dY_u-c(u)E[X_u]du).
\end{align}
({\bf corresponding to (\ref{eq3-25})}).

Furthermore, $g_1$ is obtained by the Malliavin derivative:
\begin{align}  
  &g_1(s,u;t)\nonumber\\
  =&D_u\frac{\displaystyle E_Q\left[ X_s \exp\left( \int_0^tc(s)X_sdY_s-\frac{1}{2}\int_0^tc(s)^2X_s^2ds\right)\middle|\mathcal{Y}_t \right]}{\displaystyle E_Q\left[ \exp\left( \int_0^tc(s)X_sdY_s-\frac{1}{2}\int_0^tc(s)^2X_s^2ds\right)\middle|\mathcal{Y}_t\right]}\nonumber\\
  \label{eq3-16}\begin{split}
    =&c(u)\left( \frac{\displaystyle E_Q\left[ X_s X_u\exp\left( \int_0^tc(s)X_sdY_s-\frac{1}{2}\int_0^tc(s)^2X_s^2ds\right)\middle|\mathcal{Y}_t \right]}{\displaystyle E_Q\left[ \exp\left( \int_0^tc(s)X_sdY_s-\frac{1}{2}\int_0^tc(s)^2X_s^2ds\right)\middle|\mathcal{Y}_t\right]}\right.\\
    &-\frac{\displaystyle E_Q\left[ X_s \exp\left( \int_0^tc(s)X_sdY_s-\frac{1}{2}\int_0^tc(s)^2X_s^2ds\right)\middle|\mathcal{Y}_t \right]}{\displaystyle E_Q\left[ \exp\left( \int_0^tc(s)X_sdY_s-\frac{1}{2}\int_0^tc(s)^2X_s^2ds\right)\middle|\mathcal{Y}_t\right]}\\
    &\left.\times \frac{\displaystyle E_Q\left[ X_u \exp\left( \int_0^tc(s)X_sdY_s-\frac{1}{2}\int_0^tc(s)^2X_s^2ds\right)\middle|\mathcal{Y}_t \right]}{\displaystyle E_Q\left[ \exp\left( \int_0^tc(s)X_sdY_s-\frac{1}{2}\int_0^tc(s)^2X_s^2ds\right)\middle|\mathcal{Y}_t\right]} \right)
  \end{split}    
\end{align}
Since $g_1$ does not depend on $Y$, we can substitute $0$ into $Y$, and obtain
\begin{align}
  \label{eq3-21}\begin{split}
    g_1(s,u;t)=c(u)\gamma(s,u;t)
  \end{split}  
\end{align}
where
\begin{align}
  \label{eq-gamma-s-u-t}\begin{split}
    &\gamma(s,u;t)=\frac{\displaystyle E\left[ X_s X_u\exp\left(-\frac{1}{2}\int_0^tc(s)^2X_s^2ds\right)\middle|\mathcal{Y}_t \right]}{\displaystyle E_Q\left[ \exp\left( -\frac{1}{2}\int_0^tc(s)^2X_s^2ds\right)\middle|\mathcal{Y}_t\right]}\\
  &-\frac{\displaystyle E\left[ X_s \exp\left(-\frac{1}{2}\int_0^tc(s)^2X_s^2ds\right)\middle|\mathcal{Y}_t \right]}{\displaystyle E\left[ \exp\left( -\frac{1}{2}\int_0^tc(s)^2X_s^2ds\right)\middle|\mathcal{Y}_t\right]}\times \frac{\displaystyle E\left[ X_u \exp\left( -\frac{1}{2}\int_0^tc(s)^2X_s^2ds\right)\middle|\mathcal{Y}_t \right]}{\displaystyle E_Q\left[ \exp\left( -\frac{1}{2}\int_0^tc(s)^2X_s^2ds\right)\middle|\mathcal{Y}_t\right]}
  \end{split}  
\end{align}
({\bf corresponding to (\ref{eq3-27})}). Here, we again used the fact that $X$ has the same law under $P$ and $Q$.
Substituting this expression into (\ref{eq3-18}), we obtain the conditional expectation formula
\begin{align}
  \label{eq3-30}E[X_s|\mathcal{Y}_t]=&E[X_t]+\int_0^t c(u)\gamma(s,u;t)(dY_u-c(u)E[X_u]du).
\end{align}
({\bf corresponding to (\ref{eq3-28})}). Moreover, recalling the Kallianpur-Striebel formula (\ref{eq3-15}), we can see (\ref{eq3-16}) as
\begin{align*}
  E[X_sX_u|\mathcal{Y}_t]-E[X_s|\mathcal{Y}_t]E[X_u|\mathcal{Y}_t]=\mathrm{Cov}(X_s,X_u|\mathcal{Y}_t).
\end{align*}
Since this is equivalent to (\ref{eq3-21}), we obtain the conditional covariance formula
\begin{align}
  \label{eq3-31}\mathrm{Cov}(X_s,X_u|\mathcal{Y}_t)=\gamma(s,u;t).
\end{align}
({\bf corresponding to (\ref{eq3-29})}). The third and higher derivatives can be obtained as further Malliavin derivatives of (\ref{eq3-16}), but they should be zero since $g_1$ does not depend on $Y$. 

In this way, we have shown that the conditional distribution of $\{X_s\}_{0\leq s\leq t}$ given $\mathcal{Y}_t$ is a Gaussian process where the mean and covariance is given by (\ref{eq3-30}) and (\ref{eq3-31}). Therefore, calculation of the conditional distribution is reduced to evaluating $\gamma(s,u;t)$ given by (\ref{eq-gamma-s-u-t}).

In order to evaluate (\ref{eq-gamma-s-u-t}), let us introduce $\phi(s;t)$ and $\xi_{s;t}~(0\leq s\leq t)$ defined by (\ref{eq-def-phi}) and (\ref{eq-def-xi}) (with $\sigma=1$), and assume $X_0=0$ for simplicity. Then by the density formula in \citet{Liptser2001}, we have
\begin{align}
  \label{eq3-20}\begin{split}
    \frac{d\mu_{\xi.}}{d\mu_{X.}}(X.)=&\exp\left( \int_0^t\phi(s;t)X_sdX_s\right.\\
  &\left.-\frac{1}{2}\int_0^t\left\{ 2a(s)+b(s)^2\phi(s;t) \right\}\phi(s;t)X_s^2ds \right).
  \end{split}
\end{align}
where $X.$ and $\xi.$ are the paths of $\{X_s\}_{0\leq s\leq t}$ and $\{\xi_{s;t}\}_{0\leq s\leq t}$ respectively, and $\mu_X$ and $\mu_\xi$ are their distributions.
By It\^o's formula and (\ref{eq-def-phi}), we have
\begin{align*}
  &\phi(t;t)X_t^2-\phi(0;t)X_0\\
  =&2\int_0^t \phi(s;t)X_s dX_s+\int_0^t X_s^2 \frac{d}{ds}\phi(s;t)ds+\int_0^t \phi(s;t)b(s)^2ds\\
  =&2\int_0^t \phi(s;t)X_s dX_s+\int_0^t X_s^2 \left\{ -\phi(s;t)^2b(s)^2-2a(s)\phi(s;t)+c(s)^2 \right\}ds\\
  &+\int_0^t \phi(s;t)b(s)^2ds.
\end{align*}
Recalling $\phi(t;t)=0$ and $X_0=0$, we obtain 
\begin{align*}
  &\int_0^t \phi(s;t)X_s dX_s\\
  =&\frac{1}{2}\int_0^t\left\{ 2a(s)+b(s)^2\phi(s;t) \right\}\phi(s;t)X_s^2ds\\
  &-\int_0^t c(s)^2X_s^2ds-\int_0^t \phi(s;t)b(s)^2ds
\end{align*}
Therefore, (\ref{eq3-20}) can be rewritten as
\begin{align*}
  \frac{d\mu_{\xi.}}{d\mu_{X.}}(X.)=\exp\left( -\int_0^t c(s)^2X_s^2ds-\int_0^t \phi(s;t)b(s)^2ds \right),
\end{align*}
and (\ref{eq-gamma-s-u-t}) yields
\begin{align}
  \label{eq3-22}\begin{split}
    \gamma(s,u;t)
  =&\frac{\displaystyle E\left[X_u X_s \frac{d\mu_{\xi.}}{d\mu_{X.}}(X.) \right] \exp\left( -\int_0^t \phi(s;t)b(s)^2ds \right)}{\displaystyle \exp\left( -\int_0^t \phi(s;t)b(s)^2ds\right)}\\
  &-\frac{\displaystyle E\left[X_s \frac{d\mu_{\xi.}}{d\mu_{X.}}(X.) \right] \exp\left( -\int_0^t \phi(s;t)b(s)^2ds \right)}{\displaystyle \exp\left( -\int_0^t \phi(s;t)b(s)^2ds\right)}\\
  &\times  \frac{\displaystyle E\left[X_u  \frac{d\mu_{\xi.}}{d\mu_{X.}}(X.) \right] \exp\left( -\int_0^t \phi(s;t)b(s)^2ds \right)}{\displaystyle \exp\left( -\int_0^t \phi(s;t)b(s)^2ds\right)}.\\
  =&E[\xi_{s;t} \xi_{u;t}]-E[\xi_{s;t}]E[\xi_{u;t}]=\mathrm{Cov}(\xi_{s;t},\xi_{u;t}).
  \end{split}  
\end{align}  
Using this result, (\ref{eq3-30}) and (\ref{eq3-31}) can be rewritten as
\begin{align*}
  E[X_s|\mathcal{Y}_t]=&E[X_t]+\int_0^t c(u)\gamma(s,u;t)(dY_u-c(u)E[X_u]du).
\end{align*}
and
\begin{align*}
  \mathrm{Cov}(X_s,X_u|\mathcal{Y}_t)=E[\xi_{s;t}\xi_{u;t}],
\end{align*}
which leads to Theorem \ref{main-theorem-linear}.

\subsubsection{Proof of the Theorem}\label{section-proof-main-theorem-linear}
In this section, we present the proof based on the discussion above. The core of the proof involves justifying the Malliavin derivative (\ref{eq-Malliavin-derivative}), which will be stated in Lemma \ref{lemma3-3}. In this lemma, we first establish (\ref{eq-Malliavin-derivative}) by replacing the stochastic integral with a discrete sum, and then obtain the desired result by taking the limit of this equality. The convergence of the limit is guaranteed by Lemma \ref{lemma-3-2}.

Before proceeding with this argument, we start by verifying the expression (\ref{eq-2nd-Wiener-Chaos}) using the Gaussian property of \(X\). For this purpose, we rewrite $Z_t$ as
\begin{align}
  Z_t=&\exp\left( \int_0^tX_s^\top c(s)^\top(\sigma(s)\sigma(s)^\top)^{-1} dY_s\nonumber\right.\\
  &\left.-\frac{1}{2}\int_0^tX_s^\top c(s)^\top(\sigma(s)\sigma(s)^\top)^{-1} c(s)X_s ds\right)\nonumber\\
  \label{eq-Z-Ito}\begin{split}
    =&\exp\left( X_t^\top \hat{Y}_t- \int_0^t\hat{Y}_s^\top dX_s-\frac{1}{2}\int_0^tX_s^\top c(s)^\top(\sigma(s)\sigma(s)^\top)^{-1} c(s)X_s ds\right),
  \end{split}
\end{align}
where
\begin{align}
  \label{def-hat-Y}\hat{Y}_s=\int_0^s c(u)^\top(\sigma(u)\sigma(u)^\top)^{-1}dY_u.
\end{align}
Furthermore, let us write
\begin{align}
  \label{eq-def-zt-y}\begin{split}
    z_t(y)=&\exp\left( X_t^\top y(t)- \int_0^ty(s)^\top dX_s\right.\\
  &\left.-\frac{1}{2}\int_0^tX_s^\top c(s)^\top(\sigma(s)\sigma(s)^\top)^{-1} c(s)X_s ds\right)
  \end{split}  
\end{align}
for a bounded measurable function $y:[0,t]\to \mathbb{R}^{d_1}$.

Then the following lemma holds.
\begin{lemma}\label{lemma-3-1}
  Let $\{\overline{\xi}_{s}\}_{0\leq s\leq t}$ be the solution of
  \begin{align}
    \label{def-xi-bar}d{\overline{\xi}}_{s;t}=\left\{ a(s)+b(s)b(s)^\top \phi(s;t) \right\}{\overline{\xi}}_{s;t}ds+b(s)d\tilde{V}_s,~~{\overline{\xi}}_{0;t}=0,
  \end{align}
  and define
  \begin{align}
    \label{def-alpha-s-t}\alpha(s;t)=\exp\left( \int_0^s \{a(u)+b(u)b(u)^\top \phi(u;t)\}du \right).
  \end{align}

  Then for any bounded measurable $y:[0,t]\to \mathbb{R}^{d_1}$ and $u_1,\cdots,u_n \in [0,t]$, it holds
  \begin{align}
    \label{eq-E-z-t}\begin{split}
      &E_Q[z_t(y)]\\
    =&\sqrt{\det(I_{d_1}-\Sigma^\frac{1}{2}\phi(0;t)\Sigma^\frac{1}{2})}\exp\left( \frac{1}{2}\int_0^t b(s)^\top \phi(s;t)b(s)ds \right)\\
  &\times \exp\biggl(\frac{1}{2}\beta(t;y)^\top \Sigma^\frac{1}{2}(I_{d_1}-\Sigma^\frac{1}{2}\phi(0;t)\Sigma^\frac{1}{2})^{-1}\Sigma^\frac{1}{2}\beta(t;y) \biggr)\\
  &\times \exp\biggl(\mu^\top \alpha(t;t)^\top y(t)\\
  &- \int_0^ty(s)^\top \left\{ a(s)+b(s)b(s)^\top \phi(s;t) \right\}\alpha(s;t)ds \mu+\frac{1}{2}\mu^\top \phi(0;t) \mu\biggr)\\
  &\times \exp\biggl( \frac{1}{2}y(t)^\top E[\overline{\xi}_{t;t}\overline{\xi}_{t;t}^\top]y(t)
  +\frac{1}{2}\int_0^t y(s)^\top b(s)b(s)^\top y(s)ds\\
  &-\int_0^t y(s) b(s)b(s)^\top y(t) ds  \biggr),
    \end{split}
  \end{align}
  and
  \begin{align}
    \label{eq-E-frac-z-t}\begin{split}
      &\frac{E_Q[ X_{u_1}\otimes \cdots \otimes X_{u_{n}}z_t(y)]}{E_Q[ z_t(y)]}\\   
    =&E\left[\bigotimes_{j=1}^n\biggl\{  \alpha(u_j;t)\left(\mu+ \Sigma^\frac{1}{2}(I_{d_1}-\Sigma^\frac{1}{2}\phi(0;t)\Sigma^\frac{1}{2})^{-\frac{1}{2}} \bigl(U+\beta(t;y)\bigr) \right) \right.\\
    &\qquad \qquad+  \overline{\xi}_{u_j;t}+E[\overline{\xi}_{u_j;t}\eta_t(y)]\biggr\}\Biggr],
    \end{split}    
  \end{align}
  where
  \begin{align*}
    &\mu=E[X_0],~~\Sigma=\mathrm{Cov}(X_0,X_0),\\
    &\beta(t;y)= \alpha(t;t)^\top y(t)- \int_0^t \alpha(s;t)^\top \left\{ a(s)+b(s)b(s)^\top \phi(s;t) \right\}^\top y(s) ds+ \phi(0;t)\mu,\\
    &\eta_t(y)=\overline{\xi}_{t;t}^\top y(t)- \int_0^ty(s)^\top d\overline{\xi}_{s;t},\\
    &E[\overline{\xi}_{s;t}\eta_t(y)]=E[\overline{\xi}_{s;t}\overline{\xi}_{t;t}^\top]y(t)-\int_0^s  b(u)b(u)^\top y(u)du,
  \end{align*}
  and $U$ is a $d_1$-dimensional standard normal random variable independent of $\{\overline{\xi}_{s,t}\}_{0\leq s\leq t}$.
\end{lemma}
\begin{proof}
  Let us define $\{{\breve{\xi}}_{s,t}\}_{0\leq s\leq t}$ by
  \begin{align}
    \label{eq-breve-xi}d{\breve{\xi}}_{s;t}=\left\{ a(s)+b(s)b(s)^\top \phi(s;t) \right\}{\breve{\xi}}_{s;t}ds+b(s)d\tilde{V}_s,~~{\breve{\xi}}_{0;t}=X_0,
  \end{align} 
  and let $X.$ and ${\breve{\xi}}.$ be the paths of $\{X_s\}_{0\leq s\leq t}$ and $\{{\breve{\xi}}_{s,t}\}_{0\leq s\leq t}$, and $\mu_X$ and $\mu_{\breve{\xi}}$ be their distributions on $P$. Then by (7.138) in \citet{Liptser2001}, $\mu_X$ and $\mu_{\breve{\xi}}$ are equivalent, and it holds that 
  \begin{align}
    \label{eq3-3}\begin{split}      
    \frac{d\mu_{\breve{\xi}}}{d\mu_X}(X.)=\exp\biggl(&\int_0^t X_{s}^\top\phi(s;t) dX_s\\
    &-\frac{1}{2}\int_0^tX_s^\top \phi(s;t) (2a(s)+b(s)b(s)^\top \phi(s;t))X_s ds \biggr).
    \end{split}
  \end{align}
  Here, note that the condition (7.137) in the reference can be replaced by 
  \begin{align*}
    \int_0^T |\alpha_t(x)|^2 dt<\infty~~\mu_{{\breve{\xi}}}~\textrm{and}~\mu_\eta\textrm{-a.s.},
  \end{align*}
  and this condition can be written as
  \begin{align*}
    \int_0^t |b(s)b(s)^\top \phi(s;t){\breve{\xi}}_{s,t}|^2ds<\infty,~~\int_0^t |b(s)b(s)^\top \phi(s;t)X_{s}|^2ds<\infty~~\textrm{a.s.}
  \end{align*}
  in our setup. This can be verified by (\ref{eq-assumption-b}) and the continuities of $\phi(s;t)$, $X_s$ and ${\breve{\xi}}_{s;t}$. 

  On the other hand, by It\^o's formula and (\ref{eq-def-phi}), we have
  \begin{align*}
    &X_t^\top \phi(t;t) X_t-X_0^\top \phi(0;t) X_0\\
    =&2\int_0^t X_s^\top\phi(s;t)dX_s+\int_0^t X_s^\top \frac{\partial}{\partial s}\phi(s;t) X_sds + \int_0^t b(s)^\top \phi(s;t)b(s)ds\\
    =&2\int_0^t X_s^\top\phi(s;t)dX_s-\int_0^t X_s^\top \phi(s;t)b(s)b(s)^\top \phi(s;t) X_sds \\
    &-\int_0^t X_s^\top a(s)^\top \phi(s;t) X_sds-\int_0^t X_s^\top  \phi(s;t)a(s) X_sds\\
    &+\int_0^t X_s^\top c(s)^\top (\sigma(s)\sigma(s)^\top)^{-1} c(s) X_sds
    + \int_0^t b(s)^\top \phi(s;t)b(s)ds\\
    =&2\int_0^t X_s^\top\phi(s;t)dX_s-\int_0^t X_s^\top \phi(s;t)b(s)b(s)^\top \phi(s;t) X_sds \\
    &-2\int_0^t X_s^\top  \phi(s;t)a(s) X_sds\\
    &+\int_0^t X_s^\top c(s)^\top (\sigma(s)\sigma(s)^\top)^{-1} c(s) X_sds
    + \int_0^t b(s)^\top \phi(s;t)b(s)ds.
  \end{align*}
  Therefore, noting that $\phi(t;t)=0$, it holds
  \begin{align*}
    \int_0^t X_s^\top\phi(s;t)dX_s=&\frac{1}{2}\int_0^t X_s^\top \phi(s;t)b(s)b(s)^\top \phi(s;t) X_sds\\&+\int_0^t X_s^\top  \phi(s;t)a(s) X_sds\\
    &-\frac{1}{2}\int_0^t X_s^\top c(s)^\top (\sigma(s)\sigma(s)^\top)^{-1} c(s) X_sds\\
    &-\frac{1}{2}\int_0^t b(s)^\top \phi(s;t)b(s)ds-\frac{1}{2}X_0^\top \phi(0;t) X_0.
  \end{align*}  
  Together with (\ref{eq3-3}), we obtain
  \begin{align}
    \label{eq-measure-change}\begin{split}      
    \frac{d\mu_{\breve{\xi}}}{d\mu_X}(X.)=\exp\biggl( &-\frac{1}{2}\int_0^t X_s^\top c(s)^\top (\sigma(s)\sigma(s)^\top)^{-1} c(s) X_sds\\
    &-\frac{1}{2}\int_0^t b(s)^\top \phi(s;t)b(s)ds-\frac{1}{2}X_0^\top \phi(0;t) X_0\biggr).
    \end{split}
  \end{align}
  Hence, it follows that
  \begin{align*}
    z_t(y)=&\exp\left( X_t^\top y(t)-\int_0^t y(s)^\top dX_s\right.\\
    &\left.+\frac{1}{2}\int_0^t b(s)^\top \phi(s;t)b(s)ds+\frac{1}{2}X_0^\top \phi(0;t) X_0 \right)\frac{d\mu_{\breve{\xi}}}{d\mu_X}(X.)
  \end{align*}
  and
  \begin{align}
    \label{eq-z-y-1}\begin{split}
      E_Q[z_t(y)]=E\biggl[ \exp\biggl(&{\breve{\xi}}_{t;t}^\top y(t)- \int_0^ty(s)^\top d{\breve{\xi}}_{s;t}\\
     &+\frac{1}{2}\left. \left.\int_0^t b(s)^\top \phi(s;t)b(s)ds+\frac{1}{2}X_0^\top \phi(0;t) X_0\right) \right].
    \end{split}    
  \end{align}
  Furthermore, note that $\breve{\xi}_{s,t}$ in (\ref{eq-breve-xi}) can be written as
  \begin{align}
    \label{eq-breve-xi-decompisition}\breve{\xi}_{s,t}=\alpha(s;t)X_0+\overline{\xi}_{s,t}
  \end{align}
  (for the definitions of $\overline{\xi}$ and $\alpha$, see (\ref{def-xi-bar}) and (\ref{def-alpha-s-t})). Substituting this expression into (\ref{eq-z-y-1}), we obtain
  \begin{align}
    &E_Q[z_t(y)]\nonumber \\
     \label{eq3-7}\begin{split}
      =&\exp\left( \frac{1}{2}\int_0^t b(s)^\top \phi(s;t)b(s)ds \right)\\
     &\times E\biggl[ \exp\biggl(X_0^\top \alpha(t;t)^\top y(t)\\
     &\left. \left.- \int_0^ty(s)^\top \left\{ a(s)+b(s)b(s)^\top \phi(s;t) \right\}\alpha(s;t)ds X_0+\frac{1}{2}X_0^\top \phi(0;t) X_0\right) \right]\\
     &\times E\left[ \exp\left(\overline{\xi}_{t;t}^\top y(t)- \int_0^ty(s)^\top d\overline{\xi}_{s;t}\right) \right].
     \end{split}     
  \end{align}

  For the first expectation on the right-hand side, let \( X_0 = \mu + \Sigma^{\frac{1}{2}} U \), where \( U \) is a standard normal random variable, \( \mu \in \mathbb{R}^{d_1} \), and \( \Sigma \in M_{d_1}(\mathbb{R}) \). Then, by straightforward calculation, we obtain
  \begin{align}
    &E\biggl[ \exp\biggl(X_0^\top \alpha(t;t)^\top y(t)\nonumber\\
     &\left. \left.- \int_0^ty(s)^\top \left\{ a(s)+b(s)b(s)^\top \phi(s;t) \right\}\alpha(s;t)ds X_0+\frac{1}{2}X_0^\top \phi(0;t) X_0\right) \right]\nonumber\\
     \label{eq3-8}\begin{split}
      =&\sqrt{\det|I_{d_1}-\Sigma^\frac{1}{2}\phi(0;t)\Sigma^\frac{1}{2}|}\\
      &\times \exp\biggl(\frac{1}{2}\beta(t;y)^\top \Sigma^\frac{1}{2}(I_{d_1}-\Sigma^\frac{1}{2}\phi(0;t)\Sigma^\frac{1}{2})^{-1}\Sigma^\frac{1}{2}\beta(t;y) \biggr)\\
     &\times \exp\biggl(\mu^\top \alpha(t;t)^\top y(t)\\
     &- \int_0^ty(s)^\top \left\{ a(s)+b(s)b(s)^\top \phi(s;t) \right\}\alpha(s;t)ds \mu+\frac{1}{2}\mu^\top \phi(0;t) \mu\biggr).
     \end{split}     
  \end{align}  
  Here, we used the negative semi-definiteness of $\phi(0;t)$.

  For the second expectation, since $\displaystyle \overline{\xi}_{t;t}^\top y(t)- \int_0^ty(s)^\top d\overline{\xi}_{s;t}$ is a Gaussian random variable with mean 0, it follows that 
  \begin{align}
    &E\left[ \exp\left(\overline{\xi}_{t;t}^\top y(t)- \int_0^ty(s)^\top d\overline{\xi}_{s;t}\right) \right]\nonumber\\
    =&\exp\left( \frac{1}{2}E\left[ \left( \overline{\xi}_{t;t}^\top y(t)- \int_0^ty(s)^\top d\overline{\xi}_{s;t} \right)^2 \right] \right)\nonumber\\
    \label{eq3-9}\begin{split}      
    =&\exp\biggl( \frac{1}{2}y(t)^\top E[\overline{\xi}_{t;t}\overline{\xi}_{t;t}^\top]y(t)
    +\frac{1}{2}\int_0^t y(s)^\top b(s)b(s)^\top y(s)ds\\
    &-\int_0^t y(s)^\top b(s)b(s)^\top y(t) ds  \biggr).
    \end{split}
  \end{align}
  Putting (\ref{eq3-7}), (\ref{eq3-8}) and (\ref{eq3-9}) together, we obtain (\ref{eq-E-z-t}).

  In the same way as (\ref{eq-z-y-1}) and (\ref{eq3-7}), we can show that
  \begin{align*}
    &E_Q[ X_{u_1}\otimes \cdots \otimes X_{u_{n}}z_t(y)]\nonumber\\
    =&E\biggl[\breve{\xi}_{u_1;t}\otimes\cdots \otimes \breve{\xi}_{u_{n};t} \exp\biggl({\breve{\xi}}_{t;t}^\top y(t)- \int_0^ty(s)^\top d{\breve{\xi}}_{s;t}\\
    &+\frac{1}{2}\left. \left.\int_0^t b(s)^\top \phi(s;t)b(s)ds+\frac{1}{2}X_0^\top \phi(0;t) X_0\right) \right]\\
     \begin{split}
      =&\exp\left( \frac{1}{2}\int_0^t b(s)^\top \phi(s;t)b(s)ds \right)\\
     &\times E\biggl[\{\alpha(u_1;t)X_0+\overline{\xi}_{u_1;t}\}\otimes \cdots \otimes \{\alpha(u_1;t)X_0+\overline{\xi}_{u_1;t}\}\\
     &\times  \exp\biggl(X_0^\top \alpha(t;t)^\top y(t)\\
     &- \int_0^ty(s)^\top \left\{ a(s)+b(s)b(s)^\top \phi(s;t) \right\}\alpha(s;t)ds X_0+\frac{1}{2}X_0^\top \phi(0;t) X_0\biggr)\\
     &\left.\times  \exp\left(\overline{\xi}_{t;t}^\top y(t)- \int_0^ty(s)^\top d\overline{\xi}_{s;t}\right) \right].
     \end{split}     
  \end{align*}  
  From this expression, a straightforward calculation yields (\ref{eq-E-frac-z-t}) by using the independence of \( X_0 \) and \( \{\overline{\xi}_s\}_{0 \leq s \leq t} \).

\end{proof}

Now we introduce discretized versions of $Z_t$ and $\{X_s\}_{0\leq s\leq t}$ by
\begin{align}
  \label{eq-def-hat-Z}\begin{split}    
  Z_t^{(n)}=&\exp\left( X_t^\top \hat{Y}_{t}- \sum_{i=1}^n \hat{Y}_{s_{i-1}}^\top (X_{s_i}-X_{s_{i-1}}) \right.\\
  &\left.-\frac{1}{2}\int_0^tX_s^\top c(s)^\top(\sigma(s)\sigma(s)^\top)^{-1} c(s)X_s ds\right),
  \end{split}
\end{align}
and
\begin{align}
  \label{eq-def-X-n}X_s^{(n)}=\sum_{i=1}^n X_{s_{i-1}}1_{[s_{i-1},s_{i})}(s)+X_{t}1_{\{t\}}(s)=X_{\frac{t}{n}\left[\frac{s}{t}n\right]}
\end{align}
where $\displaystyle s_i=s_i^n=\frac{i}{n}t$. 

Then the previous lemma gives the following results.

\begin{lemma}\label{lemma-3-2}
  For any fixed $u_1,\cdots,u_k \in [0,t]$, it holds that
  \begin{align}
    \label{eq-convergence-frac-Z}\frac{E_Q[X_{u_1}^{(n)}\otimes \cdots \otimes X_{u_{k}}^{(n)}Z_t^{(n)}|\mathcal{Y}_t]}{E_Q[Z_t^{(n)}|\mathcal{Y}_t]}  \xrightarrow{L^2} \frac{E_Q[X_{u_1}\otimes \cdots \otimes X_{u_{k}}Z_t|\mathcal{Y}_t]}{E_Q[Z_t|\mathcal{Y}_t]}
  \end{align}
  as $n\to \infty$.
\end{lemma}
\begin{proof}
  First, recall (\ref{def-hat-Y}), (\ref{eq-Z-Ito}), and (\ref{eq-def-zt-y}), allowing us to write \( Z_t = z_t(\hat{Y}) \). Then, as in the proof of Proposition \ref{prop-Z-martingale}, Lemma \ref{lemma-stochastic-integral-measurable-function} yields
\begin{align*}
  E_Q[Z_t | \mathcal{Y}_t] = E_Q[z_t(y)] \big|_{y = \hat{Y}}
\end{align*}
and
\begin{align*}
  E_Q[X_{u_1} \otimes \cdots \otimes X_{u_k} Z_t | \mathcal{Y}_t] = E_Q[X_{u_1} \otimes \cdots \otimes X_{u_k} z_t(y)] \big|_{y = \hat{Y}}.
\end{align*}
Furthermore, if we define
\begin{align*}
  \hat{Y}^{(n)}(s) = \sum_{i=1}^n \hat{Y}_{s_{i-1}} 1_{[s_{i-1}, s_i)}(s) + \hat{Y}_t 1_{\{t\}}(s),
\end{align*}
then we have \( Z_t^{(n)} = z_t(\hat{Y}^{(n)}) \), and it follows that
\begin{align*}
  E_Q[Z_t^{(n)} | \mathcal{Y}_t] = E_Q[z_t(y)] \big|_{y = \hat{Y}^{(n)}}
\end{align*}
and
\begin{align*}
  E_Q[X_{u_1}^{(n)} \otimes \cdots \otimes X_{u_k}^{(n)} Z_t^{(n)} | \mathcal{Y}_t] = E_Q[X_{u_1}^{(n)} \otimes \cdots \otimes X_{u_k}^{(n)} z_t(y)] \big|_{y = \hat{Y}^{(n)}}.
\end{align*}
Therefore, the desired result follows from (\ref{eq-E-frac-z-t}) and the convergence
\begin{align*}
  \sup_{0 \leq s \leq t} E_Q\left[\left| \hat{Y}_s - \hat{Y}_s^{(n)} \right|^p \right] \to 0 \quad (n \to \infty)
\end{align*}
for \( p \geq 1 \).

\end{proof}

\begin{lemma}\label{lemma3-3}
  For every \( s, u \in [0, t] \), we have
\begin{align}
  \label{eq-mean-X-under-tilde-E}
  \tilde{E}_t[X_s] = E[X_s] + \int_0^t \mathrm{Cov}(\xi_s, \xi_u) c(u)^\top (\sigma(u) \sigma(u)^\top)^{-1} (dY_u - c(u) E[X_u] \, du)
\end{align}
and
\begin{align}
  \label{eq-cov-X-under-tilde-E}
  \tilde{E}_t[X_s \otimes X_u] - \tilde{E}_t[X_s] \otimes \tilde{E}_t[X_u] = \mathrm{Cov}(\xi_s, \xi_u).
\end{align}
Furthermore, for any \( u_1, \cdots, u_k \in [0, t] \), the third and higher cumulants of \( X_{u_1}, \cdots, X_{u_k} \) under \( \tilde{P}_t \) are almost surely zero.
\end{lemma}
\begin{proof}
  First, let us write
  \begin{align*}
    &z_t^{(n)}(y_0,\cdots,y_n)\\
    =&\exp\left( X_t^\top y_{n}- \sum_{i=1}^n y_{i-1}^\top (X_{s_i}-X_{s_{i-1}})\right.\\
    &\left. -\frac{1}{2}\int_0^tX_s^\top c(s)^\top(\sigma(s)\sigma(s)^\top)^{-1} c(s)X_s ds\right),
  \end{align*}
  and consider the function
  \begin{align*}
    F(y_0,\cdots,y_{n})=&E_Q\left[z_t^{(n)}(y_0,\cdots,y_n)\right]
  \end{align*}
  for $y_0,\cdots,y_{n} \in \mathbb{R}^{d_1}$. Then it is easy to verify that
  \begin{align*}
    &\partial_{y_{i}}F(y_0,\cdots,y_{n})
    =-E_Q\left[(X_{s_{i+1}}-X_{s_{i}})z_t^{(n)}(y_0,\cdots,y_n)\right]~~(i=0,1,\cdots,n),
  \end{align*}
  where $X_{s_{n+1}}=0$ for convenience, and higher derivatives can be calculated in the same manner. 
  Therefore, we have 
  \begin{align}
    &\sum_{i,j=0}^{n}\partial_{y_i}\otimes \partial_{y_j}\log F(y_0,\cdots,y_n)1_{[0,s_{i+1})}(s)1_{[0,s_{j+1})}(u)\nonumber\\
    =&\sum_{i,j=0}^{n}\frac{E_Q[(X_{s_{i+1}}-X_{s_i})\otimes (X_{s_{j+1}}-X_{s_j})z_t^{(n)}(y_0,\cdots,y_n)]}{E_Q[z_t^{(n)}(y_0,\cdots,y_n)]}\nonumber\\
    &\times 1_{[0,s_{i+1})}(s)1_{[0,s_{j+1})}(u)\nonumber\\
    &-\sum_{i,j=0}^{n}\frac{E_Q[(X_{s_{i+1}}-X_{s_i})z_t^{(n)}(y_0,\cdots,y_n)]}{E_Q[z_t^{(n)}(y_0,\cdots,y_n)]}1_{[0,s_{i+1})}(s)\nonumber\\
    &\otimes \frac{ E_Q[(X_{s_{j+1}}-X_{s_j})z_t^{(n)}(y_0,\cdots,y_n)]}{E_Q[z_t^{(n)}(y_0,\cdots,y_n)]}1_{[0,s_{j+1})}(u)\nonumber \\
    \label{eq3-11}\begin{split}
      =&\frac{E_Q[X_s^{(n)}\otimes X_u^{(n)}z_t^{(n)}(y_0,\cdots,y_n)]}{E_Q[z_t^{(n)}(y_0,\cdots,y_n)]}\\
    &-\frac{E_Q[X_s^{(n)}z_t^{(n)}(y_0,\cdots,y_n)]\otimes E_Q[X_u^{(n)}z_t^{(n)}(y_0,\cdots,y_n)]}{E_Q[z_t^{(n)}(y_0,\cdots,y_n)]^2},
    \end{split}    
  \end{align}
  where $X^{(n)}$ is defined in (\ref{eq-def-X-n}).

  On the other hand, it follows from (\ref{eq-E-z-t}) that $\log F(y_1,\cdots,y_{n})$ is a quadratic form of $y_0,\cdots,y_{n}$, and thus $\partial_{y_i}\otimes \partial_{y_j}\log F(y_0,\cdots,y_n)$ does not depend on $y_0,\cdots,y_n$. Thus (\ref{eq3-11}) does not depend on $y_0,\cdots,y_n$, and we have
  \begin{align*}
    &\frac{E_Q[X_s^{(n)}\otimes X_u^{(n)}z_t^{(n)}(y_0,\cdots,y_n)]}{E_Q[z_t^{(n)}(y_0,\cdots,y_n)]}\\
    &-\frac{E_Q[X_s^{(n)}z_t^{(n)}(y_0,\cdots,y_n)]\otimes E_Q[X_u^{(n)}z_t^{(n)}(y_0,\cdots,y_n)]}{E_Q[z_t^{(n)}(y_0,\cdots,y_n)]^2}\\
    =&\frac{E_Q[X_s^{(n)}\otimes X_u^{(n)}z_t^{(n)}(0,\cdots,0)]}{E_Q[z_t^{(n)}(0,\cdots,0)]}\\
    &-\frac{E_Q[X_s^{(n)}z_t^{(n)}(0,\cdots,0)]\otimes E_Q[X_u^{(n)}z_t^{(n)}(0,\cdots,0)]}{E_Q[z_t^{(n)}(0,\cdots,0)]^2}.
  \end{align*}
  Substituting $(\hat{Y}_{s_0},\cdots,\hat{Y}_{s_{n}})$ into $(y_0,\cdots,y_n)$, we obtain
  \begin{align}
      &\frac{E_Q[X_s^{(n)}\otimes X_u^{(n)}Z_t^{(n)}|\mathcal{Y}_t]}{E_Q[Z_t^{(n)}|\mathcal{Y}_t]}
      -\frac{E_Q[X_s^{(n)}Z_t^{(n)}|\mathcal{Y}_t]\otimes E_Q[X_u^{(n)}Z_t^{(n)}|\mathcal{Y}_t]}{E_Q[Z_t^{(n)}|\mathcal{Y}_t]^2}\nonumber\\
    =&\frac{\displaystyle E_Q\left[X_s^{(n)}\otimes X_u^{(n)}\exp\left( -\frac{1}{2}\int_0^tX_s^\top c(s)^\top(\sigma(s)\sigma(s)^\top)^{-1} c(s)X_s ds\right)\right]}{\displaystyle E_Q\left[\exp\left( -\frac{1}{2}\int_0^tX_s^\top c(s)^\top(\sigma(s)\sigma(s)^\top)^{-1} c(s)X_s ds\right)\right]}\nonumber\\
    &-\frac{\displaystyle E_Q\left[X_s^{(n)}\exp\left( -\frac{1}{2}\int_0^tX_s^\top c(s)^\top(\sigma(s)\sigma(s)^\top)^{-1} c(s)X_s ds\right)\right]}{\displaystyle E_Q\left[\exp\left( -\frac{1}{2}\int_0^tX_s^\top c(s)^\top(\sigma(s)\sigma(s)^\top)^{-1} c(s)X_s ds\right)\right]}\nonumber\\
    \label{eq3-24}&\otimes \frac{\displaystyle E_Q\left[X_u^{(n)}\exp\left( -\frac{1}{2}\int_0^tX_s^\top c(s)^\top(\sigma(s)\sigma(s)^\top)^{-1} c(s)X_s ds\right)\right]}{\displaystyle E_Q\left[\exp\left( -\frac{1}{2}\int_0^tX_s^\top c(s)^\top(\sigma(s)\sigma(s)^\top)^{-1} c(s)X_s ds\right)\right]}
  \end{align}
  since $z_t^{(n)}(\hat{Y}_{s_0},\cdots,\hat{Y}_{s_{n}})=Z_t^{(n)}$, as introduced in (\ref{eq-def-hat-Z}).

  By Lemma \ref{lemma-3-2}, the left-hand side of this converges in $L^2$ to 
  \begin{align*}
    &\frac{E_Q[X_s\otimes X_uZ_t|\mathcal{Y}_t]}{E_Q[Z_t|\mathcal{Y}_t]}-\frac{E_Q[X_sZ_t|\mathcal{Y}_t]\otimes E_Q[X_uZ_t|\mathcal{Y}_t]}{E_Q[Z_t|\mathcal{Y}_t]^2}\\
    =&\tilde{E}_t[X_s\otimes X_u]-\tilde{E}_t[X_s]\otimes \tilde{E}_t[X_u]\nonumber
  \end{align*}
  as $n \to \infty$. On the other hand, the $L^2$-convergence of the right-hand side directly follows from the $L^2$-convergence of $X_s^{(n)}$. Therefore, we obtain
  \begin{align}
    &\tilde{E}_t[X_s\otimes X_u]-\tilde{E}_t[X_s]\otimes \tilde{E}_t[X_u]\nonumber\\
    =&\frac{\displaystyle E_Q\left[X_s\otimes X_u\exp\left( -\frac{1}{2}\int_0^tX_s^\top c(s)^\top(\sigma(s)\sigma(s)^\top)^{-1} c(s)X_s ds\right)\middle|\mathcal{Y}_t\right]}{\displaystyle E_Q\left[\exp\left( -\frac{1}{2}\int_0^tX_s^\top c(s)^\top(\sigma(s)\sigma(s)^\top)^{-1} c(s)X_s ds\right)\middle|\mathcal{Y}_t\right]}\nonumber\\
    &-\frac{\displaystyle E_Q\left[X_s\exp\left( -\frac{1}{2}\int_0^tX_s^\top c(s)^\top(\sigma(s)\sigma(s)^\top)^{-1} c(s)X_s ds\right)\middle|\mathcal{Y}_t\right]}{\displaystyle E_Q\left[\exp\left( -\frac{1}{2}\int_0^tX_s^\top c(s)^\top(\sigma(s)\sigma(s)^\top)^{-1} c(s)X_s ds\right)\middle|\mathcal{Y}_t\right]}\nonumber\\
    &\otimes \frac{\displaystyle E_Q\left[\exp\left( -\frac{1}{2}\int_0^tX_s^\top c(s)^\top(\sigma(s)\sigma(s)^\top)^{-1} c(s)X_s ds\right)\middle|\mathcal{Y}_t\right]}{\displaystyle E_Q\left[\exp\left( -\frac{1}{2}\int_0^tX_s^\top c(s)^\top(\sigma(s)\sigma(s)^\top)^{-1} c(s)X_s ds\right)\middle|\mathcal{Y}_t\right]}\nonumber\\
    \label{eq-cov-breve-xi}\begin{split}      
    =&\frac{\displaystyle E_Q\left[\breve{\xi}_{s;t}\otimes \breve{\xi}_{u;t}\exp\left( \frac{1}{2}X_0^\top \phi(0;t)X_0 \right)\right]}{\displaystyle E_Q\left[\exp\left( \frac{1}{2}X_0^\top \phi(0;t)X_0 \right)\right]}\\
    &-\frac{\displaystyle E_Q\left[\breve{\xi}_{s;t}\exp\left( \frac{1}{2}X_0^\top \phi(0;t)X_0 \right)\right]}{\displaystyle E_Q\left[\exp\left( \frac{1}{2}X_0^\top \phi(0;t)X_0 \right)\right]}
    \otimes \frac{\displaystyle E_Q\left[\breve{\xi}_{u;t}\exp\left( \frac{1}{2}X_0^\top \phi(0;t)X_0 \right)\right]}{\displaystyle E_Q\left[\exp\left( \frac{1}{2}X_0^\top \phi(0;t)X_0 \right)\right]}.
    \end{split}
  \end{align}
  Here, $\breve{\xi}$ is introduced in (\ref{eq-breve-xi}), and we used the density formula (\ref{eq-measure-change}) in the final equality.

  Now, recall the expression (\ref{eq-breve-xi-decompisition}):
  \begin{align*}
    \breve{\xi}_{s; t} = \alpha(s; t) X_0 + \overline{\xi}_{s; t},
  \end{align*} 
  where \( \overline{\xi}_{s; t} \) is independent of \( X_0 \) and \( E[\overline{\xi}_{s; t}] = 0 \). Using this expression, (\ref{eq-cov-breve-xi}) becomes equivalent to  
  \begin{align*}
    &\mathrm{Cov}(\overline{\xi}_{s;t},\overline{\xi}_{u;t})\\
    &+\frac{\displaystyle E\left[\alpha(s;t)X_0\otimes \alpha(u;t)X_0\exp\left( \frac{1}{2}X_0^\top \phi(0;t)X_0 \right)\right]}{\displaystyle E\left[\exp\left( \frac{1}{2}X_0^\top \phi(0;t)X_0 \right)\right]}\\
    &-\frac{\displaystyle E\left[\alpha(s;t)X_0\exp\left( \frac{1}{2}X_0^\top \phi(0;t)X_0 \right)\right]}{\displaystyle E\left[\exp\left( \frac{1}{2}X_0^\top \phi(0;t)X_0 \right)\right]}\\
    &\otimes \frac{\displaystyle E\left[\alpha(u;t)X_0\exp\left( \frac{1}{2}X_0^\top \phi(0;t)X_0 \right)\right]}{\displaystyle E\left[\exp\left( \frac{1}{2}X_0^\top \phi(0;t)X_0 \right)\right]}\\
    =&\mathrm{Cov}(\overline{\xi}_{s;t},\overline{\xi}_{u;t})+\alpha(s;t)\Sigma^\frac{1}{2}(I-\Sigma^\frac{1}{2}\phi(0;t)\Sigma^\frac{1}{2})^{-1}\Sigma^\frac{1}{2}\alpha(u;t)^\top\\
    =&\mathrm{Cov}(\xi_{s;t},\xi_{u;t}),
  \end{align*}
  Here, we used $V[X_0]=\Sigma$ and (\ref{eq-cov-xi0}). This implies (\ref{eq-cov-X-under-tilde-E}).

  In the same way, since the third and higher derivatives of $\log F(y_0,\cdots,y_n)$ are zero, we have
  \begin{align*}
    0=&\sum_{i_1,\cdots,i_k=0}^{n}\partial_{y_{i_1}}\otimes \cdots\otimes \partial_{y_{i_k}}\log F(\hat{Y}_{s_0},\cdots,\hat{Y}_{s_{n}})1_{[0,s_{i_1}]}(u_1)\cdots 1_{[0,s_{i_k}]}(u_k)
  \end{align*}
  for $k\geq 3$. The right-hand side converges in probability to the $k$-th cumulant of $X_{u_1},\cdots,X_{u_k}$, and thus it must be almost surely zero.

  Finally, to prove (\ref{eq-mean-X-under-tilde-E}), let us consider
  \begin{align}
     &\sum_{i=0}^{n}\partial_{y_i}\log F(\hat{Y}_{s_0},\cdots,\hat{Y}_{s_{n}})1_{[0,s_{i+1})}(s)\nonumber\\
     \label{eq3-10}\begin{split}
    =&-\sum_{i=0}^n \frac{E_Q[(X_{s_{i+1}}-X_{s_i})Z_t^{(n)}|\mathcal{Y}_t]}{E_Q[Z_t^{(n)}|\mathcal{Y}_t]}1_{[0,s_{i+1})}(s)
    =\frac{E_Q[X_s^{(n)}Z_t^{(n)}|\mathcal{Y}_t]}{E_Q[Z_t^{(n)}|\mathcal{Y}_t]},
    \end{split}    
  \end{align}
  and let $D$ be the Malliavin derivative with respect to the Brownian motion 
  \begin{align*}
    \overline{Y}_s=\int_0^s (\sigma(s)\sigma(s)^\top)^{-\frac{1}{2}}dY_s
  \end{align*}
  on $Q$. Since $\log F(y_0,\cdots,y_n)$ is a quadratic form, $\partial_{y_i}\log F(y_0,\cdots,y_n)$ is linear. Hence, by the definition of the Malliavin derivative of smooth random variables \citep{nualart2006malliavin} and $\displaystyle \hat{Y}_t=\int_0^t c(s)^\top (\sigma(s)\sigma(s)^\top)^{-\frac{1}{2}}d\overline{Y}_s,$ the Malliavin derivative of (\ref{eq3-10}) is
  \begin{align*}
    &D_u\frac{E_Q[X_s^{(n)}Z_t^{(n)}|\mathcal{Y}_t]}{E_Q[Z_t^{(n)}|\mathcal{Y}_t]}\\
    =&\sum_{i=0}^{n}(\sigma(u)\sigma(u)^\top)^{-\frac{1}{2}}c(u)\partial_{y_j}\partial_{y_i}\log F(\hat{Y}_{s_0},\cdots,\hat{Y}_{s_{n}})1_{[0,s_i]}(s)1_{[0,s_j]}(u)\\
    =&\frac{(\sigma(u)\sigma(u)^\top)^{-\frac{1}{2}}c(u)E[X_u^{(n)}\otimes X_s^{(n)}Z_t^{(n)}|\mathcal{Y}_t]}{E[Z_t^{(n)}|\mathcal{Y}_t]}\\
    &-\frac{(\sigma(u)\sigma(u)^\top)^{-\frac{1}{2}}c(u)E[X_u^{(n)}Z_t^{(n)}|\mathcal{Y}_t]\otimes E[X_s^{(n)}Z_t^{(n)}|\mathcal{Y}_t]}{E[Z_t^{(n)}|\mathcal{Y}_t]^2}.
  \end{align*}
  According to (\ref{eq3-24}), this is deterministic. Thus, by the Clerk-Ocone formula \citep{nualart2006malliavin}, we can write
  \begin{align*}
    &\frac{E_Q[X_s^{(n)}Z_t^{(n)}|\mathcal{Y}_t]}{E_Q[Z_t^{(n)}|\mathcal{Y}_t]}\\
    =&g^{(n)}(s;t)+\int_0^t\left( \frac{E[X_u^{(n)}\otimes X_s^{(n)}Z_t^{(n)}|\mathcal{Y}_t]}{E[Z_t^{(n)}|\mathcal{Y}_t]}\right.\\
    &\left.\qquad -\frac{E_Q[X_u^{(n)}Z_t^{(n)}|\mathcal{Y}_t]\otimes E_Q[X_s^{(n)}Z_t^{(n)}|\mathcal{Y}_t]}{E_Q[Z_t^{(n)}|\mathcal{Y}_t]^2} \right)^\top c(u)^\top (\sigma(u)\sigma(u)^\top)^{-1}dY_u,
  \end{align*}
  where $g^{(n)}(s;t)$ is a deterministic function. Due to Lemma \ref{lemma-3-2}, for every $s,t$, $g^{(n)}(s;t)$ converges, and
  \begin{align*}
    \left\{\frac{E_Q[X_u^{(n)}\otimes X_s^{(n)}Z_t^{(n)}|\mathcal{Y}_t]}{E_Q[Z_t^{(n)}|\mathcal{Y}_t]}
    -\frac{E_Q[X_u^{(n)}Z_t^{(n)}|\mathcal{Y}_t]\otimes E_Q[X_s^{(n)}Z_t^{(n)}|\mathcal{Y}_t]}{E_Q[Z_t^{(n)}|\mathcal{Y}_t]^2}\right\}_{0\leq u \leq t}
  \end{align*}
  converges in $L^2(\Omega;H)$ where $H=L^2([0,t])$. On the other hand, we have shown above that for each $s,u \in [0,t]$, this converges to $\mathrm{Cov}(\xi_u,\xi_s)$ in $L^2$. 

  Therefore, by letting $n\to \infty$, we obtain
  \begin{align*}
    \tilde{E}_t[X_s]&\left(=\frac{E_Q[X_sZ_t|\mathcal{Y}_t]}{E_Q[Z_t|\mathcal{Y}_t]} \right)\\
    &=g(s;t)+\int_0^t \mathrm{Cov}(\xi_s,\xi_u)c(u)^\top(\sigma(u)\sigma(u)^\top)^{-1}dY_u,
  \end{align*}
  where $g(s;t)$ is a deterministic function. By taking expectations of both sides under $P$, it follows that 
  \begin{align*}
    E[X_s]&(=E[E[X_s|\mathcal{Y}_t]]=E[\tilde{E}_t[X_s]])\\
    &=g(s;t)+\int_0^t \mathrm{Cov}(\xi_s,\xi_u)c(u)^\top(\sigma(u)\sigma(u)^\top)^{-1}c(u)E[X_u]du.
  \end{align*}
  This gives the expression of $g(s;t)$, which leads to the desired result.

\end{proof}

\begin{proof}[Proof of Theorem \ref{main-theorem-linear}]
  Let $\varpi_t:C_t\to \mathbb{R}^{d_1+d_2}$ be the projection where $\varpi_t(w)=w(t)$, and $\mathcal{P}$ be the class of sets of the form 
  \begin{align*}
    \varpi_{u_1}^{-1}(A_1\times B_1)\cap \cdots \cap \varpi_{u_k}^{-1}(A_k\times B_k),
  \end{align*}
  where $A_1,\cdots,A_k \in \mathcal{B}(\mathbb{R}^{d_1}),B_1,\cdots,B_k \in \mathcal{B}(\mathbb{R}^{d_2})$ and $u_1,\cdots,u_k \in [0,t]$. 
  
  Furthermore, let us define $\mathcal{D}$ by
  \begin{align*}
    \mathcal{D}=\{A \in \mathcal{C}_t;\tilde{P}_t((X.,Y.)\in A)=P((\xi.,Y.) \in A|\mathcal{Y}_t)~~\mathrm{a.s.}\}.    
  \end{align*}
  Due to (\ref{eq-mean-zeta}), (\ref{eq-cov-zeta}) and Lemma \ref{lemma3-3}, it follows that for fixed $u_1,\cdots,u_k \in [0,t]$, $(X_{u_1},\cdots,X_{u_k})$ under $\tilde{P}_t$ and $(\zeta_{u_1;t},\cdots,\zeta_{u_k;t})$ under $P(\,\cdot\,|\mathcal{Y}_t)$ have the almost surely same mean and covariance, and their third and higher cumulants are almost surely zero.
  
  Therefore, for fixed $u_1,\cdots,u_k \in [0,t]$, there exists a null set $N$ such that the distributions of $(X_{u_1},\cdots,X_{u_k})$ under $\tilde{P}_t(\,\cdot\,)(\omega)$ and $(\zeta_{u_1;t},\cdots,\zeta_{u_k;t})$ under $P(\,\cdot\,|\mathcal{Y}_t)(\omega)$ are the same for $\omega \in \Omega \backslash N$. Thus we have for $A_1,\cdots,A_k \in \mathcal{B}(\mathbb{R}^{d_1})$ and $B_1,\cdots,B_k \in \mathcal{B}(\mathbb{R}^{d_2})$
  \begin{align*}
    &\tilde{P}_t(X_{u_1}\in A_1,\cdots X_{u_k}\in A_k,Y_{u_1}\in B_1,\cdots, Y_{u_k}\in B_k)\\
    =&\tilde{P}_t(X_{u_1}\in A_1,\cdots X_{u_k}\in A_k)1_{\{Y_{u_1}\in B_1,\cdots, Y_{u_k}\in B_k\}}\\
    =&P(\zeta_{u_1,t}\in A_1,\cdots \zeta_{u_k,t}\in A_k|\mathcal{Y}_t)1_{\{Y_{u_1}\in B_1,\cdots, Y_{u_k}\in B_k\}}\\
    =&P(\zeta_{u_1,t}\in A_1,\cdots \zeta_{u_k,t}\in A_k,Y_{u_1}\in B_1,\cdots, Y_{u_k}\in B_k|\mathcal{Y}_t).
  \end{align*}
  This means $\mathcal{P}\subset \mathcal{D}$, and we obtain the the desired result by the $\pi$-$\lambda$ theorem.

\end{proof}
\subsection{Filtering and smoothing equations}\label{section-linear-equations}
In the previous section, we showed that \( \{X_s\}_{0 \leq s \leq t} \) is conditionally Gaussian given \( \mathcal{Y}_t \), and we obtained explicit formulas for its mean and covariance functions. In this section, we derive differential equations for these functions.

First, we consider the differential equation for \( \mathrm{Cov}(\xi_{t; t}, \xi_{t; t}) \), which corresponds to the covariance matrix of \( X_t \) given \( \{\mathcal{Y}_t\} \). To this end, we define \( \gamma(t) \) as the positive-semidefinite solution of the Riccati equation
\begin{align}
  \label{eq-gamma}
  \begin{split}
    \frac{d}{dt} \gamma(t) &= -\gamma(t) c(t)^\top (\sigma(t) \sigma(t)^\top)^{-1} c(t) \gamma(t) + a(t) \gamma(t) + \gamma(t) a(t)^\top + b(t) b(t)^\top
  \end{split}
\end{align}
with \( \gamma(0) = V[X_0] \). The existence of the solution is guaranteed by Theorems 2.1 and 2.2 in \citet{potter1965matrix}.

To avoid potential singularity issues with \( \gamma(t) \), we introduce \( \gamma^\epsilon(t) \) and \( \xi_{s; t}^\epsilon \) as solutions to the same equations for \( \gamma(t) \) and \( \xi_{s; t} \), where \( V[X_0] \) is replaced by \( V[X_0] + \epsilon I_{d_1} \) for \( \epsilon > 0 \). Specifically, \( \gamma^\epsilon(t) \) and \( \xi_{s; t}^\epsilon \) are given as solutions to the equations
\begin{align}
  \label{eq-def-gamma-epsilon-1}
  \begin{split}    
    \frac{d}{dt} \gamma^\epsilon(t) &= -\gamma^\epsilon(t) c(t)^\top (\sigma(t) \sigma(t)^\top)^{-1} c(t) \gamma^\epsilon(t) + a(t) \gamma^\epsilon(t) + \gamma^\epsilon(t) a(t)^\top \\
    &\quad + b(t) b(t)^\top
  \end{split}
\end{align}
and
\begin{align}
  \label{eq-def-xi-epsilon-1}
  d \xi_{s; t}^\epsilon = \left\{ a(s) + b(s) b(s)^\top \phi(s; t) \right\} \xi_{s; t}^\epsilon \, ds + b(s) \, d\tilde{V}_s,
\end{align}
where
\begin{align}
  \label{eq-def-gamma-epsilon-2}
  \gamma^\epsilon(0) = V[X_0] + \epsilon I_{d_1},
\end{align}
and \( \xi_{0; t}^\epsilon \) is a Gaussian random variable with mean zero and covariance
\begin{align}
  \label{eq-def-xi-epsilon-2}
  V[\xi_{0; t}^\epsilon] = \left\{(V[X_0] + \epsilon I_{d_1})^{-1} - \phi(0; t) \right\}^{-1}.
\end{align}
By Corollary 1 to Theorem 2.1 in \citet{potter1965matrix}, \( \gamma^\epsilon(t) \) is strictly positive.

We can also establish the following lemma.

\begin{lemma}\label{lemma-gamma-epsilon}
  (1) For every $t\geq0$, $\gamma^\epsilon(s)$ is bounded for $\epsilon \in [0,1]$ and $s \in [0,t]$.\\
  (2) The matrix $\gamma^\epsilon(t)$ converges uniformly to $\gamma(t)$ on any compact set as $\epsilon \to +0$.
\end{lemma}
\begin{proof}
  (1) Take an arbitrary $x \in \mathbb{R}^{d_1}$. Then (\ref{eq-gamma}) yields
  \begin{align*}
    \frac{d}{dt}x^\top \gamma^\epsilon(t) x=&-x^\top\gamma^\epsilon(t)c(t)^\top(\sigma(t)\sigma(t)^\top)^{-1}c(t)\gamma^\epsilon(t)x\\
    &+x^\top a(t) \gamma^\epsilon(t)x+x^\top \gamma^\epsilon(t)a(t)^\top x+x^\top b(t)b(t)^\top x\\
    \leq & 2|x^\top a(t) \gamma^\epsilon(t)^\frac{1}{2} \gamma^\epsilon(t)^\frac{1}{2}x|+x^\top b(t)b(t)^\top x\\
    \leq &2\sqrt{x^\top a(t) \gamma^\epsilon(t)a(t)^\top  x}\sqrt{x^\top  \gamma^\epsilon(t) x}+x^\top b(t)b(t)^\top x\\
    \leq &2\|a(t)\|_2 \|\gamma^\epsilon(t)\|_2|x|^2+\|b(t)\|_2^2 |x|^2,
  \end{align*}
  where $\|\cdot\|_2$ is the matrix 2-norm, and we used
  \begin{align*}
    x^\top a(t) \gamma^\epsilon(t)a(t)^\top x&=|\gamma^\epsilon(t)^\frac{1}{2}a(t)^\top x|^2\leq \|\gamma^\epsilon(t)\|_2\|a(t)\|_2^2|x|^2.
  \end{align*}
  From this, it follows that
  \begin{align*}
    \|\gamma^\epsilon(t)\|_2=\max_{|x|=1}x^\top \gamma^\epsilon(t) x\leq \|\gamma^\epsilon(0)\|_2+\int_0^t \{2\|a(s)\|_2\|\gamma^\epsilon(s)\|_2+\|b(s)\|_2^2\}ds.
  \end{align*}
  Therefore, by the Gronwall's lemma, we obtain
  \begin{align*}
    \|\gamma^\epsilon(t)\|_2\leq & \|\gamma^\epsilon(0)\|_2\exp\left( 2\int_0^t\|a(r)\|_2dr \right)\\
    &+\int_0^t \exp\left( 2\int_s^t\|a(r)\|_2dr \right)\|b(s)\|_2^2ds,
  \end{align*}
  which leads to the desired result.\\
  Statement (2) can be proven routinely using the result of (1) along with Gronwall's lemma.
\end{proof}

\begin{proposition}\label{prop-gamma}
  It holds that
  \begin{align*}
    \gamma(t)=V[\xi_{t,t}].
  \end{align*}
\end{proposition}
\begin{proof}
  First, assume that $\gamma(t)$ is strictly positive for every $t\geq 0$. Then we have 
  \begin{align*}
    \frac{d}{ds}\left( \gamma(s)^{-1} \right)
    =&-\gamma(s)^{-1} \frac{d}{ds}\gamma(s) \gamma(s)^{-1}\\
    =&-\gamma(s)^{-1}b(s)b(s)^\top \gamma(s)^{-1}-a(s)^\top \gamma(s)^{-1}-\gamma(s)^{-1}a(s)\\
    &+c(s)^\top (\sigma(s)\sigma(s)^\top)^{-1}c(s).
  \end{align*}
  This means that $\gamma(s)^{-1}$ satisfies the same equation (\ref{eq-def-phi}) as $\phi(s;t)$. Hence, if we set $z(s;t)=\gamma(s)^{-1}-\phi(s;t)$, then it follows that
  \begin{align*}
    &\frac{d}{ds}z(s;t)\\
    =&-\left\{ \gamma(s)^{-1}b(s)b(s)^\top \gamma(s)^{-1}
    -\phi(s;t)^{-1}b(s)b(s)^\top \phi(s;t)^{-1} \right\}\\
    &-a(s)^\top \left\{ \gamma(s)^{-1}-\phi(s;t) \right\} -\left\{ \gamma(s)^{-1}-\phi(s;t) \right\}a(s)\\
    =&-z(s;t)b(s)b(s)^\top z(s;t)\\
    &-\left\{ a(s)^\top+\phi(s;t)b(s)b(s)^\top \right\}z(s;t)-z(s;t)\left\{ a(s)+b(s)b(s)^\top \phi(s;t) \right\}.
  \end{align*}
  Furthermore, if we set $w(s;t)=z(s;t)^{-1}$ (this is well-defined since $\psi(s;t)$ is negative-semidefinite), we have
  \begin{align}
    \label{eq-derivative-w}\begin{split}
      \frac{d}{ds}w(s;t)=&-z(s;t)^{-1} \frac{d}{ds}z(s;t) z(s;t)^{-1}\\
    =&b(s)b(s)^\top+w(s;t)\left\{ a(s)^\top+\phi(s;t)b(s)b(s)^\top \right\}\\
    &+\left\{ a(s)+b(s)b(s)^\top\phi(s;t) \right\}w(s;t)
    \end{split}
  \end{align}
  Therefore, $w(s;t)$ can be written as
  \begin{align}
    \label{eq-expression-w}\begin{split}
      w(s;t)=&\exp\left( \int_0^s\left\{ a(r)+b(r)b(r)^\top\phi(r;t) \right\}dr \right)V[\xi_{0;t}]\\
      &\times \exp\left( \int_0^s\left\{ a(r)^\top+\phi(r;t)b(r)b(r)^\top \right\}dr \right)\\
      &+\int_0^s \exp\left( \int_u^s\left\{ a(r)+b(r)b(r)^\top \phi(r;t) \right\}dr \right)\\
      &\qquad \times b(u) b(u)^\top\exp\left( \int_u^s\left\{ a(r)^\top+\phi(r;t)b(r)b(r)^\top \right\}dr \right)du.
    \end{split}   
  \end{align}
  Here, we used $\gamma(0)=V[X_0]$ and (\ref{eq-cov-xi0}) to establish 
  \begin{align*}
    w(0;s)=(V[X_0]^{-1}-\phi(0;t))^{-1}=V[\xi_{0;t}]. 
  \end{align*}

  On the other hand, since \( \{\xi_{s; t}\}_{0 \leq s \leq t} \) is the solution of (\ref{eq-def-xi}), it can be expressed as
  \begin{align}
    \label{eq-xi-formula}
    \begin{split}
      \xi_{s; t} &= \exp\left( \int_0^s \left\{ a(r) + b(r) b(r)^\top \phi(r; t) \right\} \, dr \right) \xi_{0; t} \\
      &\quad + \int_0^s \exp\left( \int_u^s \left\{ a(r) + b(r) b(r)^\top \phi(r; t) \right\} \, dr \right) b(u) \, d\tilde{V}_u.
    \end{split}
  \end{align}
  Hence, it follows immediately from (\ref{eq-expression-w}) that \( V[\xi_{s; t}] = w(s; t) \). In particular, since
  \begin{align}
    \label{eq-w-inverse-formula}
    w(s; t) = \{\gamma(s)^{-1} - \phi(s; t)\}^{-1}
  \end{align}
  and \( \phi(t; t) = 0 \), we have \( w(t; t) = \gamma(t) \), which gives the desired result.
  
  In the case where \( \gamma(t) \) becomes singular for some \( t \), we replace \( V[X_0] \) with \( V[X_0] + \epsilon I_{d_1} \) for \( \epsilon > 0 \), yielding
  \begin{align*}
    \gamma^\epsilon(t) = V[\xi_{t; t}^\epsilon],
  \end{align*}
  since \( \gamma^\epsilon(t) \) is strictly positive. We obtain the conclusion by letting \( \epsilon \to 0 \).  
\end{proof}

\begin{proposition}\label{prop-expression-cov-xi}
  Let $0\leq u \leq s\leq t$. If $\gamma(r)$ is non-singular for every $r\in [0,t]$, it holds that 
  \begin{align*}
    \mathrm{Cov}(\xi_{s;t},\xi_{u;t})=V[\xi_{s;t}]\exp\left( -\int_u^s \left\{ a(r)^\top+\gamma(r)^{-1}b(r)b(r)^\top \right\}dr \right).
  \end{align*}
\end{proposition}
\begin{proof}
  Since \( \xi_{s; t} \) is given by the explicit formula (\ref{eq-xi-formula}), we can express
  \begin{align}
    &\mathrm{Cov}(\xi_{s;t},\xi_{u;t})\nonumber\\
    =&\exp\left( \int_0^s\left\{ a(r)+b(r)b(r)^\top\phi(r;t) \right\}dr \right)V[\xi_{0;t}]\nonumber\\
    &\exp\left( \int_0^u\left\{ a(r)^\top+\phi(r;t)b(r)b(r)^\top \right\}dr \right) \nonumber\\
    &+\int_0^u \exp\left( \int_v^s\left\{ a(r)+b(r)b(r)^\top \phi(r;t) \right\}dr \right)\nonumber\\
    &\qquad \times b(v) b(v)^\top \exp\left( \int_v^u\left\{ a(r)^\top+\phi(r;t)b(r)b(r)^\top \right\}dr \right)dv\nonumber\\
    \label{eq-cov-xi-2}=&\exp\left( \int_u^s\left\{ a(r)+b(r)b(r)^\top \phi(r;t)\right\}dr \right)w(u;t),
  \end{align} 
  where $w(u;t)=V[\xi_{u;t}]$ is given by (\ref{eq-expression-w}). Thus, it follows from (\ref{eq-derivative-w}) and (\ref{eq-w-inverse-formula}) that
  \begin{align*}
    &d_u\mathrm{Cov}(\xi_{s;t},\xi_{u;t})\\
    =&\exp\left( \int_u^s\left\{ a(r)+b(r)b(r)^\top \phi(r;t)\right\}dr \right)\\
    &\times \left[ -\left\{ a(u)+b(u)b(u)^\top \phi(u;t)\right\}w(u;t)+\frac{\partial}{\partial u}w(u;t) \right]du\\
    =&\exp\left( \int_u^s\left\{ a(r)+b(r)b(r)^\top \phi(r;t)\right\}dr \right)\\
    &\times \left[ -\left\{ a(u)+b(u)b(u)^\top \phi(u;t)\right\}w(u;t)+b(u)b(u)^\top\right.\\
    &+w(u;t)\left\{ a(u)^\top+\phi(u;t)b(u)b(u)^\top \right\}\\
    &+\left.\left\{ a(u)+b(u)b(u)^\top\phi(u;t) \right\}w(u;t)\right]du\\
    =&\exp\left( \int_u^s\left\{ a(r)+b(r)b(r)^\top \phi(r;t)\right\}dr \right)w(u,t)\\
    &\times \left\{ w(u;t)^{-1}b(u)b(u)^\top+a(u)^\top+\phi(u;t)b(u)b(u)^\top \right\}du\\
    =&\mathrm{Cov}(\xi_{s;t},\xi_{u;t}) \left[ \left\{ \gamma(u)^{-1}-\phi(u;t) \right\}b(u)b(u)^\top+a(u)^\top\right.\\
    &\left.\qquad\qquad\qquad\qquad\qquad\qquad\qquad+\phi(u;t)b(u)b(u)^\top \right]du\\
    =&\mathrm{Cov}(\xi_{s;t},\xi_{u;t})\left\{ a(u)^\top+\gamma(u)^{-1}b(u)b(u)^\top \right\}du,
  \end{align*}
  which implies the conclusion. 
\end{proof}

\begin{proposition}\label{prop-equation-xi-t-s}
  For every $s\geq0$, it holds that 
  \begin{align*}
    d_t\mathrm{Cov}(\xi_{t;t},\xi_{s;t})=\left\{ a(t)-\gamma(t)c(t)^\top(\sigma(t)\sigma(t)^\top)^{-1}c(t)\right\}\mathrm{Cov}(\xi_{t;t},\xi_{s;t})dt.
  \end{align*}
\end{proposition}
\begin{proof}
  First, assume that $\gamma(t)$ is non-singular for every $t\geq 0$. In this case, by Propositions \ref{prop-gamma} and \ref{prop-expression-cov-xi}, we can write
  \begin{align*}
    \mathrm{Cov}(\xi_{t;t},\xi_{s;t})=\gamma(t)\exp\left( -\int_s^t \left\{ a(r)^\top+\gamma(r)^{-1}b(r)b(r)^\top \right\}dr \right).
  \end{align*}
  Hence, it holds
  \begin{align*}
    &d_t\mathrm{Cov}(\xi_{t;t},\xi_{s;t})\\
    =&\left[ \frac{d}{dt}\gamma (t)-\gamma(t)\{a(t)^\top+\gamma(t)^{-1}b(t)b(t)^\top\}\right]\\
    &\times \exp\left( -\int_s^t \left\{ a(r)^\top+\gamma(r)^{-1}b(r)b(r)^\top \right\}dr \right)dt\\
    =&\left\{ -\gamma(t)c(t)^\top(\sigma(t)\sigma(t)^\top)^{-1}c(t)\gamma(t)+a(t)\gamma(t)\right\}\\
    &\times \exp\left( -\int_s^t \left\{ a(r)^\top+\gamma(r)^{-1}b(r)b(r)^\top \right\}dr \right)dt\\
    =&\left\{ -\gamma(t)c(t)^\top(\sigma(t)\sigma(t)^\top)^{-1}c(t)+a(t)\right\}\mathrm{Cov}(\xi_{t;t},\xi_{s;t})dt,
  \end{align*}
  which implies the conclusion.

  For the case where \( \gamma(t) \) is singular, we have
  \begin{align*}
    d_t\mathrm{Cov}(\xi_{t;t}^\epsilon,\xi_{s;t}^\epsilon)
    =\left\{ -\gamma^\epsilon(t)c(t)^\top(\sigma(t)\sigma(t)^\top)^{-1}c(t)+a(t)\right\}\mathrm{Cov}(\xi_{t;t}^\epsilon,\xi_{s;t}^\epsilon)dt
  \end{align*}
  and thus 
  \begin{align}
    \label{eq-cov-xi-epsilon}\mathrm{Cov}(\xi_{t;t}^\epsilon,\xi_{s;t}^\epsilon)=\exp\left( \int_s^t \left\{ a(r)-\gamma^\epsilon(r)c(r)^\top(\sigma(r)\sigma(r)^\top)^{-1}c(r)\right\}dr \right)\gamma^\epsilon(s).
  \end{align}
  By the uniform convergence of \( \gamma^\epsilon \) on any compact set, we can take the limit \( \epsilon \to 0 \) to obtain
\begin{align*}
  \mathrm{Cov}(\xi_{t; t}, \xi_{s; t}) 
  = \exp\left( \int_s^t \left\{ a(r) - \gamma(r) c(r)^\top (\sigma(r) \sigma(r)^\top)^{-1} c(r) \right\} \, dr \right) \gamma(s),
\end{align*}
which leads to the desired result.
\end{proof}

Now, let us define
\begin{align}
  \label{eq-def-mu}
  \mu_{s; t} = E[X_s] + \int_0^t \mathrm{Cov}(\xi_{s; t}, \xi_{u; t}) c(u)^\top (\sigma(u) \sigma(u)^\top)^{-1} (dY_u - c(u) E[X_u] \, du).
\end{align}
Recalling Theorem \ref{main-theorem-linear}, we find that for every \( 0 \leq s \leq t \),
\begin{align*}
  \mu_{s; t} = E[\zeta_{s; t} | \mathcal{Y}_t] = E[X_s | \mathcal{Y}_t] \quad \mathrm{a.s.}
\end{align*}
The previous proposition provides the stochastic differential equation for \( \mu_{t, t} \), which is known as the Kalman-Bucy filter.

\begin{proposition}{\bf (Kalman-Bucy filter)}\\
  The process $\{\mu_{t,t}\}_{t\geq 0}$ satisfies the stochastic differential equation
  \begin{align*}
    d\mu_{t,t}=a(t)\mu_{t,t}dt+\gamma(t) c(t)^\top(\sigma(t)\sigma(t)^\top)^{-1}\{dY_t-c(t)\mu_{t,t}dt\}.
  \end{align*} 
\end{proposition}
\begin{proof}
  Set
  \begin{align}
    \label{eq3-12}\tilde{\mu}_{t;t}=\mu_{t;t}-E[X_t],~~\tilde{Y}_t=Y_t-\int_0^t c(u)E[X_u]du.
  \end{align}
  Then according to Proposition \ref{prop-equation-xi-t-s} and (\ref{eq-def-mu}), we have
  \begin{align*}
    \tilde{\mu}_{t,t}=&\int_0^t \mathrm{Cov}(\xi_{t;t},\xi_{u;t})c(u)^\top(\sigma(u)\sigma(u)^\top)^{-1}d\tilde{Y}_u\\
    =&\int_0^t\left[\gamma(u)+
    \int_u^t\left\{ -\gamma(s)c(s)^\top(\sigma(s)\sigma(s)^\top)^{-1}c(s)+a(s)\right\}\mathrm{Cov}(\xi_{s;s},\xi_{u;s})ds\right]\\
    &\times c(u)^\top(\sigma(u)\sigma(u)^\top)^{-1}d\tilde{Y}_u\\
    =&\int_0^t a(s)\int_0^s \mathrm{Cov}(\xi_{s;s},\xi_{u;s})c(u)^\top(\sigma(u)\sigma(u)^\top)^{-1}d\tilde{Y}_uds\\
    &+\int_0^t\gamma(u) c(u)^\top(\sigma(u)\sigma(u)^\top)^{-1}d\tilde{Y}_u\\
    &-\int_0^t\gamma(s)c(s)^\top(\sigma(s)\sigma(s)^\top)^{-1}c(s)\\
    &\times \int_0^s \mathrm{Cov}(\xi_{s;s},\xi_{u;s})c(u)^\top(\sigma(u)\sigma(u)^\top)^{-1}d\tilde{Y}_uds\\
    =&\int_0^t a(s)\tilde{\mu}_{s,s}ds+\int_0^t\gamma(u) c(u)^\top(\sigma(u)\sigma(u)^\top)^{-1}\{d\tilde{Y}_u-c(s)\tilde{\mu}_{s,s}ds\},
  \end{align*}
  which gives the equation
  \begin{align*}
    d\tilde{\mu}_{t,t}=a(t)\tilde{\mu}_{t,t}dt+\gamma(t) c(t)^\top(\sigma(t)\sigma(t)^\top)^{-1}\{d\tilde{Y}_t-c(t)\tilde{\mu}_{t,t}dt\}.
  \end{align*}
  By (\ref{eq3-12}) and $d E[X_t]=a(t)E[X_t]dt$, this is equivalent to the desired result.
\end{proof}

We can also derive the forward-backward algorithm for fixed-interval smoothing, commonly referred to as the Rauch-Tung-Striebel smoother.
\begin{proposition}{\bf (Rauch-Tung-Striebel smoother)}\\
  If $\gamma(r)$ is strictly positive for every $r \in [0,t]$, it holds
  \begin{align*}
    \mu_{s;t}=&\mu_{t;t}-\int_s^t \left\{ a(r)\mu_{r;t}+b(r)b(r)^\top \gamma(r)^{-1}(\mu_{r;t}-\mu_{r;r}) \right\}dr
  \end{align*}
  for every $0\leq s\leq t$.
\end{proposition}
\begin{proof}
  Set 
  \begin{align*}
    \tilde{\mu}_{t;t}=\mu_{t;t}-E[X_t],~~\hat{Y}_t=\int_0^t c(u)^\top (\sigma(u) \sigma(u)^\top)^{-1} (dY_u - c(u) E[X_u] \, du).
  \end{align*}
  According to Proposition \ref{prop-expression-cov-xi}, for $s\leq u$ we have
  \begin{align*}
    d_s\mathrm{Cov}(\xi_{s;t},\xi_{u;t})
    =\left\{ a(s)+b(s)b(s)^\top \gamma(s)^{-1} \right\}\mathrm{Cov}(\xi_{s;t},\xi_{u;t})ds
  \end{align*}
  On the other hand, it follows from (\ref{eq-cov-xi-2}) that for $s\geq u$
  \begin{align}
    \label{eq-ds-cov-xi}d_s\mathrm{Cov}(\xi_{s;t},\xi_{u;t})=\left\{ a(s)+b(s)b(s)^\top \phi(s;t) \right\}\mathrm{Cov}(\xi_{s;t},\xi_{u;t})ds.
  \end{align}
  Therefore, it holds
  \begin{align*}
    \tilde{\mu}_{s;t}=&\int_0^t\left\{ \mathrm{Cov}(\xi_{t;t},\xi_{u;t})-\int_s^t \frac{\partial}{\partial r}\mathrm{Cov}(\xi_{r;t},\xi_{u;t})dr \right\} d\hat{Y}_u\\
    =&\tilde{\mu}_{t;t}-\int_s^t \int_0^t \frac{\partial}{\partial r}\mathrm{Cov}(\xi_{r;t},\xi_{u;t})d\hat{Y}_u dr\\
    =&\tilde{\mu}_{t;t}-\int_s^t \int_0^r \left\{ a(r)+b(r)b(r)^\top \phi(r;t) \right\}\mathrm{Cov}(\xi_{r;t},\xi_{u;t})d\hat{Y}_u dr\\
    &-\int_s^t \int_r^t \left\{ a(r)+b(r)b(r)^\top \gamma(r)^{-1} \right\}\mathrm{Cov}(\xi_{r;t},\xi_{u;t})d\hat{Y}_u dr\\
    =&\tilde{\mu}_{t;t}-\int_s^t \left\{ a(r)+b(r)b(r)^\top \gamma(r)^{-1} \right\}\int_0^t \mathrm{Cov}(\xi_{r;t},\xi_{u;t})d\hat{Y}_u dr\\
    &+\int_s^t \int_0^r b(r)b(r)^\top \{\gamma(r)^{-1}-\phi(r;t)\} \mathrm{Cov}(\xi_{r;t},\xi_{u;t})d\hat{Y}_u dr\\
    =&\tilde{\mu}_{t;t}-\int_s^t \left\{ a(r)+b(r)b(r)^\top \gamma(r)^{-1} \right\}\tilde{\mu}_{r;t} dr\\
    &+\int_s^t b(r)b(r)^\top \int_0^r  V[\xi_{r;t}]^{-1} \mathrm{Cov}(\xi_{r;t},\xi_{u;t})d\hat{Y}_u dr.
  \end{align*}
  In the last equation, we utilized the relation (\ref{eq-w-inverse-formula}).

  Furthermore, according to Proposition \ref{prop-expression-cov-xi}, $V[\xi_{r;t}]^{-1} \mathrm{Cov}(\xi_{r;t},\xi_{u;t})$ does not depend on $t$. Hence, we have
  \begin{align*}
    \int_0^r  V[\xi_{r;t}]^{-1} \mathrm{Cov}(\xi_{r;t},\xi_{u;t})d\hat{Y}_u &=\int_0^r  V[\xi_{r;r}]^{-1} \mathrm{Cov}(\xi_{r;r},\xi_{u;r})d\hat{Y}_u\\
    &=\gamma(r)^{-1}\tilde{\mu}_{r,r}.
  \end{align*}

  Therefore, we obtain
  \begin{align*}
    \tilde{\mu}_{s;t}=&\tilde{\mu}_{t;t}-\int_s^t \left\{ a(r)+b(r)b(r)^\top \gamma(r)^{-1} \right\}\tilde{\mu}_{r;t} dr\\
    &+\int_s^t b(r)b(r)^\top \gamma(r)^{-1}\tilde{\mu}_{r,r} dr.
  \end{align*}

  Finally, we obtain the desired result by using $\tilde{\mu}_{t;t}=\mu_{t;t}-E[X_t]$ and $dE[X_t]=a(t)E[X_t]dt$.  
\end{proof}

Next, we derive the differential equations for $\mathrm{Cov}(\xi_{s;t},\xi_{u;t})$ and $\mu_{s;t}$ with respect to $t$.

\begin{proposition}\label{prop-derivative-cov-xi-s-u-t}
  It holds for $s,u\geq0$ and $t\geq s\vee u$ that
  \begin{align*}
    d_t\mathrm{Cov}(\xi_{s;t},\xi_{u;t})=-\mathrm{Cov}(\xi_{t;t},\xi_{s;t})^\top c(t)^\top(\sigma(t)\sigma(t)^\top)^{-1}c(t)
    \mathrm{Cov}(\xi_{t;t},\xi_{u;t})dt.
  \end{align*}
\end{proposition}
\begin{proof}
  First assume that $\gamma(t)$ is positive for every $t\geq 0$. By exchanging $\gamma(s)^{-1}$ and $\phi(s;t)$ in (\ref{eq-derivative-w}), and recalling $w(s;t)=V[\xi_{s;t}]$, it holds  
  \begin{align*}
    \frac{d}{ds}V[\xi_{s;t}]=&-b(s)b(s)^\top+V[\xi_{s;t}]\left\{ a(s)^\top + \gamma(s)^{-1}b(s)b(s)^\top \right\}\\
    &+\left\{ a(s)^\top + \gamma(s)^{-1}b(s)b(s)^\top \right\}V[\xi_{s;t}].
  \end{align*}
  Thus, we obtain for $s\leq t$
  \begin{align*}
    V[\xi_{s;t}]=&\exp\left(-\int_s^t\left\{ a(r)+b(r)b(r)^\top\gamma(r)^{-1} \right\}dr \right)\gamma(t)\\
    &\times \exp\left(-\int_s^t\left\{ a(r)^\top+\gamma(r)^{-1}b(r)b(r)^\top \right\}dr \right)\\
    &+\int_s^t \exp\left( \int_u^s\left\{ a(r)+b(r)b(r)^\top \gamma(r)^{-1} \right\}dr \right)\\
    &\qquad \times b(u) b(u)^\top\exp\left( \int_u^s\left\{ a(r)^\top+\gamma(r)^{-1}b(r)b(r)^\top \right\}dr \right)du,
  \end{align*}
  where we used $V[\xi_{t;t}]=\gamma(t)$. Differentiating this yields
  \begin{align*}
    &d_tV[\xi_{s;t}]\\
    =&\exp\left(-\int_s^t\left\{ a(r)+b(r)b(r)^\top\gamma(r)^{-1} \right\}dr \right)\\
    &\times \biggl[ -\left\{a(t)+b(t)b(t)^\top\gamma(t)^{-1} \right\}\gamma(t) \\
    &\qquad -\gamma(t)\left\{ a(t)^\top+\gamma(t)^{-1}b(t)b(t)^\top\right\}+\frac{d}{dt}\gamma(t)+b(t)b(t)^\top \biggr]\\
    &\times \exp\left(-\int_s^t\left\{ a(r)^\top+\gamma(r)^{-1}b(r)b(r)^\top \right\}dr \right)dt\\
    =&-\exp\left(-\int_s^t\left\{ a(r)+b(r)b(r)^\top\gamma(r)^{-1} \right\}dr \right)\\
    &\times \gamma(t)c(t)^\top(\sigma(t)\sigma(t)^\top)^{-1}c(t)\gamma(t)\\
    &\times \exp\left(-\int_s^t\left\{ a(r)^\top+\gamma(r)^{-1}b(r)b(r)^\top \right\}dr \right)dt
  \end{align*}
  for $t\geq s$.

  Thus, from Proposition \ref{prop-expression-cov-xi}, we obtain
  \begin{align}
    \label{eq3-13}\begin{split}      
    &d_t\mathrm{Cov}(\xi_{s;t},\xi_{u;t})\\
    =&d_tV[\xi_{s;t}]\exp\left( -\int_u^s \left\{ a(r)^\top+\gamma(r)^{-1}b(r)b(r)^\top \right\}dr \right)\\
    =&-\exp\left(-\int_s^t\left\{ a(r)+b(r)b(r)^\top\gamma(r)^{-1} \right\}dr \right)\\
    &\times \gamma(t)c(t)^\top(\sigma(t)\sigma(t)^\top)^{-1}c(t)\gamma(t)\\
    &\times \exp\left(-\int_u^t\left\{ a(r)^\top+\gamma(r)^{-1}b(r)b(r)^\top \right\}dr \right)dt\\
    =&-\mathrm{Cov}(\xi_{t;t},\xi_{s;t})^\top c(t)^\top(\sigma(t)\sigma(t)^\top)^{-1}c(t)
    \mathrm{Cov}(\xi_{t;t},\xi_{u;t})dt
    \end{split}
  \end{align}
  for $0\leq u\leq s \leq t$.

  In the case where \( s \leq u \), we can derive the same formula by noting that
\begin{align*}
  \frac{\partial}{\partial t} \mathrm{Cov}(\xi_{s; t}, \xi_{u; t})
  = \frac{\partial}{\partial t} \mathrm{Cov}(\xi_{u; t}, \xi_{s; t})^\top.
\end{align*}

Finally, assume that \( \gamma(t) \) may be singular. In this case, we have
\begin{align*}
  d_t \mathrm{Cov}(\xi_{s; t}^\epsilon, \xi_{u; t}^\epsilon) 
  &= -\mathrm{Cov}(\xi_{t; t}^\epsilon, \xi_{s; t}^\epsilon)^\top c(t)^\top (\sigma(t) \sigma(t)^\top)^{-1} c(t) 
  \mathrm{Cov}(\xi_{t; t}^\epsilon, \xi_{u; t}^\epsilon) \, dt,
\end{align*}
where \( \gamma(s)^\epsilon \) and \( \xi_{s; t}^\epsilon \) are defined in (\ref{eq-def-gamma-epsilon-1}) to (\ref{eq-def-xi-epsilon-2}). By (\ref{eq-cov-xi-epsilon}), \( \mathrm{Cov}(\xi_{t; t}^\epsilon, \xi_{s; t}^\epsilon) \) converges uniformly to \( \mathrm{Cov}(\xi_{t; t}, \xi_{s; t}) \) on any compact set as \( \epsilon \to 0 \), leading to the desired result.

\end{proof}

\begin{corollary}\label{cor-expression-cov-xi-s-u-t}
  For $0\leq u\leq s\leq t$, $\mathrm{Cov}(\xi_{s;t},\xi_{u;t})$ has the expression
  \begin{align*}
    &\mathrm{Cov}(\xi_{s;t},\xi_{u;t})\\
    =&\exp\left( \int_u^s \left\{ a(r)-\gamma(r)c(r)^\top(\sigma(r)\sigma(r)^\top)^{-1}c(r)\right\}dr \right)\gamma(u)\\
    &-\gamma(s)\int_s^t \exp\left( \int_s^v \left\{ a(r)^\top-c(r)^\top(\sigma(r)\sigma(r)^\top)^{-1}c(r)\gamma(r)\right\}dr \right)\\
    &\times c(v)^\top (\sigma(v)\sigma(v)^\top)^{-1}c(v)\\
    &\times \exp\left( \int_u^v \left\{ a(r)-\gamma(r)c(r)^\top(\sigma(r)\sigma(r)^\top)^{-1}c(r)\right\}dr \right)dv\gamma(u)
  \end{align*}
\end{corollary}
\begin{proof}
  By Proposition \ref{prop-derivative-cov-xi-s-u-t}, we have
  \begin{align*}
    &\mathrm{Cov}(\xi_{s;t},\xi_{u;t})\\
    =&\mathrm{Cov}(\xi_{s;s},\xi_{u;s})\\
    &-\int_s^t \mathrm{Cov}(\xi_{v;v},\xi_{s;v})^\top c(v)^\top (\sigma(v)\sigma(v)^\top)^{-1}c(v) \mathrm{Cov}(\xi_{v;v},\xi_{u;v})dv.
  \end{align*}
  Then the desired result is immediately obtained, recalling that it follows from Proposition \ref{prop-equation-xi-t-s} that
  \begin{align*}
    \mathrm{Cov}(\xi_{t;t},\xi_{s;t})=\exp\left( \int_s^t \left\{ a(r)-\gamma(r)c(r)^\top(\sigma(r)\sigma(r)^\top)^{-1}c(r)\right\}dr \right)\gamma(s)
  \end{align*}
  for $0\leq s\leq t$. 
\end{proof}

The equation for fixed-point smoothing is presented in a simpler form than the literature \citep{H-infinity-optimal}. 
\begin{proposition}{\bf (Fixed-point smoother)}\\
  For any $s \geq 0$, the process $\{\mu_{s;t}\}_{t\geq s}$ satisfies the stochastic differential equation
  \begin{align*}
    d_t\mu_{s;t}=\mathrm{Cov}(\xi_{s;t},\xi_{t;t})c(t)^\top(\sigma(t)\sigma(t)^\top)^{-1}(dY_t-c(t)\mu_{t;t}dt).
  \end{align*}
\end{proposition}
\begin{proof}
  By the definition of $\mu_{s;t}$ and Proposition \ref{prop-derivative-cov-xi-s-u-t}, we have
  \begin{align*}
    d_t\mu_{s;t}=&\mathrm{Cov}(\xi_{s;t},\xi_{t;t})c(t)^\top(\sigma(t)\sigma(t)^\top)^{-1}(dY_t-c(t)E[X_t]dt)\\
    &-\int_0^t \mathrm{Cov}(\xi_{t;t},\xi_{s;t})^\top c(t)^\top(\sigma(t)\sigma(t)^\top)^{-1}c(t)
    \mathrm{Cov}(\xi_{t;t},\xi_{u;t})\\
  &\times c(u)^\top(\sigma(u)\sigma(u)^\top)^{-1}(dY_u-c(u)E[X_u]du)dt\\
  =&\mathrm{Cov}(\xi_{s;t},\xi_{t;t})c(t)^\top(\sigma(t)\sigma(t)^\top)^{-1}(dY_t-c(t)E[X_t]dt)\\
  &- \mathrm{Cov}(\xi_{t;t},\xi_{s;t})^\top c(t)^\top(\sigma(t)\sigma(t)^\top)^{-1}c(t)(\mu_{t;t}-E[X_t])dt\\
  =&\mathrm{Cov}(\xi_{s;t},\xi_{t;t})c(t)^\top(\sigma(t)\sigma(t)^\top)^{-1}(dY_t-c(t)\mu_{t;t}dt).
  \end{align*}
\end{proof}

We now have all the necessary formulae at hand for the asymptotic expansion. In particular, beyond simply recovering the existing linear filtering and smoothing equations, we have derived new derivation formulae involving the cross-covariance $\mathrm{Cov}(\xi_{s;t}, \xi_{u;t})$. These equations play a crucial role in the next section, which is why Theorem \ref{main-theorem-linear} was established for the conditional distribution of the entire path.

\section{Application to nonlinear filtering}\label{section-application}
In this section, we demonstrate how the results developed in Section \ref{section-linear-filtering} can be applied to the computation of nonlinear filtering. Given the complexity and length of the present work, we provide only a heuristic outline of how the linear results extend to the nonlinear case. Consequently, many of the arguments below lack full mathematical rigour, but they are natural extensions of Itô integrals and Itô's formula, and thus can be understood intuitively. A complete mathematical justification of these steps will be provided in a subsequent paper.
\subsection{Reformulation of the model}
Let us consider a nonlinear unobserved process
\begin{align}
  \label{eq-theta}
  dX_t^\epsilon = \alpha(X_t^\epsilon) \, dt + \epsilon \, \beta(X_t^\epsilon) \, dV_t, 
  \quad X_0 = 0,
\end{align}
and an observed process
\begin{align}\label{eq-Y-eps}
  dY_t^\epsilon = h(X_t^\epsilon) \, dt + \sigma(t) \, dW_t, 
  \quad Y_0 = 0,
\end{align}
where $0 < \epsilon < 1$ and $\sigma > 0$ are constants, $\alpha$, $\beta$, and $h$ are nonlinear functions of class $C^\infty$, and $V$ and $W$ are independent Wiener processes. ALso, we assume $X^\epsilon$, $Y$, $V$, and $W$ are one-dimensional for simplicity.

For this system, we consider the asymptotic expansion of $E[\theta_t^\epsilon  |  \mathcal{Y}_t^\epsilon]$ with respect to $\epsilon$, where $\{\mathcal{Y}_t^\epsilon\}_{t\geq 0}$ is the augmented filtration generated by $Y^\epsilon$ with null sets.

In order to expand this conditional expansion, we first consider the expansion of the process $X_t^\epsilon$ itself. For this purpose, let us write $\displaystyle \frac{\partial^i}{\partial \epsilon^i} X_t^\epsilon = X_t^{[i],\epsilon}$ and $X_t^{[i]} = X_t^{[i],\epsilon}$ for $i = 0, 1, 2, \cdots$. Then, $X_t^{\epsilon}$ admits the Taylor expansion
\begin{align*}
  X_t^\epsilon=X_t^{[0]}+X_t^{[1]}\epsilon+\frac{1}{2}X_t^{[2]}\epsilon^2+\cdots
\end{align*}
and the derivatives $X_t^{[i]}~(i\geq 0)$ satisfy the equations
\begin{align}
  dX_t^{[0]}=&\alpha(X_t^{[0]})dt,~~X_0^{[0]}=0,\\
  \label{eq-theta-1}dX_t^{[1]}=&\alpha'(X_t^{[0]})X_t^{[1]}dt+\beta(X_t^{[0]})dV_t,~~X_t^{[1]}=0,\\
  \label{eq-theta-2}dX_t^{[2]}=&\alpha'(X_t^{[0]})X_t^{[2]}dt+\alpha''(X_t^{[0]})(X_t^{[1]})^{2}dt+2\beta'(X_t^{[0]})X_t^{[1]}dV_t,~~X_t^{[2]}=0,\\
  &\vdots\nonumber
\end{align}
which is obtained by formally differentiating equation (\ref{eq-theta}) and evaluating it at $\epsilon=0$ (see \citet{bichteler1987malliavin} for details).

From these equations, we observe that $X_t^{[0]}$ is a deterministic process, and $X_t^{[1]}$ is a solution of the linear equation (\ref{eq-theta-1}). Furthermore, the equations of the higher derivatives are written in the form
\begin{align*}
  dX_t^{[k]}=\alpha'(X_t^{[0]})X_t^{[k]}dt+(\textrm{polynomials of } X^{[1]},\cdots,X^{[k-1]})\times (dt~\textrm{or}~dV_t).
\end{align*}
This means that $X^{[k]}~(k\geq 2)$ can be recursively expressed as "polynomials" (which contain integrations) of $X^{[1]}$. 

For example, if we assume $\beta(X_t^{[0]})\neq 0$ for every $t\geq 0$, then (\ref{eq-theta-1}) and (\ref{eq-theta-2}) yields
\begin{align*}
  X_t^{[2]}=&\int_0^t \exp\left( \int_s^t \alpha'(X_u^{[0]})du \right)\alpha''(X_s^{[0]})(X_s^{[1]})^{2}ds\\
  &+2\int_0^t \exp\left( \int_s^t \alpha'(X_u^{[0]})du \right)\beta'(X_t^{[0]})X_t^{[1]}dV_s\\
  =&\int_0^t \exp\left( \int_s^t \alpha'(X_u^{[0]})du \right)\alpha''(X_s^{[0]})(X_s^{[1]})^{2}ds\\
  &+2\int_0^t \exp\left( \int_s^t \alpha'(X_u^{[0]})du \right)\beta'(X_t^{[0]})X_t^{[1]}\frac{dX_t^{[1]}-\alpha'(X_t^{[0]})dt}{\beta(X_t^{[0]})}.
\end{align*}
In virtue of It\^o's formula, this can be transformed into
\begin{align*}
  &X_t^{[2]}=\frac{b'(X_t^{[0]})}{b(X_t^{[0]})}(X_t^{[1]})^2\\
  &+\int_0^t \exp\left( \int_s^t \alpha'(X_u^{[0]})du \right) \\
  &\times \left\{ a''(X_s^{[0]})-\frac{\alpha'(X_s^{[0]})\beta'(X_s^{[0]})+\beta''(X_s^{[0]})}{\beta(X_s^{[0]})}+\frac{\beta'(X_s^{[0]})}{\beta(X_s^{[0]})} \right\}(X_s^{[1]})^2ds\\
  &-\int_0^t \exp\left( \int_s^t \alpha'(X_u^{[0]})du \right)\beta'(X_s^{[0]})\beta(X_s^{[0]})ds,
\end{align*}
which is written by $(X^{[1]})^2$ and its integral. A more general case, where we cannot assume $\beta(X_t^{[0]}) = 0$, and $X$ and $Y$ are multi-dimensional, is more complicated but the same idea can be applied. It will be covered in a subsequent paper.

Given this polynomial nature of the higher-order derivatives, $h(X_t^\epsilon)$ in (\ref{eq-Y-eps}) can be expanded as
\begin{align}
  h(X_t^{\epsilon})=&h(X_t^{[0]})+h'(X_t^{[0]})X_t^{[1]}\epsilon\nonumber\\
  &+\frac{1}{2}\left\{ h''(X_t^{[0]})(X_t^{[1]})^2+h'(X_t^{[0]})X_t^{[2]} \right\}\epsilon^2+\cdots\nonumber\\
  =&h(X_t^{[0]})+h'(X_t^{[0]})X_t^{[1]}\epsilon+\sum_{i=2}^\infty (\textrm{polynomial of }X^{[1]})\epsilon^i.\label{eq-expansion-h}
\end{align}
Thus, the original system given by equations (\ref{eq-theta}) and (\ref{eq-Y-eps}) can be reformulated as a polynomially perturbed linear model
\begin{align*}
  &dX_t^{[1]}=\alpha'(X_t^{[0]})X_t^{[1]}dt+\beta(X_t^{[0]})dV_t,\\
  &dY_t^\epsilon=\left\{ h(X_t^{[0]})+h'(X_t^{[0]})X_t^{[1]}\epsilon+\sum_{i=2}^\infty (\textrm{polynomial of }X^{[1]})\epsilon^i \right\}dt+\sigma(t)dW_t,
\end{align*}
with the conditional expectation expanded as
\begin{align*}
  E[X_t^\epsilon|\mathcal{Y}_t^\epsilon]&=X_t^{[0]}+E[X_t^{[1]}|\mathcal{Y}_t^\epsilon]\epsilon+\frac{1}{2}E[X_t^{[2]}|\mathcal{Y}_t^\epsilon]\epsilon^2+\cdots\\
  &=X_t^{[0]}+E[X_t^{[1]}|\mathcal{Y}_t^\epsilon]\epsilon+\sum_{i=2}^\infty E\bigl[(\textrm{polynomial of }X^{[1]})|\mathcal{Y}_t^\epsilon\bigr]\epsilon^i.
\end{align*}
In what follows, we rewrite $X^{[1]}$ as $X$, and consider the simplified model
\begin{alignat}{2}
  \label{eq-model-sim1}
    & dX_t = aX_t \, dt + b \, dV_t, \qquad & X_0 = 0, \\
  \label{eq-model-sim2}
    & dY_t^\epsilon = (cX_t + \epsilon g(X_t)) \, dt + \sigma \, dW_t, \qquad & Y_0^\epsilon = 0,
\end{alignat}
where $a, b, c \in \mathbb{R}$, $0 < \epsilon < 1$, $\sigma > 0$, and $g$ is a polynomial. 

This model may appear too simplified, but it is sufficient to convey the essential idea of our approach. In particular, the derivation of the equations to calculate the expansion is the same as the original setting.

\subsection{Asymptotic expansion for the polynomially perturbed linear model}
Now, we consider the conditional expectation $E[f(X_t) | \mathcal{Y}_t^\epsilon]$ for the model (\ref{eq-model-sim1}) and (\ref{eq-model-sim2}), where $f$ is a polynomial. To illustrate the core idea, it is enough to consider the case of $f(x)=x$.

Fix $T > 0$, and introduce a new probability measure $Q$ by
\begin{align*}
  &Q(A) \\
  &= E\left[ 1_A \exp\left( 
    -\int_0^T \frac{cX_t + \epsilon g(X_t)}{\sigma^2} \, dW_t 
    - \frac{1}{2} \int_0^T \frac{\{cX_t + \epsilon g(X_t)\}^2}{\sigma^2} \, dt 
  \right) \right]
\end{align*}
for $A \in \mathcal{F}$. This measure is well-defined due to Proposition \ref{prop-Z-martingale}. 

Furthermore, by Proposition \ref{prop-Kallianpur-Striebel}, we can write 
\begin{align}
  \label{eq4-1}
  \begin{split}
    &E[X_t |  \mathcal{Y}_t^\epsilon] \\
    &= \frac{\displaystyle E_Q\left[ 
      X_t \exp\left( 
        \int_0^t \frac{cX_s + \epsilon g(X_s)}{\sigma^2} \, dY_s^\epsilon 
        - \frac{1}{2} \int_0^t \frac{\{cX_s + \epsilon g(X_s)\}^2}{\sigma^2} \, ds 
      \right) 
       \middle|  \mathcal{Y}_t 
    \right]}{\displaystyle E_Q\left[ 
      \exp\left( 
        \int_0^t \frac{cX_s + \epsilon g(X_s)}{\sigma^2} \, dY_s^\epsilon 
        - \frac{1}{2} \int_0^t \frac{\{cX_s + \epsilon g(X_s)\}^2}{\sigma^2} \, ds 
      \right) 
       \middle|  \mathcal{Y}_t^\epsilon
    \right]}
  \end{split}
\end{align}
for $0 \leq t \leq T$, where $E_Q$ denotes the expectation with respect to $Q$. Let us denote this conditional expectation by $m_t^\epsilon$, and consider expanding $m_t^\epsilon$ with respect to $\epsilon$.

For this purpose, define a conditional measure $\tilde{P}_t$ by 
\begin{align}
  \label{def-tilde-P-t}\tilde{P}_t(A)=\frac{\displaystyle E_Q\left[ 1_A\exp\left( \int_0^t \frac{cX_s}{\sigma^2}dY_s^\epsilon-\frac{1}{2}\int_0^t \frac{(cX_s)^2}{\sigma^2}ds \right)\middle|\mathcal{Y}_t^\epsilon \right]}{\displaystyle E_Q\left[ \exp\left( \int_0^t \frac{cX_s}{\sigma^2}dY_s^\epsilon-\frac{1}{2}\int_0^t \frac{(cX_s)^2}{\sigma^2}ds \right)\middle|\mathcal{Y}_t^\epsilon \right]}
\end{align}
for $A \in \mathcal{X}_t\vee \mathcal{Y}_t^\epsilon$, and let $\tilde{E}_t$ be the expectation with respect to $\tilde{P}_t$ (see (\ref{def-tilde-P})).

Then the coefficient of $\epsilon^0$ in the expansion of $m_t^\epsilon$ can be written as
\begin{align*}
  m_t^0=\frac{\displaystyle E_Q\left[ X_t \exp\left( \int_0^t \frac{cX_s}{\sigma^2}dY_s^\epsilon-\frac{1}{2}\int_0^t \frac{(cX_s)^2}{\sigma^2}ds \right)\middle|\mathcal{Y}_t^\epsilon \right]}{\displaystyle E_Q\left[  \exp\left( \int_0^t \frac{cX_s}{\sigma^2}dY_s^\epsilon-\frac{1}{2}\int_0^t \frac{(cX_s)^2}{\sigma^2}ds \right)\middle|\mathcal{Y}_t^\epsilon \right]}=\tilde{E}_t[X_t].
\end{align*}
Furthermore, differentiating (\ref{eq4-1}) yields
\begin{align*}
  &\left.\frac{d}{d\epsilon}m_t^\epsilon\right|_{\epsilon=0}\\
  =&\frac{\displaystyle E_Q\left[ X_t \int_0^t g(X_s)(dY_s^\epsilon-cX_sds) \exp\left( \int_0^t \frac{cX_s}{\sigma^2}dY_s^\epsilon-\frac{1}{2}\int_0^t \frac{(cX_s)^2}{\sigma^2}ds \right)\middle|\mathcal{Y}_t^\epsilon \right]}{\displaystyle E_Q\left[  \exp\left( \int_0^t \frac{cX_s}{\sigma^2}dY_s^\epsilon-\frac{1}{2}\int_0^t \frac{(cX_s)^2}{\sigma^2}ds \right)\middle|\mathcal{Y}_t^\epsilon \right]}\\
  &-\frac{\displaystyle E_Q\left[ X_t  \exp\left( \int_0^t \frac{cX_s}{\sigma^2}dY_s^\epsilon-\frac{1}{2}\int_0^t \frac{(cX_s)^2}{\sigma^2}ds \right)\middle|\mathcal{Y}_t^\epsilon \right]}{\displaystyle E_Q\left[  \exp\left( \int_0^t \frac{cX_s}{\sigma^2}dY_s^\epsilon-\frac{1}{2}\int_0^t \frac{(cX_s)^2}{\sigma^2}ds \right)\middle|\mathcal{Y}_t^\epsilon \right]}\\
  &\times \frac{\displaystyle E_Q\left[  \int_0^t g(X_s)(dY_s^\epsilon-cX_sds) \exp\left( \int_0^t \frac{cX_s}{\sigma^2}dY_s^\epsilon-\frac{1}{2}\int_0^t \frac{(cX_s)^2}{\sigma^2}ds \right)\middle|\mathcal{Y}_t^\epsilon \right]}{\displaystyle E_Q\left[  \exp\left( \int_0^t \frac{cX_s}{\sigma^2}dY_s^\epsilon-\frac{1}{2}\int_0^t \frac{(cX_s)^2}{\sigma^2}ds \right)\middle|\mathcal{Y}_t^\epsilon \right]}\\
  =&\tilde{E}_t\left[ X_t\int_0^t g(X_s)(dY_s^\epsilon-cX_sds) \right]-\tilde{E}_t\left[ X_t \right]\tilde{E}_t\left[ \int_0^t g(X_s)(dY_s^\epsilon-cX_sds) \right],
\end{align*}
which gives the coefficient of $\epsilon^1$. Here, note that \( Y^\epsilon \) is the given observation and is treated as fixed; in particular, it is not differentiated.

In order to find higher order coefficients, write
\begin{align*}
  &\exp\left( \int_0^t \frac{cX_s+\epsilon g(X_s)}{\sigma^2}dY_s^\epsilon-\frac{1}{2}\int_0^t\frac{\{cX_s+\epsilon g(X_s)\}^2}{\sigma^2} ds \right)\\
  =& \exp\left( \int_0^t \frac{cX_s}{\sigma^2}dY_s^\epsilon-\frac{1}{2}\int_0^t \frac{(cX_s)^2}{\sigma^2}ds \right)\\
  &\times \exp\left( \epsilon \int_0^t \frac{g(X_s)}{\sigma^2}(dY_s^\epsilon-cX_sds)-\frac{1}{2}\epsilon^2\int_0^t \frac{g(X_s)^2}{\sigma^2}ds \right).
\end{align*}
Here, set
\begin{align*}
  K_t^\epsilon=\exp\left( \epsilon \int_0^t \frac{g(X_s)}{\sigma^2}(dY_s^\epsilon-cX_sds)-\frac{1}{2}\epsilon^2\int_0^t \frac{g(X_s)^2}{\sigma^2}ds \right).
\end{align*}
Then $K_t^\epsilon$ is an exponential local martingale, and we have the expansion
\begin{align*}
  K_t^\epsilon=&1+\frac{\epsilon}{\sigma^2}\int_0^t K_s^\epsilon g(X_s)(dY_s^\epsilon-cX_sds)\\
  =&1+\frac{\epsilon}{\sigma^2}\int_0^t g(X_s)(dY_s^\epsilon-cX_sds)\\
  &+\frac{\epsilon^2}{\sigma^4}\int_0^t K_u^\epsilon g(X_u)(dY_u^\epsilon-cX_udu)g(X_s)(dY_s^\epsilon-cX_sds)\\
  =&\cdots\\
  =&1+\sum_{i=1}^n\frac{\epsilon^i}{\sigma^{2i}}\int_0^t\int_0^{t_i}\cdots \int_0^{t_2}g(X_{t_1})\cdots g(X_{t_{n}})\\
  &\times (dY_{t_1}^\epsilon-cX_{t_1}dt_1)\cdots (dY_{t_i}^\epsilon-cX_{t_i}dt_i)\\
  &+\frac{\epsilon^{n+1}}{\sigma^{2(n+1)}}\int_0^t\int_0^{t_{n+1}}\cdots \int_0^{t_2}K_{t_1}^\epsilon g(X_{t_1})\cdots g(X_{t_{n+1}})\\
  &\times (dY_{t_1}^\epsilon-cX_{t_1}dt_1)\cdots (dY_{t_{n+1}}^\epsilon-cX_{t_{n+1}}dt_{n+1})
\end{align*}
for any $n \in \mathbb{N}$. The last term is shown to be $O(\epsilon^{n+1})$ in some sense.
Therefore, (\ref{eq4-1}) can be expanded as
\begin{align}
  E[X_t|\mathcal{Y}_t^\epsilon]=&\frac{\displaystyle E_Q\left[ X_t K_t^\epsilon\exp\left( \int_0^t \frac{cX_s}{\sigma^2}dY_s^\epsilon-\frac{1}{2}\int_0^t \frac{(cX_s)^2}{\sigma^2}ds \right)\middle|\mathcal{Y}_t^\epsilon \right]}{\displaystyle E_Q\left[ K_t^\epsilon\exp\left( \int_0^t \frac{cX_s}{\sigma^2}dY_s^\epsilon-\frac{1}{2}\int_0^t \frac{(cX_s)^2}{\sigma^2}ds \right)\middle|\mathcal{Y}_t^\epsilon \right]}\nonumber\\
  =&\frac{\displaystyle \frac{\displaystyle E_Q\left[ X_t K_t^\epsilon\exp\left( \int_0^t \frac{cX_s}{\sigma^2}dY_s^\epsilon-\frac{1}{2}\int_0^t \frac{(cX_s)^2}{\sigma^2}ds \right)\middle|\mathcal{Y}_t^\epsilon \right]}{\displaystyle E_Q\left[ \exp\left( \int_0^t \frac{cX_s}{\sigma^2}dY_s^\epsilon-\frac{1}{2}\int_0^t \frac{(cX_s)^2}{\sigma^2}ds \right)\middle|\mathcal{Y}_t^\epsilon \right]}}{\displaystyle \frac{\displaystyle E_Q\left[ K_t^\epsilon\exp\left( \int_0^t \frac{cX_s}{\sigma^2}dY_s^\epsilon-\frac{1}{2}\int_0^t \frac{(cX_s)^2}{\sigma^2}ds \right)\middle|\mathcal{Y}_t^\epsilon \right]}{\displaystyle E_Q\left[ \exp\left( \int_0^t \frac{cX_s}{\sigma^2}dY_s^\epsilon-\frac{1}{2}\int_0^t \frac{(cX_s)^2}{\sigma^2}ds \right)\middle|\mathcal{Y}_t^\epsilon \right]}}\nonumber\\
  \label{eq4-5}\begin{split}
    =&\frac{\tilde{E}_t[X_tK_t^\epsilon]}{\tilde{E}_t[K_t^\epsilon]}\\
  =&J_t^0(X_t) + \left\{ J_t^1(X_t) - J_t^0(X_t) J_t^1(1) \right\} \epsilon \\
  &+ \left\{ J_t^2(X_t) - J_t^0(X_t) J_t^2(1) - J_t^1(X_t) J_t^1(1) + J_t^0(X_t) J_t^1(1)^2 \right\} \epsilon^2 \\
  &+ \cdots + \sum_{k=0}^n \sum_{p=1}^{n-k} \sum_{\substack{i_1,\cdots,i_p \in \mathbb{N} \\ i_1+\cdots+i_p=n-k}} (-1)^p J_t^k(X_t) J_t^{i_1}(1) \cdots J_t^{i_p}(1) \epsilon^n\\
  &+O(\epsilon^{n+1}),
  \end{split}  
\end{align}
where
\begin{align}
  \label{eq4-3}\begin{split}
    J_t^i(U) = \frac{1}{\sigma^{2n}} \tilde{E}_t\Biggl[ &U \int_0^t \int_0^{t_i} \cdots \int_0^{t_2} g(X_{t_1}) \cdots g(X_{t_i}) \\
  &\times (dY_{t_1}^\epsilon - cX_{t_1}  dt_1) \cdots (dY_{t_i}^\epsilon - cX_{t_i}  dt_i) \Biggr]
  \end{split}    
\end{align}
for a random variable $U$. Thus, the expansion of $E[X_t|\mathcal{Y}_t^\epsilon]$ is reduced to computing the (conditional) expectation (\ref{eq4-3}).

According to Theorem \ref{main-theorem-linear}, the law of $\{X_s\}_{0\leq s\leq t}$ under $\tilde{P}_t$ in (\ref{def-tilde-P-t}) is a Gaussian process, so if $g$ is a polynomial, (\ref{eq4-3}) can be expressed using 
\begin{align*}
 \mu_{s;t}=\tilde{E}_t[X_s],~~\mathrm{Cov}_{\tilde{P}_t}(X_{s},X_{u})=\gamma(s,u;t). 
\end{align*}
For instance, when $g(x)=x^3$, we have
\begin{align*}
  J_t^0(X_t)=&\tilde{E}_t[X_t]=\mu_{t;t},\\
  J_t^1(X_t)=&\tilde{E}_t\left[ X_t\int_0^t X_s^3 (dY_s^\epsilon - cX_sds) \right]\nonumber\\
  =&\int_0^t \left( \tilde{E}_t\left[ X_tX_s^3 \right] dY_s^\epsilon - c\tilde{E}_t\left[ X_tX_s^4 \right]ds\right)\nonumber\\
    =&\mu_{t;t}\int_0^t \mu_{s;t}^3 (dY_s^\epsilon-c\mu_{s;t}ds)
    +3\mu_{t;t}\int_0^t \mu_{s;t}\gamma(s,s;t) (dY_s^\epsilon-c\mu_{s;t}ds)\\
  &+3\int_0^t \mu_{s;t}^2 \gamma(s,t;t) (dY_s^\epsilon-c\mu_{s;t}ds)\\
  &-c\int_0^t \mu_{s;t}^3 \gamma(s,t;t)ds
  -3c\mu_{t;t}\int_0^t \mu_{s;t}^2 \gamma(s,t;t)ds\\
  &+3\int_0^t \gamma(s,t;t)\gamma(s,s;t) (dY_s^\epsilon-c\mu_{s;t}ds)\\
  &-6c\int_0^t \mu_{s;t}\gamma(s,t;t)\gamma(s,s;t)ds\\
  &-3c\mu_{t;t}\int_0^t \gamma(s,t;t)\gamma(s,s;t)ds,\\
  J_t^0(1)=&\tilde{E}_t\left[ \int_0^t X_s^3 (dY_s^\epsilon - cX_s \, ds) \right]\nonumber\\
  =&\int_0^t \left( \tilde{E}_t\left[ X_s^3 \right] dY_s^\epsilon - c\tilde{E}_t\left[ X_s^4 \right]ds\right)\nonumber\\
    =&\int_0^t \mu_{s;t}^3 (dY_s^\epsilon-c\mu_{s;t}ds)+3\int_0^t \mu_{s;t}\gamma(s,s;t) (dY_s^\epsilon-c\mu_{s;t}ds)\\
  &-3c\int_0^t \mu_{s;t}^2 \gamma(s,t;t)ds-3c\int_0^t \gamma(s,t;t)\gamma(s,s;t)ds.
\end{align*}
Using this, the coefficient of $\epsilon^1$ in (\ref{eq4-5}) is represented as 
\begin{align}
  &J_t^1(X_t) - J_t^0(X_t) J_t^1(1) \nonumber \\
  \label{eq4-6}\begin{split}
    =& 3\int_0^t \mu_{s;t}^2 \gamma(s,t;t) (dY_s^\epsilon - c\mu_{s;t} ds)
      - c \int_0^t \mu_{s;t}^3 \gamma(s,t;t) \, ds \\
      &+ 3\int_0^t \gamma(s,t;t) \gamma(s,s;t) (dY_s^\epsilon - c\mu_{s;t}ds) \\
      &- 6c \int_0^t \mu_{s;t} \gamma(s,t;t) \gamma(s,s;t) \, ds.
  \end{split}      
\end{align}
Here, for numerical stability, we preserved $dY_s^\epsilon - c\mu_{s;t} \, ds$ without fully expanding integrals. Note also that some of the stochastic integrals above are not well-defined, since $\{\mu_{s;t}\}_{0 \leq s \leq t}$ are not adapted (they depend on the full history of $\{Y_s^\epsilon\}_{0 \leq s \leq t}$). However, we can justify the procedure by appropriately interpreting the integrals.

Anyway, $J_t^i(X_t)$ and $J_t^i(1)$ are expanded into a sum of terms of the form
\begin{align*}
  A_t(p,q,r,\alpha;n)
    =&\int_0^t\int_0^{t_n}\cdots \int_0^{t_{2}}\prod_{i=1}^n(\mu_{t_i;t})^{p_i}\prod_{i=1}^n \mathrm{Cov}(\xi_{t;t},\xi_{t_i;t})^{q_{i}}\\
    &\times \prod_{1\leq i\leq j\leq n}\mathrm{Cov}(\xi_{t_i;t},\xi_{t_j;t})^{r_{ij}}d\lambda_{t_1;t}^{(\alpha_1)}\cdots d\lambda_{t_n;t}^{(\alpha_n)},
\end{align*}
where $n \in \mathbb{N}$, $p=(p_i)_{i=1,\cdots,n} \in \mathbb{Z}_+^n,~q=(q_i)_{i=1,\cdots,n} \in \mathbb{Z}_+^n$, $r=(r_{ij})_{i,j=1,\cdots,n} \in M_n(\mathbb{Z}_+)$, $\alpha=(\alpha_i)_{i=1,\cdots,n} \in \{0,1\}^n$, and
\begin{align*}
  d\lambda_{s;t}^{(i)}=\begin{cases}
    ds&(i=0)\\
    dY_s^\epsilon-c\mu_{s;t}ds&(i=1).
  \end{cases}
\end{align*}
For example, (\ref{eq4-6}) can be written as
\begin{align*}
  3A_t(2,1,0,1;1)-cA_t(3,1,0,0;1)+3A_t(0,1,1,1;1)-6cA_t(1,1,1,0;1).
\end{align*}
To derive a recursive formula for $A_t(p,q,r,\alpha;n)$, recall the equations derived in Section \ref{section-linear-equations} ($\gamma(s,u;t)$ is equivalent to $\mathrm{Cov}(\xi_{s;t},\xi_{u;t})$ in the previous section): 
\begin{align}
  \label{eq4-7}&d_t\mu_{t;t}=a\mu_{t;t}dt+ \frac{c\gamma(t)}{\sigma^2}(dY_t^\epsilon-c\mu_{t;t}dt),\\  
  &\label{eq4-8}d_t\mu_{s;t}=\frac{c}{\sigma^2}\gamma(s,t;t)(dY_t^\epsilon-c\mu_{t;t}dt),\\
  \label{eq4-9}&d_t\gamma(s,t;t)=\left\{ a-\frac{c^2\gamma(t)}{\sigma^2} \right\}\gamma(s,t;t)dt,\\
  \label{eq4-10}&d_t\gamma(s,u;t)=-\frac{c^2}{\sigma^2}\gamma(s,t;t)
    \gamma(u,t;t)dt,
\end{align}
where $\gamma(t)=\gamma(t,t;t)$ is the solution of
\begin{align}
  \label{eq4-11}\frac{d}{dt}\gamma(t)=-\frac{c^2}{\sigma^2}\gamma(t)^2+2a\gamma(t)+b^2.
\end{align}
As mentioned above, the stochastic integral defining $A_t(p,q,r,\alpha;n)$ may not be interpreted as an It\^o integral, but it can be differentiated in accordance with the following rule, which is similar to It\^o's rule:\vspace{10pt}\\
{\bf Rule:} If $\alpha_{s;t}$ satisfies 
\begin{align*}
  d_t\alpha_{s;t}=p_{s;t}dt+q_{s;t}dY_t^\epsilon,
\end{align*}
then
\begin{align}
  \label{eq-derivative-rule}d_t\int_0^t \alpha_{s;t}dY_s^\epsilon=\alpha_{t;t}dY_t^\epsilon+\int_0^t p_{s;t}dY_s^\epsilon\,dt+\int_0^t q_{s;t}dY_s^\epsilon\,dY_t^\epsilon+\sigma^2q_{t;t}dt.
\end{align}
Here, the final term results from the Itô correction term $q_{t;t}dY_t\times dY_t$.\vspace{10pt}

For example, let us consider the derivative of 
\begin{align*}
  A_t(1,0,0,1;1)=\int_0^t \mu_{s;t}(dY_s^\epsilon-c\mu_{s;t}ds)=\int_0^t \mu_{s;t}dY_s^\epsilon-c\int_0^t\mu_{s;t}^2ds.
\end{align*}
By (\ref{eq4-8}) and applying (\ref{eq-derivative-rule}) yields
\begin{align*}
  d_t\int_0^t \mu_{s;t}dY_s^\epsilon=\mu_{t,t}dY_t^\epsilon+\frac{c}{\sigma^2}\int_0^t \gamma(s,t;t) dY_s^\epsilon (dY_t^\epsilon-c\mu_{t;t}dt)+c\gamma(t)dt.
\end{align*}
and
\begin{align*}
  d_t\int_0^t\mu_{s;t}^2ds=\mu_{t;t}^2dt+2\frac{c}{\sigma^2}\int_0^t\mu_{s;t}\gamma(s,t;t)ds(dY_t^\epsilon-c\mu_{t;t}dt).
\end{align*}
Thus, we obtain the equation
\begin{align*}
  dA_t(1,0,0,1;1)=&\mu_{t;t}(dY_t^\epsilon-c\mu_{t;t}dt)\\
  &+\frac{c}{\sigma^2}\int_0^t \gamma(s,t;t)(dY_s-c\mu_{s;t}ds)\,(dY_t^\epsilon-c\mu_{t;t}dt)\\
  &-\frac{c^2}{\sigma^2}\int_0^t \mu_{s;t}\gamma(s,t;t)ds\,(dY_t^\epsilon-c\mu_{t;t}dt)\\
  &-\frac{c^2}{\sigma^2}\int_0^t \gamma(s,t;t)^2ds\,dt+c\gamma(t)dt\\
  =&\mu_{t;t}(dY_t^\epsilon-c\mu_{t;t}dt)\\
  &+\frac{c}{\sigma^2}A_t(0,1,0,1;1)(dY_t^\epsilon-c\mu_{t;t}dt)\\
  &-\frac{c^2}{\sigma^2}A_t(1,1,0,0;1)(dY_t^\epsilon-c\mu_{t;t}dt)\\
  &-\frac{c^2}{\sigma^2}A_t(0,2,0,0;1)dt+c\gamma(t)dt.
\end{align*} 
In the same way, we have
\begin{align*}
  dA_t(0,1,0,1;1)=&\left\{ a-\frac{c^2\gamma(t)}{\sigma^2} \right\}\int_0^t \gamma(s,t;t)(dY_s^\epsilon-c\mu_{s;t}ds)\,dt\\
  &-\frac{c^2}{\sigma^2}\int_0^t \gamma(s,t;t)^2ds\,(dY_t^\epsilon-c\mu_{t;t}dt),\\
  =&\left\{ a-\frac{c^2\gamma(t)}{\sigma^2} \right\}A_t(0,1,0,1;1)dt\\
  &-\frac{c^2}{\sigma^2}A_t(0,2,0,0;1)(dY_t^\epsilon-c\mu_{t;t}dt),\\
  dA_t(1,1,0,0;1)=&\left\{ a-\frac{c^2\gamma(t)}{\sigma^2} \right\}A_t(1,1,0,0;1)\\
  &+\frac{c}{\sigma^2}\int_0^t \gamma(s,t;t)^2ds\,(dY_t^\epsilon-c\mu_{t;t}dt),\\
  =&\left\{ a-\frac{c^2\gamma(t)}{\sigma^2} \right\}dA_t(1,1,0,0;1)\\
  &+\frac{c}{\sigma^2}A_t(0,2,0,0;1)(dY_t^\epsilon-c\mu_{t;t}dt),\\
  dA_t(0,2,0,0;1)=&2\left\{ a-\frac{c^2\gamma(t)}{\sigma^2} \right\}\int_0^t \gamma(s,t;t)^2ds\,dt\\
  =&2\left\{ a-\frac{c^2\gamma(t)}{\sigma^2} \right\}A_t(0,2,0,0;1)dt.
\end{align*}
Therefore, $A_t(1,0,0,1;1)$ can be computed by solving the four-dimensional stochastic differential equation
\begin{align}
  dA_t(1,0,0,1;1)=&\mu_{t;t}(dY_t^\epsilon-c\mu_{t;t}dt)+\frac{c}{\sigma^2}A_t(0,1,0,1;1)(dY_t^\epsilon-c\mu_{t;t}dt)\nonumber\\
  &-\frac{c^2}{\sigma^2}A_t(1,1,0,0;1)(dY_t^\epsilon-c\mu_{t;t}dt)\nonumber\\
  &-\frac{c^2}{\sigma^2}A_t(0,2,0,0;1)dt+c\gamma(t)dt,\label{eq-ode-1}\\
  dA_t(0,1,0,1;1)=&\left\{ a-\frac{c^2\gamma(t)}{\sigma^2} \right\}A_t(0,1,0,1;1)dt\nonumber\\
  &-\frac{c^2}{\sigma^2}A_t(0,2,0,0;1)(dY_t^\epsilon-c\mu_{t;t}dt),\label{eq-ode-2}\\
  dA_t(1,1,0,0;1)=&\left\{ a-\frac{c^2\gamma(t)}{\sigma^2} \right\}dA_t(1,1,0,0;1)\nonumber\\
  &+\frac{c}{\sigma^2}A_t(0,2,0,0;1)(dY_t^\epsilon-c\mu_{t;t}dt),\label{eq-ode-3}\\
  dA_t(0,2,0,0;1)=&2\left\{ a-\frac{c^2\gamma(t)}{\sigma^2} \right\}A_t(0,2,0,0;1)dt,\label{eq-ode-4}
\end{align}
where $\mu_{t;t}$ and $\gamma(t)$ is given by (\ref{eq4-7}) and (\ref{eq4-11}). A numerical solution to this equation can be recursively computed using the Euler-Maruyama method.

Similarly, $A_t(p,q,r,\alpha;n)$ is eventually reduced to the case of $p=r=\alpha=0$ by repeatedly differentiating it using (\ref{eq4-8})--(\ref{eq4-10}). A closed system of equations can then be obtained due to (\ref{eq4-9}).

The algorithm can be summarized in the following steps:
\begin{enumerate}
  \item Expand the conditional expectation as in (\ref{eq4-5}).
  \item Due to the Gaussian property of $X$, express (\ref{eq4-3}) in terms of the mean $\mu_{s;t}$ and covariance $\gamma(s,u;t)$, resulting in an expression like (\ref{eq4-6}).
  \item Differentiate each term from the previous step according to (\ref{eq4-8})--(\ref{eq4-10}) and the rule (\ref{eq-derivative-rule}), and derive a finite-dimensional stochastic differential equation to calculate the term.
\end{enumerate}

\subsection{Numerical simulation 1: Linear case}
In the first simulation, we set $g(x)=x$ and consider the following model 
\begin{alignat*}{2}
    & dX_t = aX_t \, dt + b \, dV_t, \qquad & X_0 = 0, \\
    & dY_t^\epsilon = (cX_t + \epsilon X_t ) \, dt + \sigma \, dW_t, \qquad & Y_0 = 0,
\end{alignat*}
In this case, the true value of the conditional expectation $E[X_t|\mathcal{Y}_t^\epsilon]$ can be calculated through the Kalman-Bucy filter, which serves as a benchmark for evaluating the accuracy of our asymptotic expansion.

We simulated a path of $X_t$ and $Y_t^\epsilon$ over the interval $0\leq t\leq 10$ using the parameters
\begin{align*}
  a = -0.4, \quad b = 0.5, \quad c = 1, \quad \sigma = 0.3, \quad \epsilon = 0.2.
\end{align*}
For the simulated path, we expanded the conditional expectation $E[X_t|\mathcal{Y}_t^\epsilon]$ with respect to $\epsilon$ following the algorithm described in the previous section. Figure \ref{fig:linear-plot} shows the plots of this expansion and the true conditional expansion calculated by the Kalman-Bucy filter. presents the plots of this expansion alongside the true conditional expectation computed using the Kalman-Bucy filter. As shown in the figure, the expanded filters converge toward the true line as the order of expansion increases, demonstrating the validity of our algorithm.

\begin{figure}[t]
  \centering
  \begin{subfigure}[b]{\linewidth}
      \centering
      \includegraphics[width=0.80\linewidth]{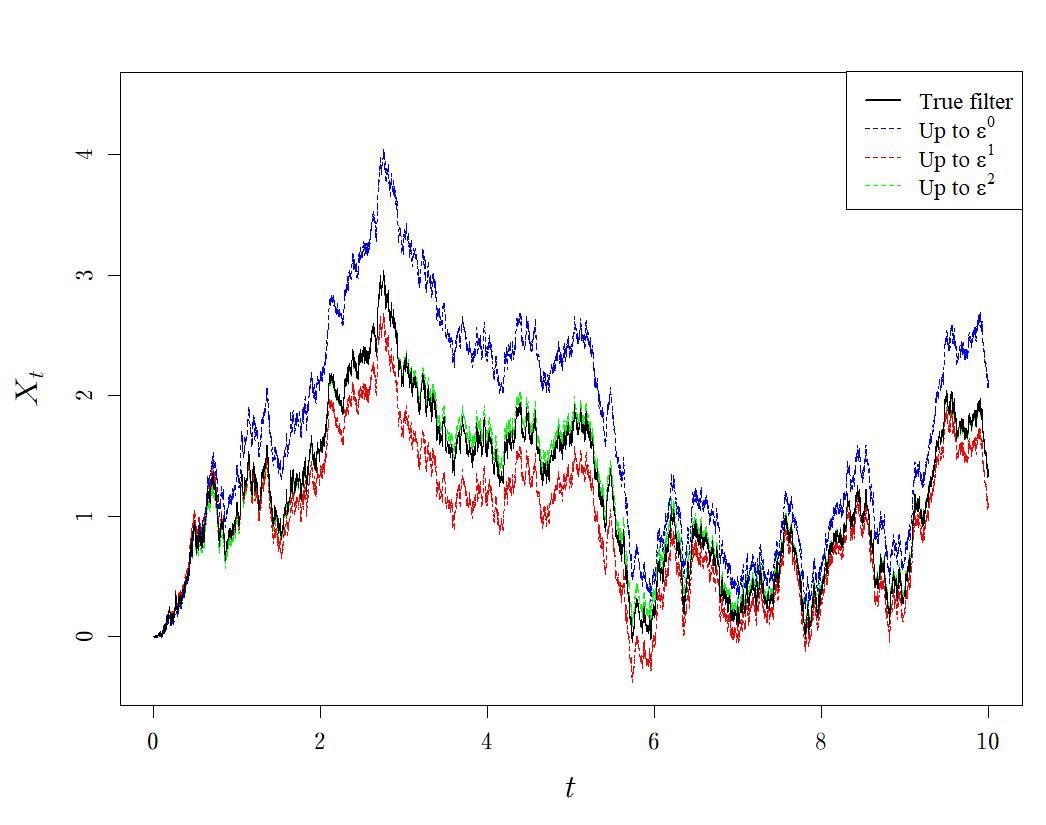}
  \end{subfigure}
  \caption{Plots of the expanded filter and the true filter.}
  \label{fig:linear-plot}
\end{figure}

\subsection{Numerical simulation 2: cubic sensor}
In the second simulation, we consider the case of $g(x)=x^3$: 
\begin{alignat*}{2}
    & dX_t = aX_t \, dt + b \, dV_t, \qquad & X_0 = 0, \\
    & dY_t^\epsilon = (cX_t + \epsilon X_t^3 ) \, dt + \sigma \, dW_t, \qquad & Y_0 = 0,
\end{alignat*}
This model is referred to as the cubic sensor, which serves as a benchmark in several studies \citep{armstrong2019optimal}.

In this case, the true filter cannot be computed to directly evaluate the performance of the expansion. Instead, we generated 1,000 sample paths of $X_t$ and $Y_t$ on the interval $0 \leq t \leq 100$, aand computed the asymptotic expansion for each path. To assess the accuracy of the expansion, we evaluated the integrated squared error for each path, defined as
\begin{align*}
    \int_0^{100} (X_t - N_t^{[n],\epsilon})^2 \, dt,
\end{align*}
where 
\begin{align*}
  N_t^{[n],\epsilon}=&J_t^0(X_t) + \left\{ J_t^1(X_t) - J_t^0(X_t) J_t^1(1) \right\} \epsilon \\
  &+ \left\{ J_t^2(X_t) - J_t^0(X_t) J_t^2(1) - J_t^1(X_t) J_t^1(1) + J_t^0(X_t) J_t^1(1)^2 \right\} \epsilon^2 \\
  &+ \cdots + \sum_{k=0}^n \sum_{p=1}^{n-k} \sum_{\substack{i_1,\cdots,i_p \in \mathbb{N} \\ i_1+\cdots+i_p=n-k}} (-1)^p J_t^k(X_t) J_t^{i_1}(1) \cdots J_t^{i_p}(1) \epsilon^n
\end{align*}
is the $n$-th order approximation (see (\ref{eq4-5})). Due to the $L^2$-optimality of the conditional expectation, the true filter is expected to minimize the mean of integrated squared error.

Here, we used the following parameters:
\begin{align*}
    a = -0.4, \quad b = 0.5, \quad c = 1, \quad \sigma = 0.3, \quad \epsilon = 0.2,
\end{align*}
and simulations were performed with a step size of 0.01. 

Table \ref{table-errors} shows the minimum, median, mean, and maximum of the integrated squared error for the 1,000 simulated paths. Note that $N_t^{[0],\epsilon} = \mu_{t,t}$ corresponds to the extended Kalman filter (linear approximation).
\begin{table}[h]
  \renewcommand{\arraystretch}{1.3}
  \begin{center}
    \caption{\label{table-errors}Integrated squared errors of each filter.}
   \begin{tabular}{|l||c|c|c|}\hline
     Filter & $\mu_{t;t}$ &$N_t^{[1],\epsilon}$&$N_t^{[2],\epsilon}$ \\\hline \hline
     Min. & 8.200 &7.957&7.912\\\hline
     Median & 10.91 &10.73&11.19\\\hline
     Mean&10.98&10.76&25.01\\\hline
     Max. &16.87&35.60&9215\\\hline
   \end{tabular}
  \end{center}
 \end{table}
 We can observe the superiority of the asymptotic expansion over the extended Kalman filter in terms of the minimum and median of $N_t^{[1],\epsilon}$, but it does not perform as well in terms of the mean and maximum.

\begin{figure}[htbp]
  \centering
  \begin{subfigure}[b]{\linewidth}
      \centering
      \includegraphics[width=0.99\linewidth]{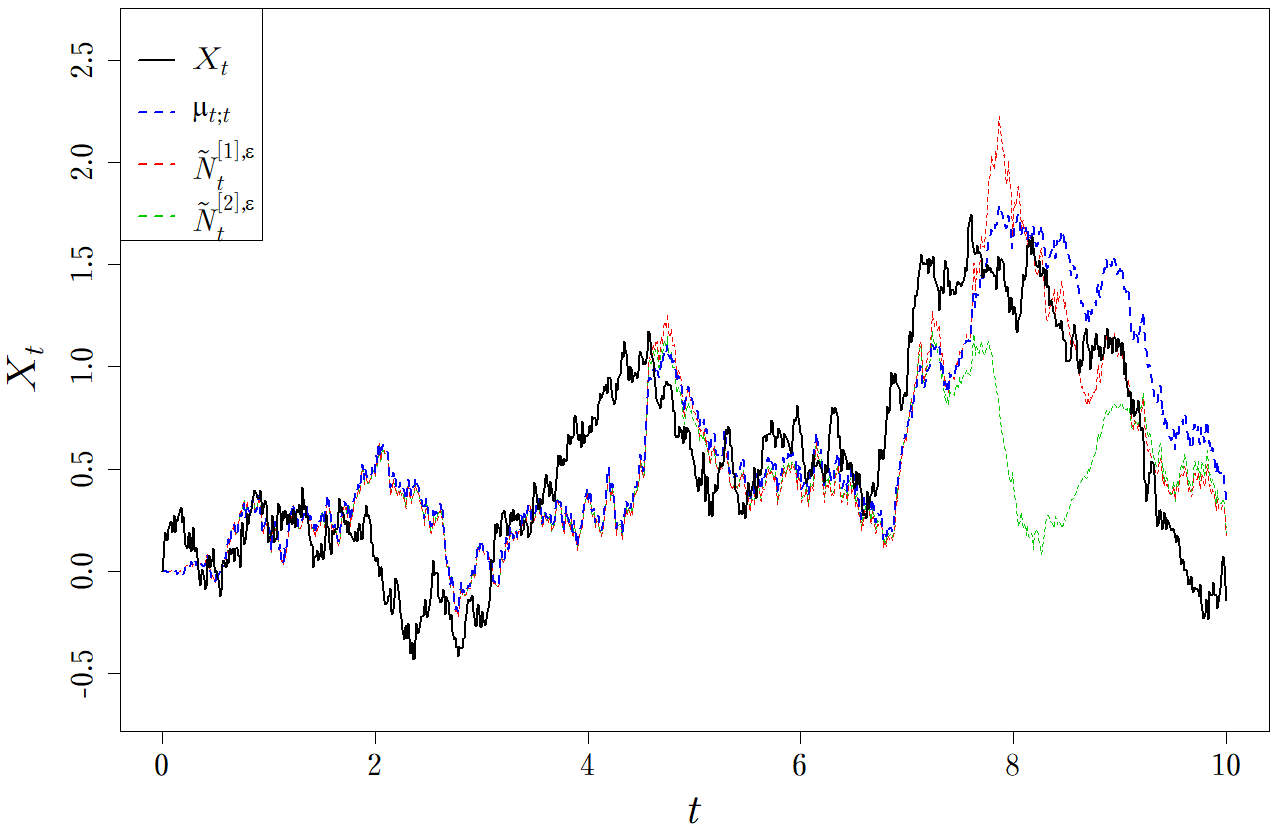}
  \end{subfigure}
  \vspace{0.2cm}

  \begin{subfigure}[b]{\linewidth}
      \centering
      \includegraphics[width=0.99\linewidth]{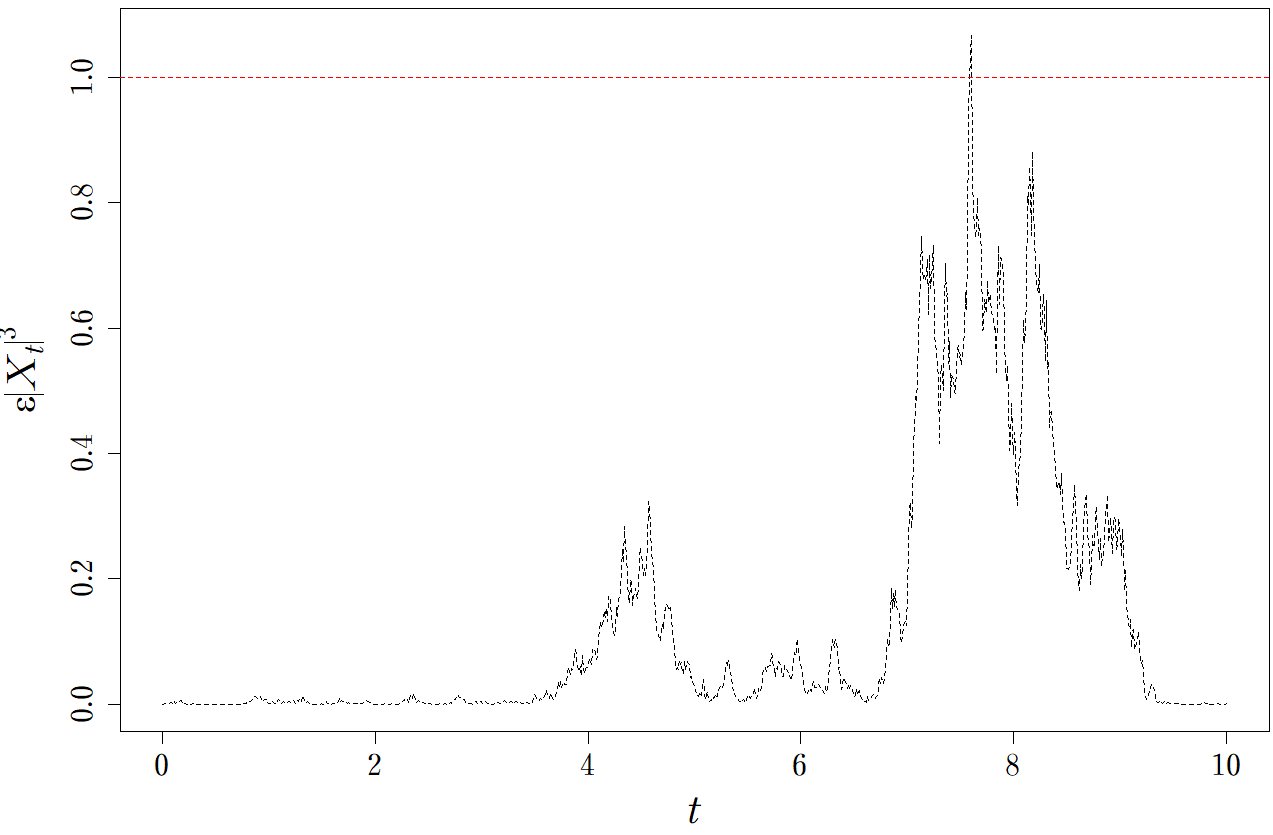}
  \end{subfigure}

  \caption{Plots of the filtering results and the perturbation term for a single simulated path.}
  \label{fig:filter-plots}
\end{figure}

This result can be explained by Figure \ref{fig:filter-plots}, which shows the plots of \(X_t\), \(\mu_{t;t}\), \(N_t^{[1],\epsilon}\), \(N_t^{[2],\epsilon}\), and \(\epsilon |X_t|^3\) for a single simulated path for the interval \(0 \leq t \leq 10\). From the figure, we observe that \(N_t^{[1],\epsilon}\) and \(N_t^{[2],\epsilon}\) improve upon the extended Kalman filter \(\mu_{t;t}\) over much of the interval, but \(N_t^{[2],\epsilon}\) begins to deviate from the true process before $t = 8$.

Looking at the lower panel, we find that this deviation aligns with the moments when the absolute value of the perturbation term \(\epsilon |X_t|^3\) exceeds 1. This behavior is expected, as our expansion is performed with respect to \(\epsilon {X_t}^3\), rather than just \(\epsilon\). Consequently, the expansion diverges when \(\epsilon |X_t|^3 > 1\). However, the upper panel suggests that this effect is not persistent, as the path of \(N_t^{[2],\epsilon}\) returns closer to the true path as soon as \(\epsilon |X_t|^3\) falls below 1. Hence, the minimum and median values in Table \ref{table-errors} capture the convergence of the asymptotic expansion when \(\epsilon |X_t|^3 < 1\), whereas the maximum and mean values reflect its divergence when \(\epsilon |X_t|^3 > 1\). 

To manage the divergence of the expansion, we propose a modification to the coefficients. Let us write \(N_t^{[n],\epsilon}\) as
\[
N_t^{[n],\epsilon} = n_t^{[0]} + n_t^{[1]}\epsilon + n_t^{[2]}\epsilon^2 + \cdots + n_t^{[n]}\epsilon^n,
\]
where the coefficients are given by:
\begin{align*}
  &n_t^{[0]} = J_t^0(X_t), \\
&n_t^{[1]} = J_t^1(X_t) - J_t^0(X_t)J_t^1(1), \quad\\
&n_t^{[2]}=J_t^2(X_t) - J_t^0(X_t) J_t^2(1) - J_t^1(X_t) J_t^1(1) + J_t^0(X_t) J_t^1(1)^2.
\end{align*}
As discussed above, the term \(n_t^{[i]}\epsilon^i\) either converges or diverges at an exponential rate, depending on the value of \(\epsilon |X_t|^3\). To suppress the divergence and ensure exponential decay, we fix a parameter \(r > 0\) and define the modified coefficients \(\tilde{n}_t^{[i]}\) as follows:
\[
\tilde{n}_t^{[0]} = n_t^{[0]}, \quad
\tilde{n}_t^{[i]} =
\begin{cases} 
n_t^{[i]} & \text{if } |n_t^{[i]}\epsilon^i| \leq r |\tilde{n}_t^{[i-1]}\epsilon^{i-1}|, \\
\displaystyle r |\tilde{n}_t^{[i-1]}\epsilon^{i-1}|\frac{n_t^{[i]}}{|n_t^{[i]}\epsilon^i|} & \text{if } |n_t^{[i]}\epsilon^i| > r |\tilde{n}_t^{[i-1]}\epsilon^{i-1}|,
\end{cases}
\]
for \(i = 1, 2, \cdots\). Using these coefficients, we define a modified filter by
\[
\tilde{N}_t^{[i],\epsilon} = \tilde{n}_t^{[0]} + \tilde{n}_t^{[1]}\epsilon + \tilde{n}_t^{[2]}\epsilon^2 + \cdots + \tilde{n}_t^{[i]}\epsilon^i.
\]
By construction, we have \(|\tilde{n}_t^{[i]}\epsilon^i| \leq r |\tilde{n}_t^{[i-1]}\epsilon^{i-1}|\), which guarantees the convergence of \(\tilde{N}_t^{[i],\epsilon}\) for all \(t \geq 0\) if \(0 < r < 1\).

\begin{figure}[h]
  \centering
      \includegraphics[width=0.70\linewidth]{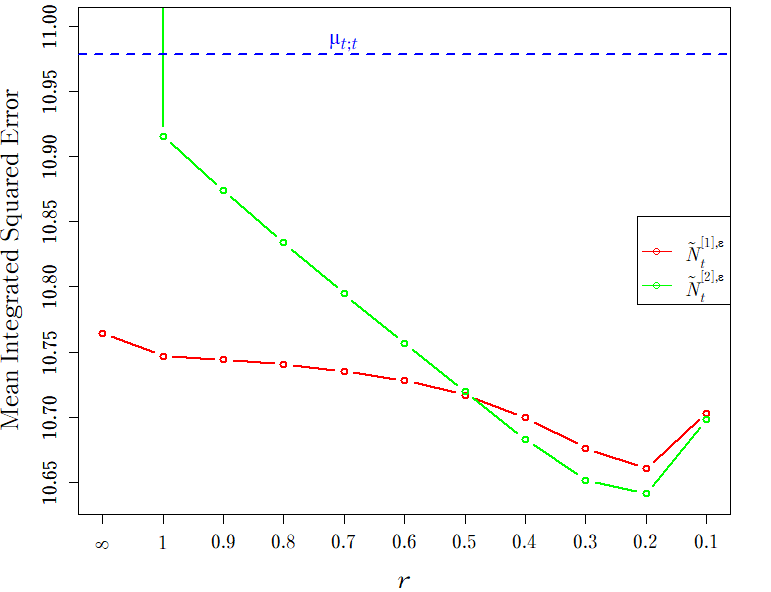}
  \caption{Mean integrated squared errors of the modified filters.}
  \label{fig:MISE-plot}
\end{figure}
\begin{figure}[h!]
  \centering
      \includegraphics[width=0.90\linewidth]{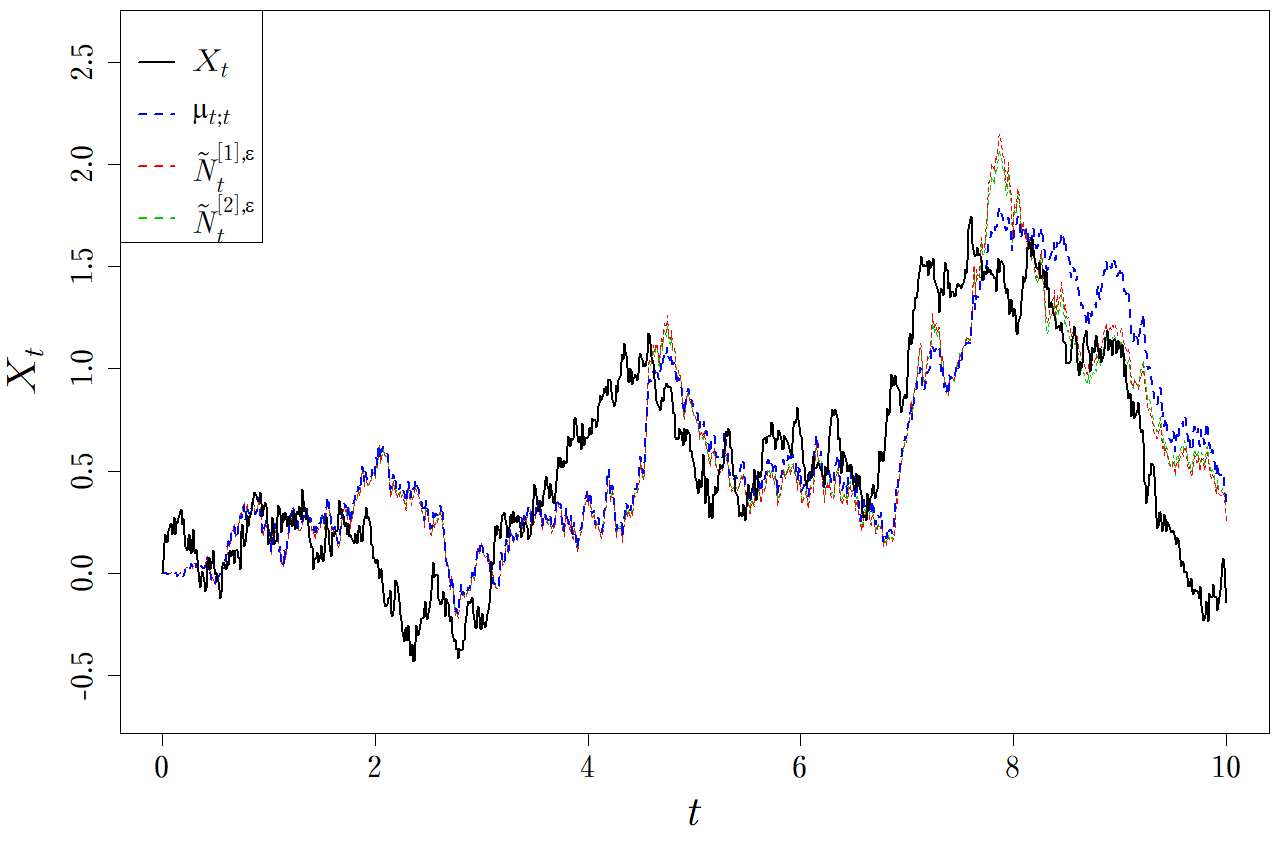}
  \caption{The modified filter of $r=0.2$ for the same simulated path as Figure \ref{fig:filter-plots}.}
  \label{fig:modified-filter-plot}
\end{figure}

Figure \ref{fig:MISE-plot} shows the mean integrated squared error of $\tilde{N}_t^{[1],\epsilon}$ and $\tilde{N}_t^{[2],\epsilon}$ for $r = 0.1, 0.2, \ldots, 1$ and $r = \infty$, based on the 1,000 simulated paths. Note that $\tilde{N}_t^{[i],\epsilon} = N_t^{[i],\epsilon}$ when $r = \infty$. From this result, we observe that the accuracy improves with the modification, and the error decreases as $r$ becomes smaller. The errors of $\tilde{N}_t^{[1],\epsilon}$ and $\tilde{N}_t^{[2],\epsilon}$ both attain their minimum at $r = 0.2$. Figure \ref{fig:modified-filter-plot} shows the modified expansions with $r = 0.2$ for the same path as in Figure \ref{fig:filter-plots}. From this figure, we observe that $\tilde{N}_t^{[1],\epsilon}$ is closer to the true path than $\mu_{t;t}$, and $\tilde{N}_t^{[2],\epsilon}$ slightly improves upon $\tilde{N}_t^{[1],\epsilon}$.

The optimal value of $r$ coincides with $\epsilon = 0.2$ in this case, but they are not necessarily related. In fact, as shown in Figure \ref{fig:plots-for-other-epsilons}, the error is minimized around $r = 0.3$ regardless of the value of $\epsilon$. For each $\epsilon$, we observe that $\tilde{N}_t^{[1],\epsilon}$ and $\tilde{N}_t^{[2],\epsilon}$ can improve upon the linear approximation for sufficiently small $r$, and $\tilde{N}_t^{[2],\epsilon}$ attains a smaller error than $\tilde{N}_t^{[1],\epsilon}$.

\begin{figure}[H]
  \centering
  \begin{subfigure}[b]{0.49\linewidth}
      \centering
      \includegraphics[width=\linewidth]{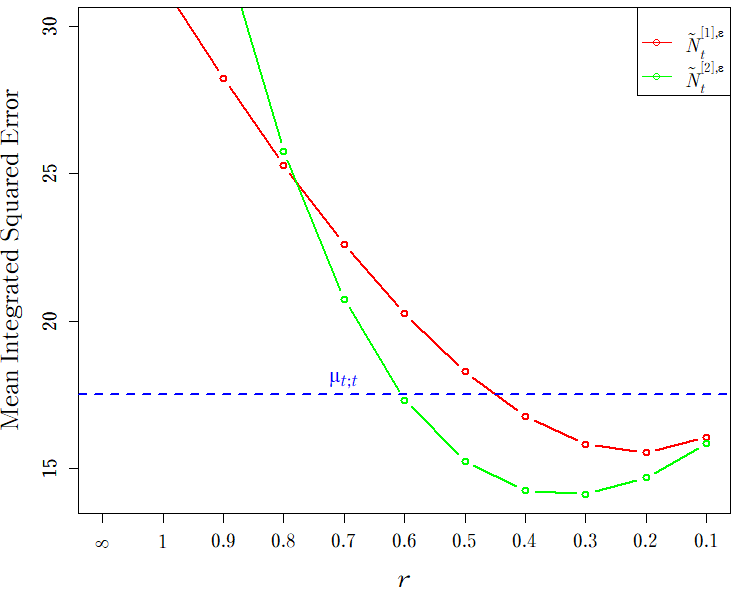}
      \caption{$\epsilon=0.8$}
      \label{fig:doctor-plot5}
  \end{subfigure}
  \hfill
  \begin{subfigure}[b]{0.49\linewidth}
      \centering
      \includegraphics[width=\linewidth]{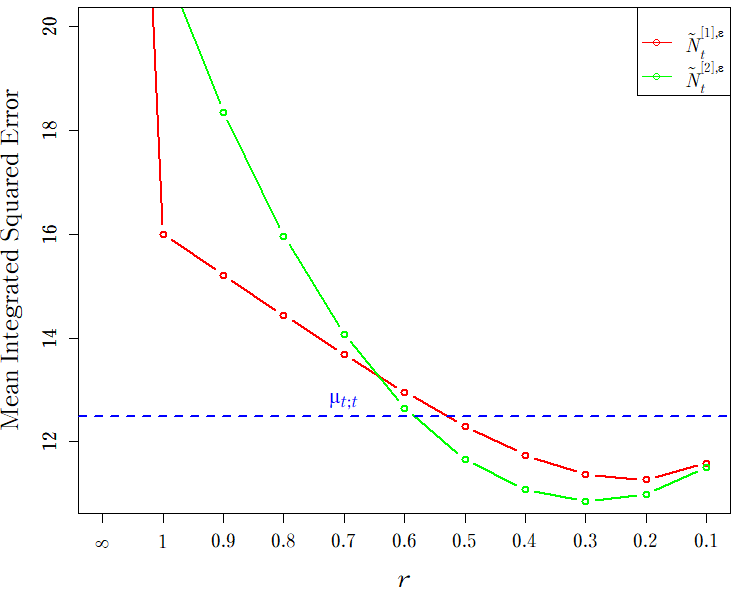}
      \caption{$\epsilon=0.5$}
      \label{fig:doctor-plot6}
  \end{subfigure}

  \vspace{0.5cm} 

  \begin{subfigure}[b]{0.49\linewidth}
      \centering
      \includegraphics[width=\linewidth]{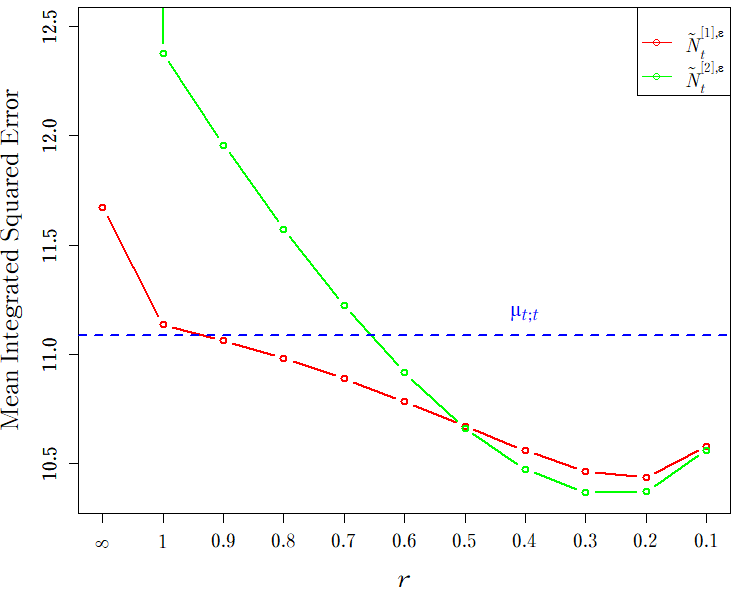}
      \caption{$\epsilon=0.3$}
      \label{fig:doctor-plot7}
  \end{subfigure}
  \hfill
  \begin{subfigure}[b]{0.49\linewidth}
      \centering
      \includegraphics[width=\linewidth]{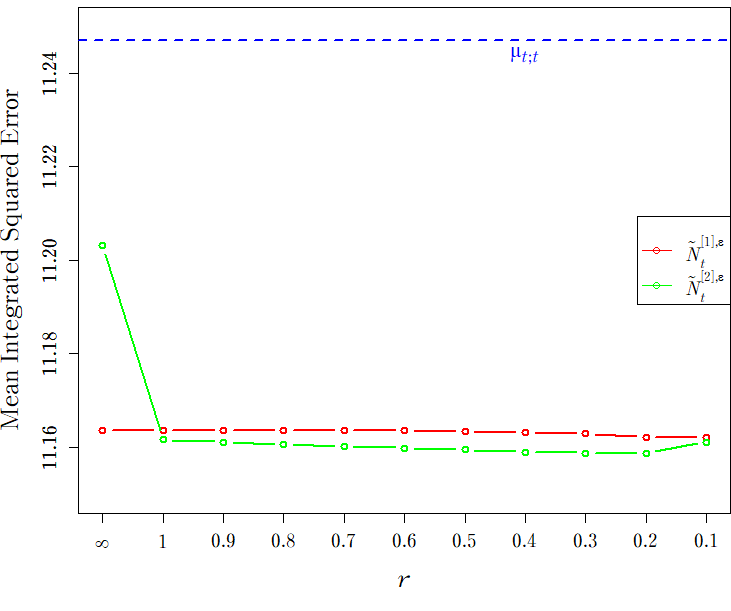}
      \caption{$\epsilon=0.1$}
      \label{fig:doctor-plot8}
  \end{subfigure}

  \caption{Mean integrated squared errors for $\epsilon=0.8,0.5,0.3,0.1$.}
  \label{fig:plots-for-other-epsilons}
\end{figure}
In summary, although the raw asymptotic expansion does not necessarily improve the accuracy of the filter over the entire time horizon due to the divergence of the perturbation term, we can control this divergence by choosing an appropriate parameter $r$ and modifying the coefficients of the expansion. With this modification, we observe an overall improvement in the filter's accuracy. While it is difficult to theoretically determine the optimal value of the parameter $r$, a relatively good value of $r$ can be obtained by conducting simulations in advance. 

\section{Discussion}\label{section-discussion}
In our approach, the computation of the asymptotic expansion is reduced to solving ordinary differential equations, such as (\ref{eq-ode-1})--(\ref{eq-ode-4}). This yields a significant computational advantage over particle filters and PDE-based methods. In particular, it suggests that our method has the potential to be applied in high-dimensional settings, such as data assimilation, where these existing approaches become computationally prohibitive and linearization methods are typically employed.

Although we have focused solely on the expansion of the conditional expectation, the author has also conceived a method for expanding the conditional density. For instance, Figure \ref{fig:density-plot} shows the first order approximation of the conditional density of $X_t$ at $t=5$ for the simulated path of Figure \ref{fig:filter-plots}. While the theoretical justification of this approach is yet to be established, the fact that such multi-modal distributions can be analytically derived from the expansion suggests that this method may offer a powerful tool for representing complex posterior distributions beyond the scope of any existing method.

In this paper, we have provided numerical results only for the perturbed linear model. The simulation for the original nonlinear model has not yet been conducted, primarily due to the complexity of its implementation. Nevertheless, the reformulation part essentially amounts to a Taylor expansion, and the subsequent procedures closely follow those of the perturbed linear case. Moreover, we expect that the modification of the coefficients will not be required in the nonlinear case. This is because the polynomial terms in equation~(\ref{eq-expansion-h}) arise as higher-order terms in the expansion of $h(X_t^\epsilon)$, and their smallness is therefore expected.

\section*{Acknowledgement}
The author is deeply grateful to N. Yoshida for his valuable advice and insightful discussions.
\begin{figure}[H]
  \centering
  \begin{subfigure}[b]{\linewidth}
      \centering
      \includegraphics[width=0.70\linewidth]{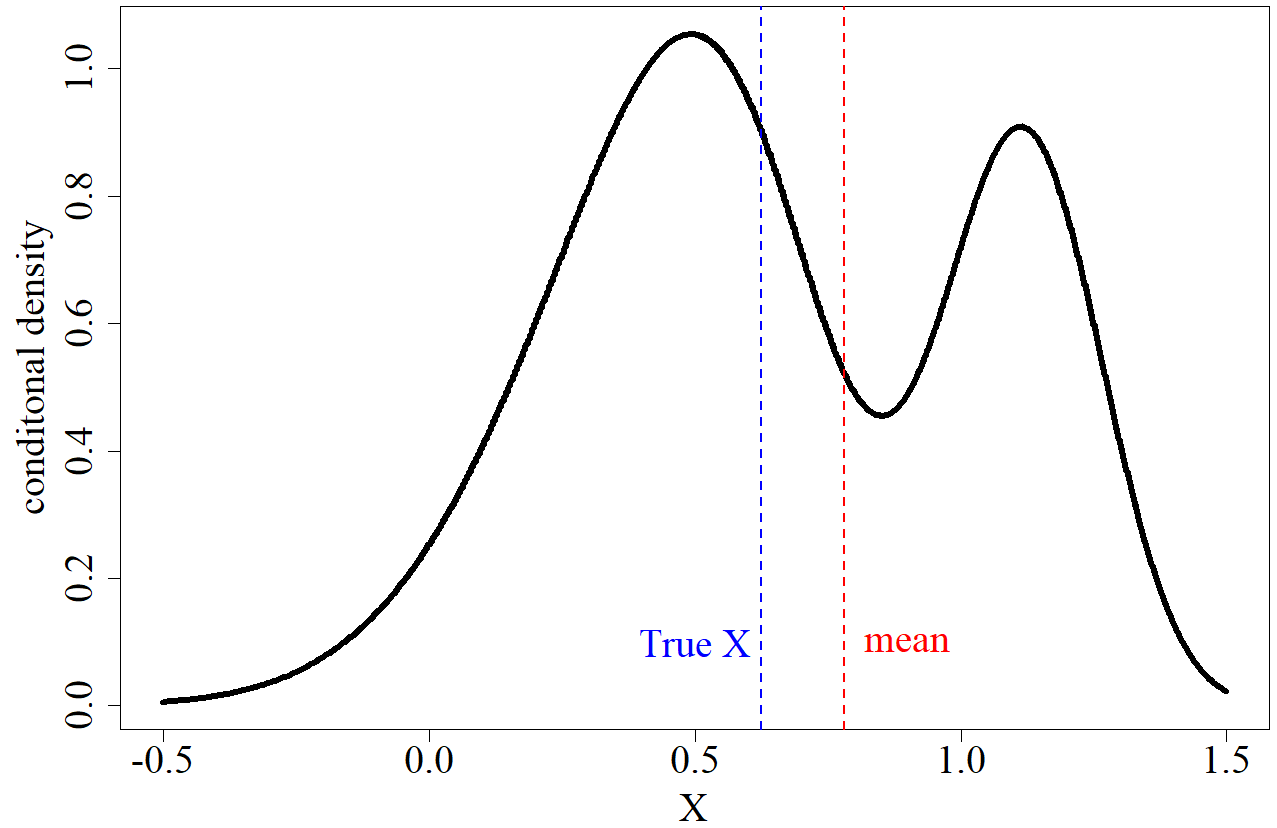}
  \end{subfigure}
  \caption{First order approximation of the conditional density of $X_{5}$.}
  \label{fig:density-plot}
\end{figure}

\bibliographystyle{abbrvnat}

\bibliography{paper-1-arxiv}

\end{document}